\documentclass[a4paper,12pt]{article}
\usepackage{amsmath, bm}
\usepackage{amssymb, amsthm, graphicx}
\usepackage{enumitem}
\usepackage{color}
\usepackage{float}
\usepackage[small]{caption}
\usepackage{subcaption}
\usepackage[mathscr]{euscript}
\usepackage{dsfont}
\usepackage{natbib}
\usepackage{placeins}
\usepackage{rotating}
\usepackage{pdflscape}
\makeatletter
\renewcommand{\eqref}[1]{\tagform@{\ref{#1}}}
\def\maketag@@@#1{\hbox{#1}}
\makeatother
\usepackage{bibentry}
\usepackage[hang,flushmargin]{footmisc} 
\usepackage[left=2.7cm, right=2.7cm, bottom=2.7cm, top=2.7cm]{geometry}
\usepackage{setspace}

\usepackage{hyperref}
\hypersetup{
hidelinks,
}

\newcommand{\reals}{\mathbb{R}}
\newcommand{\integers}{\mathbb{Z}}
\newcommand{\naturals}{\mathbb{N}}

\newcommand{\pr}{\mathbb{P}}        
\newcommand{\ex}{\mathbb{E}}        
\newcommand{\cov}{\textnormal{Cov}} 

\newcommand{\normal}{N}        

\newcommand{\ind}{1} 

\newcommand{\bfbeta}{\bm{\beta}} 
\newcommand{\X}{\boldsymbol{X}} 
\newcommand{\grid}{\mathcal{G}_T} 
\newcommand{\interval}{\mathcal{I}_{u, h}} 

\newcommand{\doublehat}[1]{\skew{5.5}\widehat{\widehat{#1}}}
\newcommand{\doublehattwo}[1]{\widehat{\widehat{#1}}}

\theoremstyle{plain}

\newtheorem{theorem}{Theorem}[section]
\newtheorem{prop}[theorem]{Proposition}

\newtheorem{remark}[theorem]{Remark}

\newtheorem{definitionA}{Definition}[section]

\newtheorem{propA}[definitionA]{Proposition}
\newtheorem{lemmaA}[definitionA]{Lemma}

\makeatletter
\newcommand{\lefteqno}{\let\veqno\@@leqno}
\makeatother

\newcommand{\heading}[2]
{  \setcounter{page}{1}
   \begin{center}

   {\LARGE \textbf{#1}}
   \vspace{0.35cm}
 
   {\LARGE \textbf{#2}}
   \end{center}
}

\newcommand{\authors}[4]
{  \parindent0pt
   \begin{center}
      \begin{minipage}[c][2cm][c]{6cm}
      \begin{center} 
      {\large #1} 
      
      \vspace{0.125cm}
      
      #2 
      \end{center}
      \end{minipage}
      \begin{minipage}[c][2cm][c]{5cm}
      \begin{center} 
      {\large #3}
      
      \vspace{0.125cm}

      #4 
      \end{center}
      \end{minipage}
   \end{center}
}

\begin{document}

\heading{Multiscale Comparison}{of Nonparametric Trend Curves}

\vspace{-0.7cm}

\authors{Marina Khismatullina\renewcommand{\thefootnote}{1}\footnotemark[1]}{Erasmus University Rotterdam}{Michael Vogt\renewcommand{\thefootnote}{2}\footnotemark[2]}{Ulm University} 
\footnotetext[1]{Corresponding author. Address: Erasmus School of Economics, Erasmus University Rotterdam, 3062 PA Rotterdam, Netherlands. Email: \texttt{khismatullina@ese.eur.nl}.}
\renewcommand{\thefootnote}{2}
\footnotetext[2]{Address: Institute of Statistics, Department of Mathematics and Economics, Ulm University, 89081 Ulm, Germany. Email: \texttt{m.vogt@uni-ulm.de}.}
\renewcommand{\thefootnote}{\arabic{footnote}}
\setcounter{footnote}{2}

\vspace{-1.1cm}

\renewcommand{\abstractname}{}
\begin{abstract}
{\noindent We develop new econometric methods for the comparison of nonparametric time trends. In many applications, practitioners are interested in whether the observed time series all have the same time trend. Moreover, they would often like to know which trends are different and in which time intervals they differ. We design a multiscale test to formally approach these questions. Specifically, we develop a test which allows to make rigorous confidence statements about which time trends are different and where (that is, in which time intervals) they differ. Based on our multiscale test, we further develop a clustering algorithm which allows to cluster the observed time series into groups with the same trend. We derive asymptotic theory for our test and clustering methods. The theory is complemented by a simulation study and two applications to GDP growth data and house pricing data.}
\end{abstract}

\vspace{-0.1cm}

\renewcommand{\baselinestretch}{1.2}\normalsize



\vspace{-0.25cm}

\numberwithin{equation}{section}
\allowdisplaybreaks[1]

\section{Introduction}\label{sec:intro}

The comparison of time trends is an important topic in time series analysis and has a wide range of applications in economics, finance and other fields such as climatology. In economics, one may wish to compare trends in real GDP \citep[][]{Grier1989} or trends in long-term interest rates \citep[][]{Christiansen1997} across different countries. In finance, it is of interest to compare the volatility trends of different stocks \citep[][]{Nyblom2000}. In climatology, researchers are interested in comparing the trending behaviour of temperature time series across different spatial locations \citep[][]{KarolyWu2005}.

In this paper, we develop new econometric methods for the comparison of time trends. Classically, time trends are modelled stochastically in econometrics; see e.g.\ \cite{Stock1988}. Recently, however, there has been a growing interest in econometric models with deterministic time trends; see \cite{Cai2007}, \cite{Atak2011}, \cite{Robinson2012}, \cite{ChenGaoLi2012}, \cite{Zhang2012} and \cite{Hidalgo2014} among others. Following this recent development, we consider a general panel model with deterministic time trends. We observe a panel of $n$ time series $\mathcal{T}_i = \{ (Y_{it},\X_{it}): 1 \le t \le T \}$ for $1 \le i \le n$, where $Y_{it}$ are real-valued random variables and $\X_{it} = (X_{it,1},\ldots, X_{it, d})^\top$ are $d$-dimensional random vectors. Each time series $\mathcal{T}_i$ is modelled by the equation
\begin{equation}\label{eq:model}
Y_{it} = m_i \Big( \frac{t}{T} \Big) + \bfbeta_i^\top \X_{it} + \alpha_i + \varepsilon_{it}
\end{equation}
for $1 \le t \le T$, where $\bfbeta_i$ is a $d \times 1$ vector of unknown parameters, $\X_{it}$ is a $d\times 1$ vector of individual covariates or controls, $m_i$ is an unknown nonparametric (deterministic) trend function defined on the rescaled time interval $[0,1]$, $\alpha_i$ is a fixed effect error term and $\mathcal{E}_i = \{ \varepsilon_{it}: 1 \le t \le T \}$ is a zero-mean stationary error process. A detailed description of model \eqref{eq:model} together with the technical assumptions on the various model components can be found in Section \ref{sec:model}.

An important question in many applications is whether the nonparametric time trends $m_i$ in model \eqref{eq:model} are all the same. This question can formally be addressed by a statistical test of the null hypothesis 
\[ H_0: m_1 = m_2 = \ldots = m_n \]
against the alternative $H_1$: $m_i (u) \neq m_j(u)$ for some $u \in [0,1]$ and $i \ne j$. The problem of testing $H_0$ versus $H_1$ is well studied in the literature. In a restricted version of model \eqref{eq:model} without covariates, it has been investigated in \cite{HaerdleMarron1990}, \cite{Hall1990}, \cite{DegrasWu2012} and \cite{ChenWu2018} among many others. Test procedures in a less restricted version of model \eqref{eq:model} with covariates and homogeneous parameter vectors ($\beta_i = \beta$ for all $i$) have been derived, for instance, in \cite{Zhang2012} and \cite{Hidalgo2014}.

Even though many different approaches to test $H_0$ versus $H_1$ have been developed over the years, most of them have a serious shortcoming: they are very uninformative. More specifically, they only allow to test \textit{whether} the trend curves are the same or not. But they do not give any information on \textit{which} curves are different and \textit{where} (that is, in which time intervals) they differ. The following example illustrates the importance of this point: Suppose a group of researchers wants to test whether the GDP trend is the same in ten different countries. If the researchers use a test which is uninformative in the above sense and this test rejects $H_0$, they can only infer that there are at least two countries with a different trend. They do not obtain any information on \textit{which} countries have a different GDP trend and \textit{where} (that is, in which time periods) the trends differ. However, this is exactly the information that is relevant in practice. Without it, it is hardly possible to get a good economic understanding of the situation at hand. The main aim of the present paper is to construct a test procedure which provides the required information. More specifically, we construct a test which allows to make rigorous confidence statements about (i) \textit{which} trend curves are different and (ii) \textit{where} (that is, in which time intervals) they differ.

Very roughly speaking, our approach works as follows: Let $\interval := [u-h,u+h] \subseteq [0,1]$ be the rescaled time interval with midpoint $u \in [0,1]$ and length $2h$. We call $u$ the location and $h$ the scale of the interval $\interval$. For any given location-scale pair $(u,h)$, let
\[ H_0^{[i, j]}(u, h): m_i(w) = m_j(w) \text{ for all } w \in \interval \] 
be the hypothesis that the two trend curves $m_i$ and $m_j$ are identical on the interval $\interval$. Our procedure simultaneously tests $H_0^{[i, j]}(u, h)$ for a wide range of different location-scale pairs $(u,h)$ and all pairs $(i,j)$. As it takes into account multiple scales $h$, it is called a multiscale method. Our main theoretical result shows that under suitable technical conditions, the developed multiscale test controls the familywise error rate, that is, the probability of wrongly rejecting at least one null hypothesis $H_0^{[i, j]}(u, h)$. As we will see, this allows us to make simultaneous confidence statements of the following form for a given significance level $\alpha \in (0,1)$: 
\begin{equation}\label{eq:conf-statement}
\begin{minipage}[c][1.25cm][c]{13cm}
\textit{With probability at least $1-\alpha$, the trends $m_i$ and $m_j$ differ on any interval $\interval$ and for any pair $(i,j)$ for which $H_0^{[i, j]}(u, h)$ is rejected.} 
\end{minipage}
\end{equation}
Hence, as desired, we can make confidence statements about \textit{which} trend curves are different and \textit{where} they differ.

To the best of our knowledge, there are no other test procedures available in the literature which allow to make simultaneous confidence statements of the form \eqref{eq:conf-statement} in the context of the general model \eqref{eq:model}. We are only aware of two other multiscale tests for the comparison of nonparametric time trends, both of which are restricted to a strongly simplified version of model \eqref{eq:model} without covariates. The first one is a SiZer-type test by \cite{Park2009}, for which theory has only been derived in the special case of $n=2$ time series. The second one is a multiscale test by \cite{KhismatullinaVogt2021}, which was developed to detect differences between epidemic time trends in the context of the COVID-19 pandemic. Notably, our multiscale test can be regarded as an extension of the approach in \cite{KhismatullinaVogt2021}. In the present paper, we go beyond the quite limited model setting in \cite{KhismatullinaVogt2021} and develop multiscale inference methods in the general model framework \eqref{eq:model} which allows to deal with a wide range of economic and financial applications. Moreover, we develop a clustering algorithm based on our multiscale test which allows to detect groups of time series with the same trend.

Our multiscale test is constructed step by step in Section \ref{sec:test}, whereas its theoretical properties are laid out in Section \ref{sec:theo}. The clustering algorithm is introduced and investigated in Section \ref{sec:clustering}. The proofs of all theoretical results are relegated to the technical Appendix. We complement the theoretical analysis of the paper by a simulation study and two application examples in Sections \ref{sec:sim} and \ref{sec:app}. In the first application example, we use our test and clustering methods to examine GDP growth data from different OECD contries. We in particular test the hypothesis that there is a common GDP trend in these countries and we cluster the countries into groups with the same trend. In the second example, we examine a long record of housing price data from different countries and compare the price trends in these countries by our methods.

\section{The model framework}\label{sec:model}

\subsection{Notation}

Throughout the paper, we adopt the following notation. For a vector $\mathbf{v} = (v_1, \ldots, v_m)\in\reals^m$, we write $\|\mathbf{v}\|_q = \big(\sum_{i=1}^m v_i^q\big)^{1/q}$ to denote its $\ell_q$-norm and use the shorthand $\|\mathbf{v}\| = \|\mathbf{v}\|_2$ in the special case $q = 2$. For a random variable $V$, we define its $\mathcal{L}^q$-norm by $\|V\|_q = (\ex |V|^q)^{1/q}$ and write $\|V\| := \|V\|_2$ in the case $q = 2$.

Let $\eta_t$ ($t \in \integers$) be independent and identically distributed ($\text{i.i.d.}$) random variables, write $\mathcal{F}_t  = (\ldots, \eta_{t-1}, \eta_t)$ and let $g: \reals^\infty \to \reals$ be a measurable function such that $g(\mathcal{F}_t) = g(\ldots, \eta_{t-1}, \eta_t)$ is a properly defined random variable. Following \cite{Wu2005}, we define the \textit{physical dependence measure} of the process $\{g(\mathcal{F}_t)\}_{t=-\infty}^\infty$ by
\begin{align}\label{eq:physical_dep}
\delta_q(g, t) = \| g(\mathcal{F}_t) - g(\mathcal{F}_t^\prime) \|_q,
\end{align}
where $\mathcal{F}_t^\prime  = (\ldots, \eta_{-1}, \eta^\prime_0, \eta_1, \ldots, \eta_t)$ is a coupled version of $\mathcal{F}_t$ with $\eta_0^\prime$ being an i.i.d.\ copy of $\eta_0$. Evidently, $\delta_q(g, t)$ measures the dependency of the random variable $g(\mathcal{F}_t)$ on the innovation term $\eta_0$. 

\subsection{Model}\label{subsec:model_setting}

We observe a panel of $n$ time series $\mathcal{T}_i = \{(Y_{it}, \X_{it}): 1 \le t \le T \}$ of length $T$ for $1 \le i \le n$. Each time series $\mathcal{T}_i$ satisfies the model equation 
\begin{equation}\label{eq:model_full}
Y_{it} = \bfbeta^\top_i \X_{it} + m_i \Big( \frac{t}{T} \Big) + \alpha_i + \varepsilon_{it} 
\end{equation}
for $1 \le t \le T$, where $\bfbeta_i$ is a $d \times 1$ vector of unknown parameters, $\X_{it} = (X_{it,1},\ldots,X_{it,d})^\top$ is a $d\times 1$ vector of individual covariates, $m_i$ is an unknown nonparametric trend function defined on the unit interval $[0,1]$ with $\int_0^1 m_i(u) du = 0$ for all $i$, $\alpha_i$ is a (deterministic or random) intercept term and $\mathcal{E}_i = \{ \varepsilon_{it}: 1 \le t \le T \}$ is a zero-mean stationary error process. As common in nonparametric regression, the trend functions $m_i$ in model \eqref{eq:model_full} depend on rescaled time $t/T$ rather than on real time $t$; see e.g.\ \cite{Robinson1989}, \cite{Dahlhaus1997} and \cite{VogtLinton2014} for a discussion of rescaled time in nonparametric estimation. The condition $\int_0^1 m_i(u) du = 0$ is required for identification of the trend function $m_i$ in the presence of the intercept $\alpha_i$. We allow $\alpha_i$ to be correlated with the covariates $\X_{it}$ in an arbitrary way. Hence, $\alpha_i$ can be regarded as a fixed effect error term. We do not impose any restrictions on the dependence between the fixed effects $\alpha_i$ across $i$. Similarly, the covariates $\X_{it}$ are allowed to be dependent across $i$ in an arbitrary way. As a consequence, the $n$ time series $\mathcal{T}_i$ in our panel can be correlated with each other in various ways. In contrast to $\alpha_i$ and $\X_{it}$, the error terms $\varepsilon_{it}$ are assumed to be independent across $i$. Technical conditions regarding the model are discussed below. Throughout the paper, we restrict attention to the case where the number of time series $n$ in model \eqref{eq:model_full} is fixed. Extending our theoretical results to the case where $n$ grows with the time series length $T$ is a possible topic for further research.

\subsection{Assumptions}\label{subsec:model_assumptions}

The error processes $\mathcal{E}_i = \{ \varepsilon_{it}: 1 \le t \le T\}$ satisfy the following conditions. 
\begin{enumerate}[label=(C\arabic*),leftmargin=1.05cm]
\item \label{C-err1} 
The variables $\varepsilon_{it}$ allow for the representation $\varepsilon_{it} = g_i(\mathcal{F}_{it})$, where $\mathcal{F}_{it} = (\ldots,\eta_{it-2}, \linebreak \eta_{it-1},\eta_{it})$, the variables $\eta_{it}$ are i.i.d.\ across $t$, and $g_i: \reals^\infty \rightarrow \reals$ is a measurable function such that $\varepsilon_{it}$ is well-defined. It holds that $\ex[\varepsilon_{it}] = 0$ and $\| \varepsilon_{it} \|_q \le C < \infty$ for some $q > 4$ and a sufficiently large constant $C$. 
\item \label{C-err2} The processes $\mathcal{E}_i = \{ \varepsilon_{it}: 1 \le t \le T\}$ are independent across $i$.
\end{enumerate}
Assumption \ref{C-err1} implies that the error processes $\mathcal{E}_i$ are stationary and causal (in the sense that $\varepsilon_{it}$ does not depend on future innovations $\eta_{is}$ with $s>t$). The class of error processes that satisfy condition \ref{C-err1} is very large. It includes linear processes, nonlinear transformations thereof, as well as a large variety of nonlinear processes such as Markov chain models and nonlinear autoregressive models \citep[][]{Wu2016}. Following \cite{Wu2005}, we impose conditions on the dependence structure of the error processes $\mathcal{E}_i$ in terms of the physical dependence measure $\delta_q(g_i, t)$ defined in \eqref{eq:physical_dep}. In particular, we assume the following: 
\begin{enumerate}[label=(C\arabic*),leftmargin=1.05cm]
\setcounter{enumi}{2}
\item \label{C-err3} For each $i$, $\sum\nolimits_{s \ge 0} \delta_q(g_i, s)$ is finite and $\sum\nolimits_{s \ge t} \delta_q(g_i, s) = O ( t^{-\gamma} (\log t)^{-A})$ with $q$ from \ref{C-err1}, where $A > \frac{2}{3} (1/q + 1 + \gamma)$ and $\gamma = \{q^2 - 4 + (q-2) \sqrt{q^2 + 20q + 4}\} / 8q$. 
\end{enumerate}
For fixed $i$ and $t$, the expression $\sum\nolimits_{s \ge t} \delta_q(g_i, s)$ measures the cumulative effect of the innovation $\eta_0$ on the variables $\varepsilon_{it}, \varepsilon_{it+1},\ldots$ in terms of the $\mathcal{L}^q$-norm. Condition \ref{C-err3} puts some restrictions on the decay of $\sum\nolimits_{s \ge t} \delta_q(g_i, s)$ (as a function of $t$). It is fulfilled by a wide range of stationary processes $\mathcal{E}_i$. For a detailed discussion of \ref{C-err1}--\ref{C-err3} and some examples of error processes that satisfy these conditions, see \cite{KhismatullinaVogt2020}.

\pagebreak
The covariates $\X_{it} = (X_{it,1},\ldots, X_{it, d})^\top$ are assumed to have the following properties. 
\begin{enumerate}[label=(C\arabic*),leftmargin=1.05cm]
\setcounter{enumi}{3}
\item \label{C-reg1} The variables $X_{it, j}$ allow for the representation $X_{it, j} = h_{ij}(\mathcal{G}_{it, j})$, where $\mathcal{G}_{it, j} = (\ldots, \xi_{it-1,j}, \xi_{it, j})$, the random variables $\xi_{it, j}$ are i.i.d.\ across $t$ and $h_{ij}: \reals^\infty \rightarrow \reals$ is a measurable function such that $X_{it, j}$ is well-defined. We use the notation $\X_{it} = \boldsymbol{h}_{i}(\mathcal{G}_{it})$ with $\boldsymbol{h}_i = (h_{i1}, \ldots, h_{id})^\top$ and $\mathcal{G}_{it} = (\mathcal{G}_{it,1}, \ldots, \mathcal{G}_{it, d})^\top$. It holds that $\ex [X_{it, j}]=0$ and $\| X_{it, j} \|_{q^\prime} <\infty$ for all $i$ and $j$, where $q^\prime > \max \{ 4, \theta q \}$ with $q$ from \ref{C-err1} and $\theta$ specified in \ref{C-grid} below.
\item \label{C-reg2} The matrix $\ex[\Delta \X_{it} \Delta \X_{it}^\top]$ is invertible for each $i$.
\item \label{C-reg3} For each $i$ and $j$, it holds that $\sum_{s \ge 0} \delta_{q^\prime}(h_{ij}, s)$ is finite and $\sum_{s \ge t} \delta_{q^\prime}(h_{ij}, s)= O(t^{-\alpha})$ for some $\alpha > 1/2 - 1/{q^\prime}$ with $q^\prime$ from \ref{C-reg1}.
\end{enumerate}
Assumption \ref{C-reg1} guarantees that the process $\{ \X_{it}: 1 \le t \le T \}$ is stationary and causal for each $i$. Assumption \ref{C-reg3} restricts the serial dependence of the process $\{ \X_{it}: 1 \le t \le T \}$ for each $i$ in terms of the physical dependence measure.

We finally impose some assumptions on the relationship between the covariates and the errors and on the trend functions $m_i$.
\begin{enumerate}[label=(C\arabic*),leftmargin=1.05cm]
\setcounter{enumi}{6}
\item \label{C-reg-err} The random variables $\Delta \X_{it} = \X_{it} - \X_{it-1}$ and $\Delta \varepsilon_{it} = \varepsilon_{it} - \varepsilon_{it-1}$ are uncorrelated, that is, $\cov(\Delta \X_{it}, \Delta \varepsilon_{it}) = \ex[\Delta \X_{it} \Delta \varepsilon_{it}] = 0$. 
\item \label{C-trend} The trend functions $m_i$ are Lipschitz continuous on $[0,1]$, that is, $|m_i(v) - m_i(w)| \le L |v-w|$ for all $v,w \in [0,1]$ and some constant $L < \infty$. Moreover, they are normalised such that $\int_0^1m_i (u)du = 0$ for each $i$.
\end{enumerate}

\begin{remark}
Conditions \ref{C-reg1}--\ref{C-reg3} can be relaxed to cover locally stationary regressors. For example, \ref{C-reg1} may be replaced by
\begin{enumerate}[label=(C\arabic*$^\prime$),leftmargin=1.1cm]
\setcounter{enumi}{3}
\item \label{C-reg1-star} The variables $X_{it, j}$ allow for the representation $X_{it, j} = h_{ij}(t; \mathcal{G}_{it, j})$, where $\mathcal{G}_{it, j} = (\ldots, \xi_{it-1,j}, \xi_{it, j})$, the random variables $\xi_{it, j}$ are i.i.d.\ across $t$ and $h_{ij}(t;\cdot): \reals^\infty \rightarrow \reals$ is a measurable function  for each $t$ such that $X_{it, j}$ is well-defined. We use the notation $\X_{it} = \boldsymbol{h}_{i}(t; \mathcal{G}_{it})$ with $\boldsymbol{h}_i = (h_{i1}, \ldots, h_{id})^\top$ and $\mathcal{G}_{it} = (\mathcal{G}_{it,1}, \ldots, \mathcal{G}_{it, d})^\top$. It holds that $\ex [X_{it, j}]=0$ and $\| X_{it, j} \|_{q^\prime} <\infty$ for all $i$, $j$ and $t$, where $q^\prime > \max \{ 4, \theta q \}$ with $q$ from \ref{C-err1} and $\theta$ specified in \ref{C-grid} below.
\end{enumerate}
The other assumptions can be adjusted accordingly. We conjecture that our main theoretical results still hold in this case. However, the complexity of the technical arguments will increase drastically. Hence, for the sake of clarity, we restrict attention to stationary covariates $\X_{it}$. 
\end{remark}

\section{The multiscale test}\label{sec:test}

In this section, we develop a multiscale test for the comparison of the trend curves $m_i$ in model \eqref{eq:model_full}. More specifically, we construct a multiscale test of the null hypothesis 
\[ H_0: m_1 = m_2 = \ldots = m_n. \]
The test is designed to be as informative as possible: It does not only allow to say whether $H_0$ is violated. It also gives information on which types of violation occur. In particular, it allows to infer \textit{which} trends are different and \textit{where} (that is, in which time intervals) they differ.

Our strategy to construct the test is as follows: For any given location-scale point $(u,h)$, let
\[ H_0^{[i, j]}(u, h): m_i(w) = m_j(w) \text{ for all } w \in \interval \] 
be the hypothesis that $m_i$ and $m_j$ are identical on the interval $\interval := [u-h,u+h] \subseteq [0,1]$. $H_0^{[i, j]}(u, h)$ can be viewed as a local null hypothesis that characterises the behaviour of two trend functions locally on the interval $\interval$, whereas $H_0$ is the global null hypothesis that is concerned with the comparison of all trends on the whole unit interval $[0, 1]$. We consider a large family of time intervals $\interval$ that fully cover the unit interval. More formally, we consider all intervals $\interval$ with $(u,h) \in \grid$, where $\grid$ is a set of location-scale points $(u,h)$ with the property that $\bigcup_{(u, h) \in \grid} \interval = [0,1]$. Technical conditions on the set $\grid$ are given below. The global null $H_0$ can now be formulated as 
\begin{align*}
H_0: \ & \text{The hypothesis } H_0^{[i, j]}(u, h) \text{ holds true for all intervals }  \interval \text{ with } (u, h) \in \grid \\[-0.075cm] & \text{and for all } 1 \le i < j \le n. 
\end{align*} 
This formulation shows that we can test $H_0$ by a procedure which simultaneously tests the local hypothesis $H_0^{[i, j]}(u, h)$ for all $1 \le i < j \le n$ and all $(u,h) \in \grid$.

In what follows, we design such a simultaneous test procedure step by step. To start with, we introduce some auxiliary estimators in Section \ref{subsec:test:prep}. We then construct the test statistics in Section \ref{subsec:test:stat} and set up the test procedure in Section \ref{subsec:test:test}. Finally, Section \ref{subsec:test:impl} explains how to implement the procedure in practice.

\subsection{Preliminary steps}\label{subsec:test:prep}

If the fixed effects $\alpha_i$ and the coefficient vectors $\bfbeta_i$ were known, our testing problem would be greatly simplified. In particular, we could consider the model
\begin{align*}
Y_{it}^\circ = m_i \Big( \frac{t}{T} \Big) + \varepsilon_{it} \qquad \text{with} \qquad Y_{it}^\circ := Y_{it} - \alpha_i - \bfbeta_i^\top \X_{it}, 
\end{align*}
which is a standard nonparametric regression equation. As $\alpha_i$ and $\bfbeta_i$ are not observed in practice, we construct estimators $\widehat{\alpha}_i$ and $\widehat{\bfbeta}_i$ and replace the unknown variables $Y_{it}^\circ$ by the approximations 
\begin{align*}
\widehat{Y}_{it} := Y_{it} -\widehat{\alpha}_i - \widehat{\bfbeta}_i^\top \X_{it}. 
\end{align*}
In what follows, we explain how to construct estimators of $\alpha_i$ and $\bfbeta_i$.

In order to estimate the parameter vector $\bfbeta_i$ for a given $i$, we consider the time series $\Delta \mathcal{T}_i = \{(\Delta Y_{it}, \Delta \X_{it}): 2 \leq t \leq T\}$ of the first differences $\Delta Y_{it} = Y_{it} - Y_{i t-1}$ and $\Delta  \X_{it} =  \X_{it} - \X_{it-1}$. We can write
\begin{align*}
\Delta Y_{it} = Y_{it} - Y_{i t-1} =\bfbeta_i^\top \Delta \X_{it} + \bigg(m_i \Big( \frac{t}{T} \Big) - m_i \Big(\frac{t-1}{T}\Big)\bigg) + \Delta \varepsilon_{it},
\end{align*}
where $ \Delta \varepsilon_{it} = \varepsilon_{it} - \varepsilon_{i t-1}$. Since $m_i$ is Lipschitz by Assumption \ref{C-trend}, we can use the fact that $ |m_i ( \frac{t}{T} ) - m_i (\frac{t-1}{T}) | = O(\frac{1}{T})$ to get that 
\begin{align}\label{model_with_regs}
\Delta Y_{it} = \bfbeta_i^\top \Delta \X_{it} + \Delta \varepsilon_{it} + O\Big(\frac{1}{T}\Big).
\end{align}
This suggests to estimate $\bfbeta_i$ by applying least squares methods to equation \eqref{model_with_regs}, treating $\Delta \X_{it}$ as the regressors and $\Delta Y_{it}$ as the response variable. As a result, we obtain the least squares estimator
\begin{align}\label{eq:beta:est}
\widehat{\bfbeta}_i = \Big( \sum_{t=2}^T \Delta \X_{it} \Delta \X_{it}^\top \Big)^{-1} \sum_{t=2}^T \Delta \X_{it} \Delta Y_{it}.
\end{align}
In Lemma \ref{lemma-beta-rate} in the Appendix, we show that $\widehat{\bfbeta}_i - \bfbeta_i = O_p(T^{-1/2})$ under our assumptions. 
Given $\widehat{\bfbeta}_i$, we next estimate the fixed effect term $\alpha_i$ by 
\begin{align}\label{eq:alpha:est}
\widehat{\alpha}_i &= \frac{1}{T}\sum_{t=1}^T \big(Y_{it} - \widehat{\bfbeta}_i^\top \X_{it}\big). 
\end{align}
This is a reasonable estimator of $\alpha_i$ for the following reason: For each $i$, it holds that
\begin{align*}
\widehat{\alpha}_i - \alpha_i = \big(\bfbeta_i - \widehat{\bfbeta}_i \big)^\top\frac{1}{T}\sum_{t=1}^T  \X_{it} + \frac{1}{T}\sum_{t=1}^T m_i\Big(\frac{t}{T}\Big) + \frac{1}{T}\sum_{t=1}^T \varepsilon_{it} = O_p(T^{-1/2}), 
\end{align*}
since $\frac{1}{T}\sum_{i=1}^T \varepsilon_{it} = O_p(T^{-1/2})$ by the law of large numbers, $\frac{1}{T}\sum_{i=1}^T m_i(t/T) = O(T^{-1})$ due to Lipschitz continuity of $m_i$ and the normalisation $\int_{0}^1 m_i(u)du = 0$, $\frac{1}{T}\sum_{t=1}^T  \X_{it} = O_p(1)$ by Chebyshev's inequality and $\widehat{\bfbeta}_i - \bfbeta_i = O_p (T^{-1/2})$.

In order to construct our multiscale test, we do not only require the variables $\widehat{Y}_{it}$ and thus the estimators $\widehat{\bfbeta}_i$ and $\widehat{\alpha}_i$. We also need an estimator of the long-run error variance $\sigma_i^2 = \sum\nolimits_{\ell=-\infty}^{\infty} \cov(\varepsilon_{i0}, \varepsilon_{i\ell})$ for each $i$. Throughout the paper, we assume that the long-run variance $\sigma_i^2$ does not depend on $i$, that is $\sigma_i^2 = \sigma^2$ for all $i$. For technical reasons, we nevertheless use a different estimator of $\sigma_i^2 = \sigma^2$ for each $i$. Let $\widehat{\sigma}_i^2$ be an estimator of $\sigma_i^2$ which is computed from the $i$-th time series $\{ \widehat{Y}_{it}: 1 \le t \le T \}$, that is, let $\widehat{\sigma}_i^2 = \widehat{\sigma}_i^2(\widehat{Y}_{i1},\ldots,\widehat{Y}_{iT})$ be a function of the variables $\widehat{Y}_{it}$ for $1 \le t \le T$. Our theory works with any estimator $\widehat{\sigma}_i^2$ which has the property that $\widehat{\sigma}_i^2 = \sigma_i^2 + o_p(\rho_T)$ for each $i$, where $\rho_T$ slowly converges to zero, in particular, $\rho_T = o(1/\log T)$.
We now briefly discuss two possible choices of $\widehat{\sigma}_i^2$.

Following \cite{Kim2016}, we can estimate $\sigma_i^2$ by a variant of the subseries variance estimator, which was first proposed by \cite{Carlstein1986} and then extended by \cite{WuZhao2007}. Formally, we set
\begin{align}
\widehat{\sigma}_i^2 = \frac{1}{2(M-1)s_T}\sum_{m=1}^M \bigg[ \sum_{t = 1}^{s_T} & \Big(Y_{i(t + ms_T)} - Y_{i(t + (m-1)s_T)} \nonumber \\[-0.2cm] & - \widehat{\bfbeta}_i^\top\big(\X_{i(t+ms_T)} - \X_{i(t+(m-1)s_T)}\big) \Big)\bigg]^2, \label{eq:lrv}
\end{align}
where $s_T$ is the length of the subseries and $M = \lfloor T/s_T\rfloor$ is the largest integer not exceeding $T/s_T$. As per the optimality result in \cite{Carlstein1986}, we set $s_T \asymp T^{1/3}$.  When implementing $\widehat{\sigma}_i^2$, we in particular choose $s_T = \lfloor T^{1/3}\rfloor$. According to  Lemma \ref{lemmaA:lrv} in the Appendix, $\widehat{\sigma}_i^2$ is an asymptotically consistent estimator of $\sigma_i^2$ with the rate of convergence $O_p(T^{-1/3})$.

The subseries estimator $\widehat{\sigma}_i^2$ is completely nonparametric: it does not presuppose a specific time series model for the error process $\mathcal{E}_i = \{ \varepsilon_{it}: 1 \le t \le T\}$ but merely requires $\mathcal{E}_i$ to fulfill general weak dependence conditions. In practice, however, it often makes sense to impose additional structure on the error process $\mathcal{E}_i$. A very popular yet general error model is the class of $\text{AR}(\infty)$ processes. A difference-based estimator of the long-run error variance for this error model has been developed in \cite{KhismatullinaVogt2020} and can be easily adapted to the situation at hand. A detailed discussion of the estimator and a comparison with other long-run variance estimators can be found in Section 4 of \cite{KhismatullinaVogt2020}.

\subsection{Construction of the test statistics}\label{subsec:test:stat}

A statistic to test the local null hypothesis $H_0^{[i, j]}(u, h)$ for a given location-scale point $(u,h)$ and a given pair of time series $(i,j)$ can be constructed as follows: Let
\begin{align}\label{eq:psi_hat_ij}
\widehat{\psi}_{ij, T}(u, h) = \sum\limits_{t=1}^T w_{t,T}(u, h)(\widehat{Y}_{it} - \widehat{Y}_{jt})
\end{align}
be a weighted average of the differenced variables $\widehat{Y}_{it} - \widehat{Y}_{jt}$. In particular, let $w_{t,T}(u, h)$ be local linear kernel weights of the form
\begin{equation}\label{eq:weights}
w_{t,T}(u, h) = \frac{\Lambda_{t,T}(u, h)}{ \{\sum\nolimits_{t=1}^T \Lambda_{t,T}(u, h)^2 \}^{1/2} }, 
\end{equation}
where
\[ \Lambda_{t,T}(u, h) = K\Big(\frac{\frac{t}{T}-u}{h}\Big) \Big[ S_{T,2}(u, h) - \Big(\frac{\frac{t}{T}-u}{h}\Big) S_{T,1}(u, h) \Big], \]
$S_{T,\ell}(u, h) = (Th)^{-1} \sum\nolimits_{t=1}^T K(\frac{\frac{t}{T}-u}{h}) (\frac{\frac{t}{T}-u}{h})^\ell$ for $\ell = 1,2$ and $K$ is a kernel function with the following property: 
\begin{enumerate}[label=(C\arabic*),leftmargin=1.05cm]
\setcounter{enumi}{8}
\item \label{C-ker} The kernel $K$ is non-negative, symmetric about zero and integrates to one. Moreover, it has compact support $[-1,1]$ and is Lipschitz continuous, that is,\linebreak $|K(v) - K(w)| \le C_K |v-w|$ for any $v, w \in \reals$ and some constant $C_K > 0$.
\end{enumerate} 
The kernel average $\widehat{\psi}_{ij, T}(u, h)$ can be regarded as a measure of distance between the two trend curves $m_i$ and $m_j$ on the interval $\interval = [u-h, u+h]$. We may thus use a normalized (absolute) version of the kernel average $\widehat{\psi}_{ij, T}(u, h)$, in particular, the term 
\begin{equation}\label{eq:local-stat-without-correction}
\Big| \frac{\widehat{\psi}_{ij, T}(u, h)}{(\widehat{\sigma}_i^2 + \widehat{\sigma}_j^2)^{1/2}} \Big|
\end{equation}
as a statistic to test the local null $H_0^{[i, j]}(u, h)$. Our aim is to test $H_0^{[i, j]}(u, h)$ simultaneously for a wide range of location-scale points $(u,h)$ and all possible pairs of time series $(i,j)$. To take into account that we are faced with a simultaneous test problem, we replace the statistics from \eqref{eq:local-stat-without-correction} by additively corrected versions of the form 
\begin{align}\label{eq:psi_zero_ij}
\widehat{\psi}^0_{ij, T}(u, h) =  \Big|\frac{\widehat{\psi}_{ij, T}(u, h)}{(\widehat{\sigma}_i^2 + \widehat{\sigma}_j^2)^{1/2}}\Big| - \lambda(h),
\end{align}
where $\lambda(h) = \sqrt{2 \log \{ 1/(2h) \}}$. The scale-dependent correction term $\lambda(h)$ was first introduced in the multiscale approach of \cite{DuembgenSpokoiny2001} and has been used in various contexts since then. We refer to \citet{KhismatullinaVogt2020} for a detailed discussion of the idea behind this additive correction.

In order to test the global null hypothesis $H_0$, we aggregate the test statistics $\widehat{\psi}^0_{ij, T}(u, h)$ by taking their maximum over all location-scale points $(u,h) \in \grid$ and all pairs of time series $(i,j)$ with $1 \le i < j \le n$. This leads to the multiscale test statistic 
\[ \widehat{\Psi}_{n,T} = \max_{(u,h) \in \grid, 1 \le i < j \le n} \widehat{\psi}^0_{ij, T}(u, h). \]
For our theoretical results, we suppose that the set of location-scale points $\grid$ is a subset of $\grid^{\text{full}} = \{ (u, h): \interval = [u-h,u+h] \subseteq [0,1]$ with $u = t/T$ and $h = s/T$  for some $1 \le t, s \le T$ and $h \in [h_{\min},h_{\max}] \}$ which fulfills the following conditions:
\begin{enumerate}[label=(C\arabic*),leftmargin=1.2cm]
\setcounter{enumi}{9}

\item \label{C-grid} $|\mathcal{G}_T| = O(T^\theta)$ for some arbitrarily large but fixed constant $\theta > 0$, where $|\mathcal{G}_T|$ denotes the cardinality of $\mathcal{G}_T$. 

\item \label{C-h} $h_{\min} \gg T^{-(1-\frac{2}{q})} \log T$, that is, $h_{\min} / \{ T^{-(1-\frac{2}{q})} \log T \} \rightarrow \infty$ with $q > 4$ defined in \ref{C-err1} and $h_{\max} = o(1)$.

\end{enumerate}
Assumptions \ref{C-grid} and \ref{C-h} place relatively mild restrictions on the set $\grid$: \ref{C-grid} allows $\grid$ to be very large compared to the sample size $T$. In particular, $\grid$ may grow as any polynomial of $T$. \ref{C-h} allows $\grid$ to contain intervals $[u-h,u+h]$ of many different scales $h$, ranging from very small intervals of scale $h_{\min}$ to substantially larger intervals of scale $h_{\max}$.

\subsection{The test procedure}\label{subsec:test:test}

Let $\mathcal{M} = \{ (u,h,i,j): (u,h) \in \grid \text{ and } 1 \le i < j \le n \}$ be the collection of all location-scale points $(u,h)$ and all pairs of time series $(i,j)$ under consideration. Moreover, let $\mathcal{M}_0 \subseteq \mathcal{M}$ be the set of tuples $(u,h,i,j)$ for which $H_0^{[i, j]}(u, h)$ is true.
For a given significance level $\alpha \in (0,1)$, our multiscale test is carried out as follows: 
\begin{enumerate}[label=(\roman*),leftmargin=0.75cm]

\item For each tuple $(u,h,i,j) \in \mathcal{M}$, we reject the local null hypothesis $H_0^{[i, j]}(u, h)$ if 
\[ \widehat{\psi}^0_{ij, T}(u, h) > q_{n,T}(\alpha), \]
where the critical value $q_{n,T}(\alpha)$ is constructed below such that the familywise error rate (FWER) is controlled at level $\alpha$. As usual, the FWER is defined as the probability of wrongly rejecting $H_0^{[i, j]}(u, h)$ for at least one tuple $(u,h,i,j) \in \mathcal{M}$. More formally, it is defined as 
\[ \text{FWER}(\alpha) = \pr \Big( \exists (u,h,i,j) \in \mathcal{M}_0 : \widehat{\psi}^0_{ij, T}(u, h) > q_{n,T}(\alpha) \Big) \]
for a given significance level $\alpha \in (0,1)$ and we say that the FWER is controlled at level $\alpha$ if $\text{FWER}(\alpha) \le \alpha$.

\item We reject the global null hypothesis $H_0: m_1 = m_2 = \ldots = m_n$ if at least one local null hypothesis is rejected. Put differently, we reject $H_0$ if 
\[ \widehat{\Psi}_{n,T} > q_{n,T}(\alpha). \]

\end{enumerate}
We now construct a critical value $q_{n,T}(\alpha)$ which controls the FWER at level $\alpha$. Let $q_{n,T}(\alpha)$ be the $(1-\alpha)$-quantile of the multiscale statistic $\widehat{\Psi}_{n,T}$ under the global null $H_0: m_1 = m_2 = \ldots = m_n$. Since
\begin{align*} 
\text{FWER}(\alpha)
 & = 1 - \pr \Big( \forall (u,h,i,j) \in \mathcal{M}_0 : \widehat{\psi}^0_{ij, T}(u, h) \le q_{n,T}(\alpha) \Big) \\
 & = 1 - \pr \Big( \max_{ (u,h,i,j) \in \mathcal{M}_0} \widehat{\psi}^0_{ij, T}(u, h) \le q_{n,T}(\alpha) \Big) \\
 & \le 1 - \pr_{H_0} \Big( \max_{ (u,h,i,j) \in \mathcal{M}} \widehat{\psi}^0_{ij, T}(u, h) \le q_{n,T}(\alpha) \Big) \le \alpha 
\end{align*}
with $\pr_{H_0}$ denoting the probability under $H_0$, this choice of $q_{n,T}(\alpha)$ indeed controls the FWER at the desired level. However, this choice is not computable in practice. We thus replace it by a suitable approximation: Suppose we could perfectly estimate $\bfbeta_i$ and $\sigma_i^2$, that is, $\widehat{\bfbeta}_i = \bfbeta_i$ and $\widehat{\sigma}_i^2 = \sigma_i^2$ for all $i$. Then under $H_0$, the test statistics $\widehat{\psi}^0_{ij, T}(u, h)$ would simplify to 
\begin{equation}\label{eq:stat-idealized}
\widehat{\psi}^0_{ij, T}(u, h) = \Big| \frac{\sum\nolimits_{t=1}^T w_{t,T}(u, h) \{ (\varepsilon_{it} - \bar{\varepsilon}_i) - (\varepsilon_{jt} - \bar{\varepsilon}_j) \}}{(\sigma_i^2 + \sigma_j^2)^{1/2}} \Big| - \lambda(h), 
\end{equation}
where $\bar{\varepsilon}_i = \bar{\varepsilon}_{i,T} := T^{-1} \sum_{t=1}^T \varepsilon_{it}$ denotes the empirical average of the variables $\varepsilon_{i1},\ldots,\varepsilon_{iT}$. Replacing the error terms $\varepsilon_{it}$ in \eqref{eq:stat-idealized} by normally distributed variables leads to the statistics
\begin{align}\label{eq:phi_zero_ij}
\phi^0_{ij, T}(u, h) =  \bigg|\frac{\phi_{ij, T}(u, h)}{(\sigma_i^2 + \sigma_j^2)^{1/2}}\bigg| - \lambda(h)
\end{align}
with 
\begin{align}\label{eq:phi_ij}
\phi_{ij, T}(u, h) = \sum\limits_{t=1}^T w_{t,T}(u, h) \, \big\{ \sigma_i (Z_{it} - \bar{Z}_i) - \sigma_j (Z_{jt} - \bar{Z}_j) \big\},
\end{align}
where $Z_{it}$ are independent standard normal random variables for $1 \le t \le T$ and $1 \le i \le n$
and, as before, we use the shorthand $\bar{Z}_i = \bar{Z}_{i,T} := T^{-1} \sum_{t=1}^T Z_{it}$. With this notation, we define 
\begin{align}\label{eq:Phi}
\Phi_{n,T} = \max_{1 \le i < j \le n}\max_{(u, h) \in \mathcal{G}_T} \phi^0_{ij, T}(u, h),
\end{align}
which can be regarded as a Gaussian version of the test statistic $\widehat{\Psi}_{n,T}$ under the null $H_0$ (in the idealized case with $\widehat{\bfbeta}_i = \bfbeta_i$ and $\widehat{\sigma}_i^2 = \sigma_i^2$ for all $i$). We now choose $q_{n,T}(\alpha)$ to be the $(1-\alpha)$-quantile of $\Phi_{n,T}$.

\begin{remark}
Importantly, this choice of $q_{n,T}(\alpha)$ can be computed by Monte Carlo simulations in practice: under our assumption that the long-run variance $\sigma_i^2$ does not depend on $i$ (i.e. $\sigma_i^2 = \sigma^2_j = \sigma^2$), we can rewrite the statistics from \eqref{eq:phi_zero_ij} as
\[\phi^0_{ij, T}(u, h) = \frac{1}{\sqrt{2}} \Big|\sum\limits_{t=1}^T w_{t,T}(u, h) \, \big\{ (Z_{it} - \bar{Z}_i) - (Z_{jt} - \bar{Z}_j) \big\}\Big| - \lambda(h). \] 
This shows that the distribution of these statistics and thus the distribution of the multiscale statistic $\Phi_{n,T}$ only depend on the Gaussian variables $Z_{it}$ for $1 \le i \le n$ and $1 \le t \le T$. Consequently, we can approximate the distribution of $\Phi_{n,T}$ -- and in particular its quantiles $q_{n,T}(\alpha)$ -- by simulating values of the Gaussian random variables $Z_{it}$. In Section \ref{subsec:test:impl}, we explain in detail how to compute Monte Carlo approximations of the quantiles $q_{n,T}(\alpha)$. 
\end{remark}

\subsection{Implementation of the test in practice}\label{subsec:test:impl}

In practice, we implement the test procedure as follows for a given significance level $\alpha \in (0, 1)$:
\begin{enumerate}[label=\textit{Step \arabic*.}, leftmargin=1.45cm]
\item Compute the $(1-\alpha)$-quantile $q_{n, T}(\alpha)$ of the Gaussian statistic $\Phi_{n,T}$ by Monte Carlo simulations. Specifically, draw a large number $L$ (say $L=5000$) of samples of independent standard normal random variables $\{Z_{it}^{(\ell)} : 1 \le t \le T, \, 1 \le i \le n \}$ for $1 \le \ell \le L$. For each sample $\ell$, compute the value $\Phi_{n,T}^{(\ell)}$ of the Gaussian statistic $\Phi_{n, T}$, and store these values. Calculate the empirical $(1-\alpha)$-quantile $\widehat{q}_{n, T}(\alpha)$ from the stored values $\{ \Phi_{n, T}^{(\ell)}: 1 \le \ell \le L \}$. Use $\widehat{q}_{n, T}(\alpha)$ as an approximation of the quantile $q_{n, T}(\alpha)$.
\item For each $(i, j)$ with $1 \le i < j \le n$ and each $(u, h) \in \grid$, reject the local null hypothesis $H_0^{[i, j]}(u, h)$ if $\widehat{\psi}^0_{ij, T}(u, h)> \widehat{q}_{n, T}(\alpha)$. Reject the global null hypothesis $H_0$ if at least one local hypothesis $H_0^{[i, j]}(u, h)$ is rejected. 
\item Display the test results as follows. For each pair of time series $(i,j)$, let $\mathcal{S}^{[i, j]}(\alpha)$ be the set of intervals $\mathcal{I}_{u, h}$ for which we reject $H_0^{[i, j]}(u, h)$. Produce a separate plot for each pair ($i,j)$ which displays the intervals in $\mathcal{S}^{[i, j]}(\alpha)$. This gives a graphical overview over the intervals where our tests finds a deviation from the null. 
\end{enumerate}

\begin{remark} 
In some cases, the number of intervals in $\mathcal{S}^{[i, j]}(\alpha)$ may be quite large, rendering the visual representation of the test results outlined in Step 3 cumbersome. To overcome this drawback, we replace $\mathcal{S}^{[i, j]}(\alpha)$ with the subset of minimal intervals $\mathcal{S}^{[i, j]}_{\text{min}}(\alpha) \subseteq \mathcal{S}^{[i, j]}(\alpha)$. As in \cite{Duembgen2002}, we call an interval $\mathcal{I}_{u, h} \in \mathcal{S}^{[i, j]}(\alpha)$ minimal if there is no other interval $\mathcal{I}_{u^\prime, h^\prime} \in \mathcal{S}^{[i, j]}(\alpha)$ such that $\mathcal{I}_{u^\prime, h^\prime} \subset \mathcal{I}_{u, h}$. Notably, our theoretical results remain to hold true when the sets $\mathcal{S}^{[i, j]}(\alpha)$ are replaced by $\mathcal{S}^{[i, j]}_{\text{min}}(\alpha)$. See Section \ref{sec:theo} for the details. 
\end{remark}

\section{Theoretical properties of the multiscale test}\label{sec:theo}

To start with, we investigate the auxiliary statistic
\begin{align}\label{eq:Phi_hat}
\widehat{\Phi}_{n,T} = \max_{1 \le i < j \le n}  \max_{(u, h) \in \mathcal{G}_T} \widehat{\phi}^0_{ij, T}(u, h),
\end{align}
where
\begin{equation*}
\widehat{\phi}^0_{ij, T}(u, h) =\bigg| \frac{\widehat{\phi}_{ij, T}(u, h)} {\{ \widehat{\sigma}_i^2 + \widehat{\sigma}_j^2 \}^{1/2}} \bigg| - \lambda(h)
\end{equation*}
and
\begin{align*}
\widehat{\phi}_{ij, T}(u, h) 
 & = \sum_{t=1}^T w_{t,T}(u, h) \big\{ (\varepsilon_{it} - \bar{\varepsilon}_i) + (\bfbeta_i - \widehat{\bfbeta}_i)^\top (\X_{it} - \bar{\X}_{i}) \\[-0.4cm]
 & \phantom{= \sum_{t=1}^T w_{t,T}(u, h) \big\{} - (\varepsilon_{jt} - \bar{\varepsilon}_j) -  (\bfbeta_j - \widehat{\bfbeta}_j)^\top (\X_{jt} - \bar{\X}_{j}) \big\}
\end{align*}
with $\bar{\varepsilon}_i = \bar{\varepsilon}_{i,T} := T^{-1} \sum_{t=1}^T \varepsilon_{it}$ and $\bar{\X}_{i} =  \bar{\X}_{i, T} := T^{-1}\sum_{t=1}^T  \X_{it}$. By construction, it holds that $\widehat{\phi}_{ij, T}(u, h) = \widehat{\psi}_{ij, T}(u, h)$ under $H_0^{[i, j]}(u, h)$. Hence, the auxiliary statistic $\widehat{\Phi}_{n,T}$ is identical to the multiscale test statistic $\widehat{\Psi}_{n,T}$ under the global null $H_0$. Consequently, in order to determine the distribution of the test statistic $\widehat{\Psi}_{n,T}$ under $H_0$, it suffices to study the auxiliary statistic $\widehat{\Phi}_{n,T}$. The following theorem shows that the distribution of $\widehat{\Phi}_{n,T}$ is close to the distribution of the Gaussian statistic $\Phi_{n,T}$ introduced in \eqref{eq:Phi}. 
\begin{theorem}\label{theo:stat:global}
Let \ref{C-err1}--\ref{C-h} be fulfilled. Moreover, assume that $\widehat{\sigma}_i^2 = \sigma^2_i + o_p(\rho_T)$ with $\rho_T = o(1/\log T)$ and $\sigma_i^2 = \sigma^2$ for all $i$. Then  
\begin{equation*}
\sup_{x \in \reals} \big| \pr(\widehat{\Phi}_{n, T} \le x) - \pr(\Phi_{n,T} \le x) \big| = o(1).
\end{equation*}
\end{theorem}
Theorem \ref{theo:stat:global} is key for deriving theoretical properties of our multiscale test. Its proof is provided in the Appendix.

\begin{remark}
The proof of Theorem \ref{theo:stat:global} builds on two important theoretical results: strong approximation theory for dependent processes \citep{BerkesLiuWu2014} and anti-concentration bounds for Gaussian random vectors \citep{Nazarov2003}. It can be regarded as a further development of the proof strategy in \cite{KhismatullinaVogt2020} who developed multiscale methods to test for local increases/decreases of the nonparametric trend function $m$ in the univariate time series model $Y_t = m(t/T) + \varepsilon_t$. We extend their theoretical results in several directions: we consider the case of multiple time series and work with a more general model which includes covariates and a flexible fixed effect error structure.  
\end{remark}

With the help of Theorem \ref{theo:stat:global}, we now examine the theoretical properties of our multiscale test developed in Section \ref{sec:test}. The following proposition shows that our test of the global null $H_0$ has correct (asymptotic) size.
\begin{prop}\label{prop:test}
Suppose that the conditions of Theorem \ref{theo:stat:global} are satisfied. Then under $H_0$, we have
\[ \pr \big( \widehat{\Psi}_{n,T} \le q_{n,T}(\alpha) \big) = (1 - \alpha) + o(1). \]
\end{prop}
The next proposition characterises the power of the test against a certain class of local alternatives. To formulate it, we consider a sequence of pairs of functions $m_ i := m_{i,T}$ and $m_ j := m_{j,T}$ that depend on the time series length $T$ and that are locally sufficiently far from each other.
\begin{prop}\label{prop:test:power}
Let the conditions  of Theorem \ref{theo:stat:global} be satisfied. Moreover, assume that for some pair of indices $i$ and $j$, the functions $m_ i = m_{i,T}$ and $m_ j = m_{j,T}$ have the following property: There exists $(u, h) \in \mathcal{G}_T$ with $[u-h, u+h] \subseteq [0,1]$ such that $m_{i,T}(w) - m_{j,T}(w) \ge c_T \sqrt{\log T/(Th)}$ for all $w \in [u-h, u+h]$ or $m_{j,T}(w) - m_{i,T}(w) \ge c_T \sqrt{\log T/(Th)}$ for all $w \in [u-h, u+h]$, where $\{c_T\}$ is any sequence of positive numbers with $c_T \rightarrow \infty$. Then 
\[ \pr \big( \widehat{\Psi}_{n,T} \le q_{n,T}(\alpha) \big) = o(1). \]
\end{prop}

We now turn attention to the local null hypotheses $H_0^{[i, j]}(u, h)$. As already defined above, let $\mathcal{M} = \{ (u,h,i,j): (u,h) \in \grid \text{ and } 1 \le i < j \le n \}$ be the collection of location-scale points $(u,h)$ and pairs of time series $(i,j)$ under consideration. Moreover, let $\mathcal{M}_0 \subseteq \mathcal{M}$ be the set of tuples $(u,h,i,j)$ for which $H_0^{[i, j]}(u, h)$ is true. Since we test $H_0^{[i, j]}(u, h)$ simultaneously for all $(u,h,i,j) \in \mathcal{M}$, we would like to bound the probability of making at least one false discovery. More formally, we would like to control the family-wise error rate 
\begin{align*} 
\text{FWER}(\alpha) 
 & = \pr \Big(\exists \,  (u,h,i,j) \in \mathcal{M}_0: \widehat{\psi}^0_{ij, T}(u, h) > q_{n,T}(\alpha) \Big) \\
 & = \pr \Big(\exists \,  (u,h,i,j) \in \mathcal{M}_0: \interval \in \mathcal{S}^{[i, j]}(\alpha) \Big)
\end{align*}
at level $\alpha$, where $\mathcal{S}^{[i, j]}(\alpha)$ is the set of intervals $\interval$ for which the null hypothesis $H_0^{[i, j]}(u, h)$ is rejected. The following result shows that our test procedure asymptotically controls the FWER at level $\alpha$.
\begin{prop}\label{prop:test:fwer}
Let the conditions of Theorem \ref{theo:stat:global} be satisfied. Then for any given $\alpha \in (0,1)$, 
\[ \textnormal{FWER}(\alpha) \leq \alpha + o(1). \]
\end{prop}
The statement of Proposition \ref{prop:test:fwer} can obviously be reformulated as follows: 
\begin{align*}
\pr\Big( \forall \, (u,h,i,j) \in \mathcal{M} \text{ for which } H_0^{[i, j]}(u, h) \text{ is true, } H_0^{[i, j]}(u, h) & \text{ is not rejected} \Big) \\ & \quad \ge 1 - \alpha + o(1).
\end{align*}
Consequently, we can make the following simultaneous confidence statement: 
\begin{equation}\label{eq:CS-v1}
\begin{minipage}[c][1.25cm][c]{13cm}
\textit{With (asymptotic) probability at least $1-\alpha$, the trends $m_i$ and $m_j$ differ on any interval $\interval$ and for all pairs $(i,j)$ for which $H_0^{[i, j]}(u, h)$ is rejected.}
\end{minipage}
\end{equation}
Put differently:
\begin{equation}\label{eq:CS-v2}
\begin{minipage}[c][1.25cm][c]{13cm}
\textit{With (asymptotic) probability at least $1-\alpha$, for any pair $(i,j)$ with $1 \le i < j \le n$, the two trends $m_i$ and $m_j$ differ on all intervals $\interval \in \mathcal{S}^{[i, j]}(\alpha)$.}
\end{minipage}
\end{equation}

\begin{remark}
According to \eqref{eq:CS-v2}, the graphical device introduced in Section \ref{subsec:test:impl} which depicts the intervals in $\mathcal{S}^{[i, j]}(\alpha)$ comes with the following statistical guarantee: we can claim with (asymptotic) confidence at least $1-\alpha$ that the trends $m_i$ and $m_j$ differ on all the depicted intervals. As $\mathcal{S}^{[i, j]}_{\text{min}}(\alpha) \subseteq \mathcal{S}^{[i, j]}(\alpha)$ for all $i$ and $j$, this guarantee remains to hold true when the sets $\mathcal{S}^{[i, j]}(\alpha)$ are replaced by the sets of minimal intervals $\mathcal{S}_{\text{min}}^{[i, j]}(\alpha)$.
\end{remark}

\section{Clustering}\label{sec:clustering}

Consider a situation in which the null hypothesis $H_0: m_1 = m_2 = \ldots = m_n$ is violated. Even though some of the trend functions are different in this case, part of them may still be the same. Put differently, there may be groups of time series which have the same time trend. Formally speaking, we define a group structure as follows: There exist sets or groups of time series $G_1,\ldots, G_N$ with $N \le n$ and $\{1,\ldots, n\} = \mathbin{\dot{\bigcup}}_{\ell=1}^{N} G_\ell$ such that for each $1 \le \ell \le N$,
\[ m_i = f_\ell \quad \text{for all } i \in G_\ell, \]
where $f_\ell$ are group-specific trend functions. Hence, the time series which belong to the group $G_\ell$ all have the same time trend $f_\ell$. Throughout the section, we suppose that the group-specific trend functions $f_\ell$ have the following properties: 
\begin{enumerate}[label=(C\arabic*),leftmargin=1.2cm]
\setcounter{enumi}{11}
\item \label{C-clustering} For each $\ell$, $f_\ell = f_{\ell,T}$ is a Lipschitz continuous function with $\int_0^1 f_{\ell,T}(w) dw = 0$. In particular, $|f_{\ell,T}(v) - f_{\ell,T}(w)| \le L |v-w|$ for all $v, w \in [0,1]$ and some constant $L < \infty$ that does not depend on $T$. Moreover, for any $\ell \ne \ell^\prime$, the trends $f_{\ell,T}$ and $f_{\ell^\prime,T}$ differ in the following sense: There exists $(u, h) \in \mathcal{G}_T$ with $\interval \subseteq [0,1]$ such that $f_{\ell,T}(w) - f_{\ell^\prime,T}(w) \ge c_T \sqrt{\log T/(Th)}$ for all $w \in \interval$ or $f_{\ell^\prime,T}(w) - f_{\ell,T}(w) \ge c_T \sqrt{\log T/(Th)}$ for all $w \in \interval$, where $0 < c_T \rightarrow \infty$.
\end{enumerate}
In many applications, it is natural to suppose that there is a group structure in the data. In this case, a particular interest lies in estimating the unknown groups from the data at hand. In what follows, we combine our multiscale methods with a clustering algorithm to achieve this. More specifically, we use the multiscale statistics $\max_{(u, h) \in \mathcal{G}_T}\hat{\psi}^0_{ij, T}(u, h)$ calculated for each $i$ and $j$ as distance measures which are fed into a hierarchical clustering algorithm.

To describe the algorithm, we first need to introduce the notion of a dissimilarity measure: Let $S \subseteq \{1,\ldots, n\}$ and $S^\prime \subseteq \{1,\ldots, n\}$ be two sets of time series from our sample. We define a dissimilarity measure between $S$ and $S^\prime$ by setting 
\begin{equation}\label{dissimilarity}
\widehat{\Delta}(S,S^\prime) = \max_{\substack{i \in S, \\ j \in S^\prime}} \Big(\max_{(u, h) \in \mathcal{G}_T}\hat{\psi}^0_{ij, T}(u, h)\Big). 
\end{equation}
This is commonly called a complete linkage measure of dissimilarity. Alternatively, we may work with an average or a single linkage measure. We now combine the dissimilarity measure $\widehat{\Delta}$ with a hierarchical agglomerative clustering (HAC) algorithm which proceeds as follows: 
\vspace{10pt}

\noindent \textit{Step $0$ (Initialisation):} Let $\widehat{G}_i^{[0]} = \{ i \}$ denote the $i$-th singleton cluster for $1 \le i \le n$ and define $\{\widehat{G}_1^{[0]},\ldots,\widehat{G}_n^{[0]} \}$ to be the initial partition of time series into clusters. 
\vspace{5pt}

\noindent \textit{Step $r$ (Iteration):} Let $\widehat{G}_1^{[r-1]},\ldots,\widehat{G}_{n-(r-1)}^{[r-1]}$ be the $n-(r-1)$ clusters from the previous step. Determine the pair of clusters $\widehat{G}_{\ell}^{[r-1]}$ and $\widehat{G}_{{\ell}^\prime}^{[r-1]}$ for which 

\[ \widehat{\Delta}(\widehat{G}_{\ell}^{[r-1]},\widehat{G}_{{\ell}^\prime}^{[r-1]}) = \min_{1 \le k < k^\prime \le n-(r-1)} \widehat{\Delta}(\widehat{G}_{k}^{[r-1]},\widehat{G}_{k^\prime}^{[r-1]}) \]  
and merge them into a new cluster. 
\vspace{10pt}

\noindent Iterating this procedure for $r = 1,\ldots, n-1$ yields a tree of nested partitions \linebreak $\{\widehat{G}_1^{[r]},\ldots$ $\ldots,\widehat{G}_{n-r}^{[r]}\}$, which can be graphically represented by a dendrogram. Roughly speaking, the HAC algorithm merges the $n$ singleton clusters $\widehat{G}_i^{[0]} = \{ i \}$ step by step until we end up with the cluster $\{1,\ldots, n\}$. In each step of the algorithm, the closest two clusters are merged, where the distance between clusters is measured in terms of the dissimilarity $\widehat{\Delta}$. We refer the reader to Section 14.3.12 in \cite{HastieTibshiraniFriedman2009} for an overview of hierarchical clustering methods.

When the number of groups $N$ is known, we estimate the group structure $\{G_1,\ldots, G_N\}$ by the $N$-partition $\{\widehat{G}_1^{[n-N]},\ldots,\widehat{G}_{N}^{[n-N]}\}$ produced by the HAC algorithm. When $N$ is unknown, we estimate it by the $\widehat{N}$-partition $\{\widehat{G}_1^{[n-\widehat{N}]},\ldots,\widehat{G}_{\widehat{N}}^{[n-\widehat{N}]}\}$, where $\widehat{N}$ is an estimator of $N$. The latter is defined as 
\[ \widehat{N} = \min \Big\{ r = 1,2,\ldots \Big| \max_{1 \le \ell \le r} \widehat{\Delta} \big( \widehat{G}_\ell^{[n-r]} \big) \le q_{n,T}(\alpha) \Big\}, \]
where we write $\widehat{\Delta}(S) = \widehat{\Delta}(S,S)$ for short and $q_{n,T}(\alpha)$ is the $(1-\alpha)$-quantile of $\Phi_{n,T}$ defined in Section \ref{subsec:test:test}.

The following proposition summarises the theoretical properties of the estimators $\widehat{N}$ and $\{ \widehat{G}_1,\ldots,\widehat{G}_{\widehat{N}} \}$, where we use the shorthand $\widehat{G}_\ell = \widehat{G}_\ell^{[n-\widehat{N}]}$ for $1 \le \ell \le \widehat{N}$. 
\begin{prop}\label{prop:clustering:1}
Let the conditions of Theorem \ref{theo:stat:global} and \ref{C-clustering} be satisfied. Then 
\[ \pr \Big( \big\{ \widehat{G}_1,\ldots,\widehat{G}_{\widehat{N}} \big\} = \{ G_1,\ldots, G_N \} \Big) \ge (1-\alpha) + o(1) \]
and 
\[ \pr \big( \widehat{N} = N \big) \ge (1-\alpha) + o(1). \]
\end{prop}
This result allows us to make statistical confidence statements about the estimated clusters $\{ \widehat{G}_1,\ldots,\widehat{G}_{\widehat{N}} \}$ and their number $\widehat{N}$. In particular, we can claim with asymptotic confidence at least $1 - \alpha$ that the estimated group structure is identical to the true group structure. 
The proof of Proposition \ref{prop:clustering:1} can be found in the Appendix.

Our multiscale methods do not only allow us to compute estimators of the unknown groups $G_1,\ldots, G_N$ and their number $N$. They also provide information on the locations where two group-specific trend functions $f_\ell$ and $f_{\ell^\prime}$ differ from each other. To turn this claim into a mathematically precise statement, we need to introduce some notation. First of all, note that the indexing of the estimators $\widehat{G}_1,\ldots,\widehat{G}_{\widehat{N}}$ is completely arbitrary. We could, for example, change the indexing according to the rule $\ell \mapsto \widehat{N} - \ell + 1$. In what follows, we suppose that the estimated groups are indexed such that $P( \widehat{G}_\ell = G_\ell \text{ for all } \ell ) \ge (1-\alpha) + o(1)$. Proposition \ref{prop:clustering:1} implies that this is possible without loss of generality. Keeping this convention in mind, we define the sets 
\[ \mathcal{A}_{n,T}^{[\ell,\ell^\prime]}(\alpha)= \Big\{ (u, h) \in \mathcal{G}_T: \widehat{\psi}^0_{ij, T}(u, h) > q_{n,T}(\alpha) \text{ for some } i \in \widehat{G}_\ell, j \in \widehat{G}_{\ell^\prime} \Big\} \]
and  
\[ \mathcal{S}^{[\ell,\ell^\prime]}_{n, T}(\alpha) = \big\{ \interval = [u-h, u+h]: (u, h) \in \mathcal{A}_{n,T}^{[\ell,\ell^\prime]}(\alpha) \big\} \]
for $1 \le \ell < \ell^\prime \le \widehat{N}$. An interval $\interval$ is contained in $\mathcal{S}^{[\ell,\ell^\prime]}_{n, T}(\alpha)$ if our multiscale test indicates a significant difference between the trends $m_i$ and $m_j$ on the interval $\interval$ for some $i \in \widehat{G}_\ell$ and $j \in \widehat{G}_{\ell^\prime}$. Put differently, $\interval \in \mathcal{S}^{[\ell,\ell^\prime]}_{n, T}(\alpha)$ if the test suggests a significant difference between the trends of the $\ell$-th and the $\ell^\prime$-th group on the interval $\interval$. We further let
\[ E_{n,T}^{[\ell,\ell^\prime]}(\alpha) = \Big\{ \forall \interval \in \mathcal{S}^{[\ell,\ell^\prime]}_{n, T}(\alpha): f_\ell(v) \ne f_{\ell^\prime}(v) \text{ for some } v \in \interval = [u-h, u+h] \Big\} \]
be the event that the group-specific time trends $f_\ell$ and $f_{\ell^\prime}$ differ on all intervals $\interval \in \mathcal{S}^{[\ell,\ell^\prime]}_{n, T}(\alpha)$. With this notation at hand, we can make the following formal statement whose proof is given in the Appendix.  
\begin{prop}\label{prop:clustering:2}
Under the conditions of Proposition \ref{prop:clustering:1}, the event 
\[ E_{n,T}(\alpha) = \Big\{ \bigcap_{1 \le \ell < \ell^\prime \le \widehat{N}} E_{n,T}^{[\ell,\ell^\prime]}(\alpha) \Big\} \cap \Big\{ \widehat{N} = N \text{ and } \widehat{G}_\ell = G_\ell \text{ for all } \ell \Big\} \]
asymptotically occurs with probability at least $1-\alpha$, that is, 
\[ \pr \big( E_{n,T}(\alpha) \big) \ge (1 - \alpha) + o(1). \]
\end{prop}
According to Proposition \ref{prop:clustering:2}, the sets $\mathcal{S}^{[\ell,\ell^\prime]}_{n, T}(\alpha)$ allow us to locate, with a pre-specified confidence, time intervals where the group-specific trend functions $f_\ell$ and $f_{\ell^\prime}$ differ from each other. In particular, we can claim with asymptotic confidence at least $1 - \alpha$ that the trend functions $f_\ell$ and $f_{\ell^\prime}$ differ on all intervals $\interval \in \mathcal{S}^{[\ell,\ell^\prime]}_{n, T}(\alpha)$.\footnote{This statement remains to hold true when the sets of intervals $\mathcal{S}^{[\ell,\ell^\prime]}_{n, T}(\alpha)$ are replaced by the corresponding sets of minimal intervals.}

\section{Simulations}\label{sec:sim}

In this section, we investigate our testing and clustering methods by means of a simulation study. The simulation design is as follows: We generate data from the model $Y_{it} = m_i(\frac{t}{T}) + \beta_i X_{it} +  \alpha_i  + \varepsilon_{it}$, where the number of time series $i$ is set to $n = 15$ and we consider different time series lengths $T$. For simplicity, we assume that the fixed effect term $\alpha_i$ is equal to $0$ in all time series and we only include a single regressor $X_{it}$ in the model, setting the unknown parameter $\beta_i$ to $1$ for all $i$. 
For each $i$, the errors $\varepsilon_{it}$ follow the AR(1) model $\varepsilon_{it} = a \varepsilon_{i, t-1} + \eta_{it}$, where $a = 0.25$ and the innovations $\eta_{it}$ are i.i.d.\ normally distributed with mean $0$ and variance $0.25$. Similarly, for each $i$, the covariates $X_{it}$ follow an AR($1$) process of the form $X_{it} = a_x X_{i, t-1} + \zeta_{it}$, where $a_x = 0.5$ and the innovations $\zeta_{it}$ are i.i.d.\ normally distributed with mean $0$ and variance $1$. We assume that the covariates $X_{it}$ and the error terms $\varepsilon_{js}$ are independent from each other for all $1 \leq i,j \leq n$ and $1 \leq t, s \leq T$. 
To generate data under the null $H_0: m_1 = \ldots = m_n$, we let $m_i = 0$ for all $i$ without loss of generality. To produce data under the alternative, we define $m_1(u) = b \, (u - 0.5) $ with $b \in \{ 0.75, 1, 1.25 \}$ and set $m_i = 0$ for all $i \ne 1$. Hence, all trend functions are the same except for $m_1$ which is an increasing linear function. Note that the normalisation constraint $\int_0^1 m_1(u) du = 0$ is directly satisfied in this case. For each simulation exercise, we simulate $5000$ data samples.

\begin{table}[t]
\footnotesize{
\begin{center}
\caption{Size of the multiscale test for different sample sizes $T$ and nominal sizes $\alpha$.}
\label{tab:size}
\renewcommand{\arraystretch}{1.2}
%
\begin{tabular}{cccc}
  \hline
  & \multicolumn{3}{c}{nominal size $\alpha$} \\
 $T$ & 0.01 & 0.05 & 0.1 \\
 \hline
100 & 0.009 & 0.045 & 0.087 \\ 
  250 & 0.013 & 0.063 & 0.117 \\ 
  500 & 0.013 & 0.057 & 0.112 \\ 
   \hline
\end{tabular}

\end{center}}
\footnotesize{
\begin{center}
\caption{Power of the multiscale test for different sample sizes $T$ and nominal sizes $\alpha$. Each panel corresponds to a different slope parameter $b$.}\label{tab:power}
\begin{subtable}[b]{0.32\textwidth}
\centering
\caption{$b = 0.75$}\label{tab:power_075}
\renewcommand{\arraystretch}{1.2}
%
\begin{tabular}{cccc}
  \hline
  & \multicolumn{3}{c}{nominal size $\alpha$} \\
 $T$ & 0.01 & 0.05 & 0.1 \\
 \hline
100 & 0.033 & 0.122 & 0.199 \\ 
  250 & 0.209 & 0.434 & 0.549 \\ 
  500 & 0.741 & 0.891 & 0.947 \\ 
   \hline
\end{tabular}

\end{subtable}
\begin{subtable}[b]{0.32\textwidth}
\centering
\caption{$b = 1.00$}\label{tab:power_100}
\renewcommand{\arraystretch}{1.2}
%
\begin{tabular}{cccc}
  \hline
  & \multicolumn{3}{c}{nominal size $\alpha$} \\
 $T$ & 0.01 & 0.05 & 0.1 \\
 \hline
100 & 0.105 & 0.270 & 0.376 \\ 
  250 & 0.635 & 0.840 & 0.901 \\ 
  500 & 0.994 & 0.999 & 0.999 \\ 
   \hline
\end{tabular}

\end{subtable}
\begin{subtable}[b]{0.32\textwidth}
\centering
\caption{$b = 1.25$}\label{tab:power_125}
\renewcommand{\arraystretch}{1.2}
%
\begin{tabular}{cccc}
  \hline
  & \multicolumn{3}{c}{nominal size $\alpha$} \\
 $T$ & 0.01 & 0.05 & 0.1 \\
 \hline
100 & 0.275 & 0.512 & 0.628 \\ 
  250 & 0.933 & 0.986 & 0.993 \\ 
  500 & 1.000 & 1.000 & 1.000 \\ 
   \hline
\end{tabular}

\end{subtable}
\end{center}}
\vspace{-0.4cm}
\end{table}

Our multiscale test is implemented as follows: The estimators $\widehat{\beta}_i$ and $\widehat{\alpha}_i$ of the unknown parameters $\beta_i$ and $\alpha_i$ are computed as described in Section~\ref{subsec:test:prep}. Since the errors $\varepsilon_{it}$ follow an AR($1$) process, we estimate the long-run error variance $\sigma_i^2$ by the difference-based estimator proposed in \cite{KhismatullinaVogt2020}, setting the tuning parameters $q$ and $r$ to $25$ and $10$, respectively. In order to construct our test statistics, we use an Epanechnikov kernel and the grid $\mathcal{G}_T = U_T \times H_T$ with 
\begin{align*}
U_T & = \big\{ u \in [0,1]: u = \textstyle{\frac{5t}{T}} \text{ for some } t \in \naturals \big\} \\
H_T & = \big\{ h \in \big[ \textstyle{\frac{\log T}{T}}, \textstyle{\frac{1}{4}} \big]:  h = \textstyle{\frac{5 t -3}{T}} \text{ for some } t \in \naturals \big\}. 
\end{align*}
We thus consider intervals $\mathcal{I}_{u, h} = [u-h, u+h]$ which contain $5, 15, 25, \ldots$ data points. The critical value $q_{n,T}(\alpha)$ of our test is computed by Monte Carlo methods as described in Section \ref{subsec:test:impl}, where we set $L=5000$.

The simulation results for our multiscale test can be found in Table \ref{tab:size}, which reports its actual size under $H_0$, and in Table \ref{tab:power}, which reports its power against different alternatives. Both the actual size and the power are computed as the number of simulations in which the test rejects the global null $H_0$ divided by the total number of simulations.
Inspecting Table \ref{tab:size}, one can see that the actual size is fairly close to the nominal target $\alpha$ in all the considered scenarios. Hence, the test has approximately the correct size. 
Inspecting Table \ref{tab:power}, one can further see that the test has reasonable power properties. For the smallest slope $b=0.75$ and the smallest time series length $T=100$, the power is only moderate, reflecting the fact that the alternative with $b=0.75$ is not very far away from the null. However, as we increase the slope $b$ and the sample size $T$, the power increases quickly. Already for the slope value $b = 1.00$, we reach a power of $0.99$ for $T = 500$ and for all nominal sizes $\alpha$.

We next investigate the finite sample performance of the clustering algorithm from Section \ref{sec:clustering}. To do so, we consider a very simple scenario: we generate data from the model $Y_{it} = m_i(\frac{t}{T}) + \varepsilon_{it}$, that is, we assume that there are no fixed effects and no covariates. The error terms $\varepsilon_{it}$ are specified as above. Moreover, as before, we set the number of time series to $n = 15$ and we consider different time series lengths $T$. We partition the $n = 15$ time series into $N=3$ groups, each containing $5$ time series. Specifically, we set $G_1 = \{1,\ldots, 5\}$, $G_2 = \{6,\ldots, 10\}$ and $G_3 =  \{11,\ldots, 15\}$, and we assume that $m_i = f_l$ for all $i \in G_l$ and all $l = 1, 2, 3$. The group-specific trend functions $f_1$, $f_2$ and $f_3$ are defined as $f_1(u) = 0$, $f_2(u) = 1 \cdot (u - 0.5)$ and $f_3(u) =  (- 1) \cdot (u - 0.5)$. In order to estimate the groups $G_1$, $G_2$, $G_3$ and their number $N = 3$, we use the same implementation as before followed by the clustering procedure from Section \ref{sec:clustering}.

\addtocounter{table}{-1} 
\begin{table}[t]
\footnotesize{
\begin{center}
\caption{Clustering results for different sample sizes $T$ and nominal sizes $\alpha$.}\label{tab:clustering}
\begin{subtable}[b]{0.48\textwidth}
\centering
\caption{Empirical probabilities that \\ $\widehat{N} = N$}\label{tab:clustering:1}
\renewcommand{\arraystretch}{1.2}
%
\begin{tabular}{cccc}
  \hline
  & \multicolumn{3}{c}{nominal size $\alpha$} \\
 $T$ & 0.01 & 0.05 & 0.1 \\
 \hline
100 & 0.055 & 0.188 & 0.298 \\ 
  250 & 0.713 & 0.922 & 0.939 \\ 
  500 & 0.994 & 0.979 & 0.956 \\ 
   \hline
\end{tabular}

\end{subtable}
\begin{subtable}[b]{0.48\textwidth}
\centering
\caption{\centering Empirical probabilities that $\{ \widehat{G}_1,\ldots,\widehat{G}_{\widehat{N}}\} = \{ G_1,G_2,G_3\}$}\label{tab:clustering:2}
\renewcommand{\arraystretch}{1.2}
%
\begin{tabular}{cccc}
  \hline
  & \multicolumn{3}{c}{nominal size $\alpha$} \\
 $T$ & 0.01 & 0.05 & 0.1 \\
 \hline
100 & 0.009 & 0.045 & 0.077 \\ 
  250 & 0.640 & 0.825 & 0.845 \\ 
  500 & 0.992 & 0.978 & 0.956 \\ 
   \hline
\end{tabular}

\end{subtable}
\end{center}}
\vspace{-0.5cm}
\end{table}

The simulation results are reported in Table \ref{tab:clustering}. The entries in Table \ref{tab:clustering:1} are computed as the number of simulations for which $\widehat{N} = N$ divided by the total number of simulations. They thus specify the empirical probabilities with which the estimate $\widehat{N}$ is equal to the true number of groups $N = 3$. Analogously, the entries of Table \ref{tab:clustering:2} give the empirical probabilities with which the estimated group structure $\{ \widehat{G}_1,\ldots,\widehat{G}_{\widehat{N}}\}$ equals the true one $\{G_1,G_2,G_3\}$. The results in Table \ref{tab:clustering} nicely illustrate the theoretical properties of our clustering algorithm. According to Proposition \ref{prop:clustering:1}, the probability that $\widehat{N} = N$ and $\{ \widehat{G}_1,\ldots,\widehat{G}_{\widehat{N}}\} = \{G_1,G_2,G_3\}$ should be at least $(1-\alpha)$ asymptotically. For the largest sample size $T = 500$, the empirical probabilities reported in Table \ref{tab:clustering} can indeed be seen to exceed the value $(1-\alpha)$ as predicted by Proposition \ref{prop:clustering:1}. For the smaller sample sizes $T=100$ and $T=250$, in contrast, the empirical probabilities are substantially smaller than $(1-\alpha)$. This reflects the asymptotic nature of Proposition \ref{prop:clustering:1} and is not very surprising. It simply mirrors the fact that for the smaller sample sizes $T=100$ and $T=250$, the effective noise level in the simulated data is quite high.

\begin{figure}[t!]
\centering
\includegraphics[width=\textwidth]{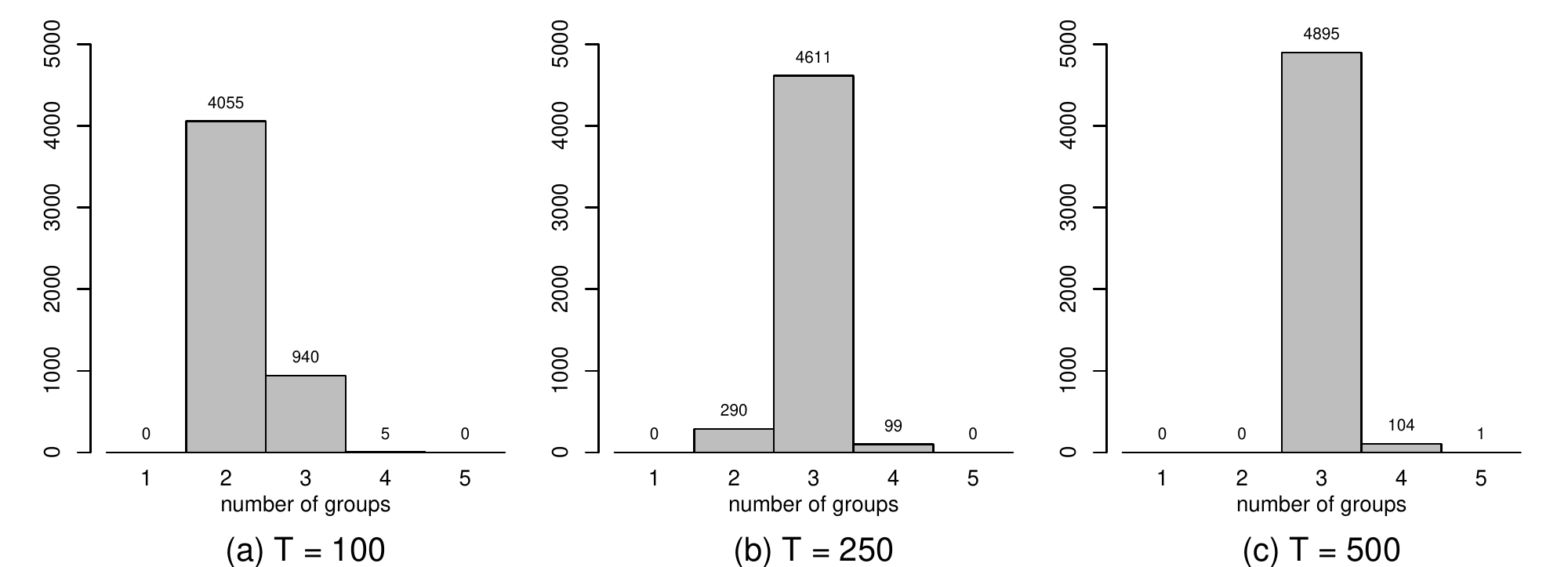}
\caption{Estimated number of groups $\widehat{N}$ for nominal size $\alpha = 0.05$. Each panel corresponds to a different sample size $T$.}\label{fig:clustering:1}
\vspace{0.25cm}

\includegraphics[width=\textwidth]{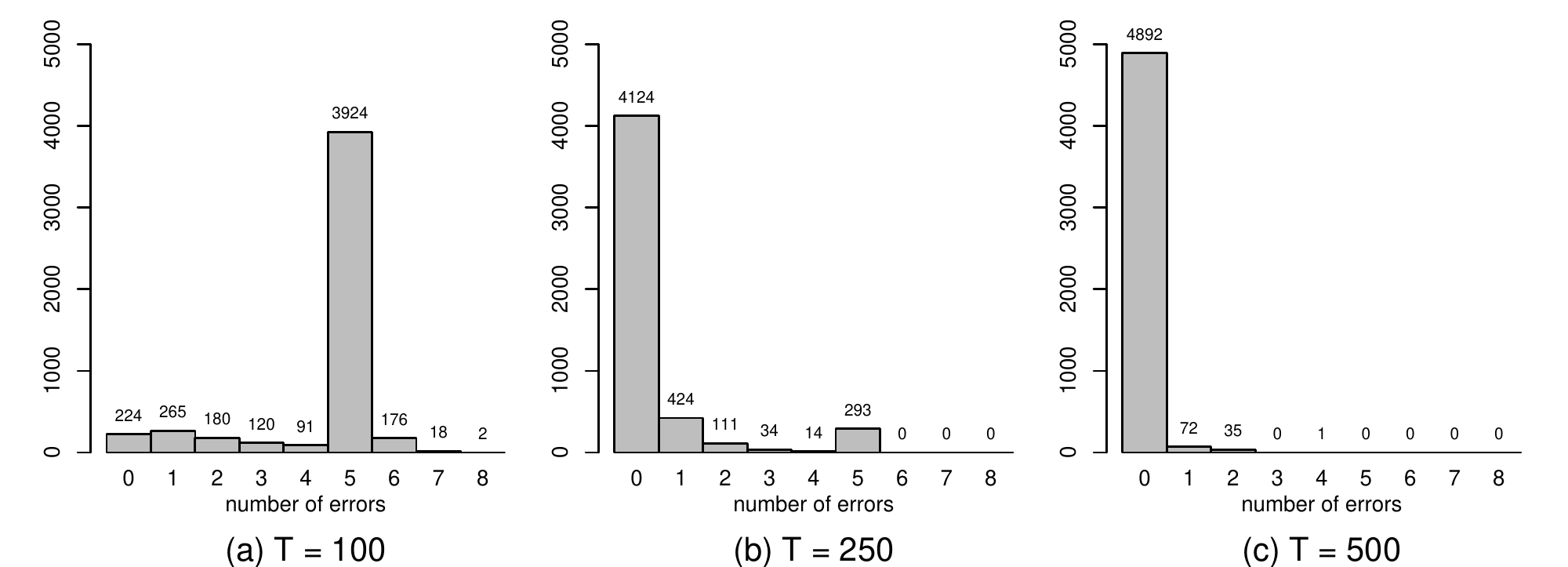}
\caption{Number of classification errors for nominal size $\alpha = 0.05$. Each panel corresponds to a different sample size $T$.}\label{fig:clustering:2}
\end{figure}

Figures \ref{fig:clustering:1} and \ref{fig:clustering:2} give more insight into what happens for $T=100$ and $T=250$. Figure \ref{fig:clustering:1} shows histograms of the $5000$ simulated values of $\widehat{N}$, while Figure \ref{fig:clustering:2} depicts histograms of the number of classification errors produced by our algorithm. By the number of classification errors, we simply mean the number of incorrectly classified time series, which is formally calculated as 
\[ \min_{\pi \in S_{\hat{N}}} \big\{ |G_1 \setminus \widehat{G}_{\pi(1)}| +|G_2 \setminus \widehat{G}_{\pi(2)}| + |G_3 \setminus \widehat{G}_{\pi(3)}| \big\} \]
with $S_{\widehat{N}}$ being the set of all permutations of $\{1, 2, \ldots, \widehat{N}\}$. The histogram of Figure \ref{fig:clustering:1} for $T=100$ clearly shows that our method underestimates the number of groups ($\widehat{N} = 2$ in $4055$ cases out of $5000$). In particular, it fails to detect the difference between two out of three groups in a large number of simulations. This is reflected in the corresponding histogram of Figure \ref{fig:clustering:2} which shows that there are exactly $5$ classification errors in $3924$ of the $5000$ simulation runs. In most of these cases, the estimated group structure $\{ \widehat{G}_1, \widehat{G}_{2}\}$ coincides with either $\{ G_1 \cup G_2,G_3\}$,  $\{ G_1, G_2\cup G_3\}$ or $ \{ G_1 \cup G_3,G_2\}$. In summary, we can conclude that the small empirical probabilities for $T=100$ in Table \ref{tab:clustering} are due to the algorithm underestimating the number of groups. Inspecting the histograms for $T=250$, one can see that the performance of the estimators $\widehat{N}$ and $\{ \widehat{G}_1,\ldots, \widehat{G}_{\widehat{N}} \}$ improves significantly, even though the corresponding empirical probabilities in Table \ref{tab:clustering} are still somewhat below the target $(1-\alpha)$.

\section{Applications}\label{sec:app}
\subsection{Analysis of GDP growth}\label{subsec:app:gdp}

In what follows, we revisit an application example from \cite{Zhang2012}. The aim is to test the hypothesis of a common trend in the GDP growth time series of several OECD countries. Since we do not have access to the original dataset of \cite{Zhang2012} and we do not know the exact data specifications used there, we work with data from the following common sources: Refinitiv Datastream, the OECD.Stat database, Federal Reserve Economics Data (FRED) and the Barro-Lee Educational Attainment dataset \citep*{Barro2013}. We consider a data specification that is as close as possible to the one in \cite{Zhang2012} with one important distinction: In the original study, the authors examined $16$ OECD countries (not specifying which ones) over the time period from the fourth quarter of 1975 up to and including the third quarter of 2010, whereas we consider only $11$ countries (Australia, Austria, Canada, Finland, France, Germany, Japan, Norway, Switzerland, UK and USA) over the same time span. The reason is that we have access to data of good quality only for these 11 countries. 
In the following list, we specify the data for our analysis.\footnote{All data were accessed and downloaded on 7 December 2021.}

\begin{itemize}[leftmargin=0.5cm]
\item \textbf{Gross domestic product ($\boldsymbol{GDP}$):} We use freely available data on \textit{Gross Domestic Product -- Expenditure Approach} from the OECD.Stat 
database (\texttt{https://stats.\linebreak oecd.org/Index.aspx}). To be as close as possible to the specification of the data in \cite{Zhang2012}, we use seasonally adjusted quarterly data on GDP expressed in millions of 2015 US dollars.\footnote{Since the publication of \cite{Zhang2012}, the OECD reference year has changed from 2005 to 2015. We have decided to analyse the latest version of the data in order to be able to make more accurate and up-to-date conclusions.} 

\item \textbf{Capital ($\boldsymbol{K}$):} We use data on \textit{Gross Fixed Capital Formation} from the OECD.Stat 
database. The data are at a quarterly frequency, seasonally adjusted, and expressed in millions of $2015$ US dollars. In contrast to \cite{Zhang2012}, who use data on \textit{Capital Stock at Constant National Prices}, we choose to work with gross fixed capital formation due to data availability. It is worth noting that since accurate data on capital stock is notoriously difficult to collect, the use of gross fixed capital formation as a measure of capital is standard in the literature; see e.g.\ \cite{Sharma1994}, \cite{Lee2002} and \cite{Lee2005}.

\item \textbf{Labour ($\boldsymbol{L}$):} We collect data on the \textit{Number of Employed People} from various sources. For most of the countries (Austria, Australia, Canada, Germany, Japan, UK and USA), we download the OECD data on \textit{Employed Population: Aged 15 and Over} retrieved from FRED (\texttt{https://fred.stlouisfed.org/}). 
The data for France and Switzerland were downloaded from Refinitiv Datastream. For all of the aforementioned countries, the observations are at a quarterly frequency and seasonally adjusted. The data for Finland and Norway were also obtained via Refinitiv Datastream, however, the only quarterly time series that are long enough for our purposes are not seasonally adjusted. Hence, for these two countries, we perform the seasonal adjustment ourselves. We in particular use the default method of the function \verb|seas| from the \verb|R| package \verb|seasonal| \citep*{Sax2018} which is an interface to X-13-ARIMA-SEATS, the seasonal adjustment software used by the US Census Bureau. 
For all of the countries, the data are given in thousands of persons.

\item \textbf{Human capital ($\boldsymbol{H}$):} We use \textit{Educational Attainment for Population Aged 25 and Over} collected from \texttt{http://www.barrolee.com} as a measure of human capital. Since the only available data is five-year census data, we follow \cite{Zhang2012} and use linear interpolation between the observations and constant extrapolation on the boundaries (second and third quarters of 2010) to obtain the quarterly time series.
\end{itemize}

For each of the $n=11$ countries in our sample, we thus observe a quarterly time series $\mathcal{T}_i = \{(Y_{it}, \X_{it}): 1 \le t \le T \}$ of length $T = 140$, where $Y_{it} = \Delta \ln GDP_{it} := \ln GDP_{it} - \ln GDP_{i(t-1)}$ and $\X_{it} = (\Delta \ln L_{it}, \Delta \ln K_{it}, \Delta \ln H_{it})^\top$ with $\Delta \ln L_{it} := \ln L_{it} - \ln L_{i(t-1)}$, $\Delta \ln K_{it} := \ln K_{it} - \ln K_{i(t-1)}$ and $\Delta \ln H_{it} := \ln H_{it} - \ln H_{i(t-1)}$. Without loss of generality, we let $\Delta \ln GDP_{i1} = \Delta \ln L_{i1} = \Delta \ln K_{i1} = \Delta \ln H_{i1} = 0$. Each time series $\mathcal{T}_i$ is assumed to follow the model $Y_{it} = m_i(t/T) + \bfbeta^\top_i \X_{it} + \alpha_i + \varepsilon_{it}$, or equivalently,  
\begin{align}
\Delta \ln GDP_{it}
 & = m_i \Big( \frac{t}{T} \Big) + \beta_{i, 1} \Delta \ln L_{it} + \beta_{i, 2} \Delta \ln K_{it} + \beta_{i, 3} \Delta \ln H_{it} + \alpha_i + \varepsilon_{it} \label{eq:model:app}
\end{align}
for $1 \le t \le T$, where $\bfbeta_i = (\beta_{i, 1}, \beta_{i, 2}, \beta_{i, 3})^\top$ is a vector of unknown parameters, $m_i$ is a country-specific unknown nonparametric time trend and $\alpha_i$ is a country-specific fixed effect.

In order to test the null hypothesis $H_0: m_1 = \ldots = m_n$ with $n = 11$ in model \eqref{eq:model:app}, we implement our multiscale test as follows: 
\begin{itemize}[leftmargin=0.5cm]

\item We choose $K$ to be the Epanechnikov kernel and consider the set of location-scale points $\mathcal{G}_T = U_T \times H_T$, where 
\begin{align*}
U_T & = \big\{ u \in [0,1]: u = \textstyle{\frac{8t + 1}{2T}} \text{ for some } t \in \naturals \big\} \\
H_T & = \big\{ h \in \big[ \textstyle{\frac{\log T}{T}}, \textstyle{\frac{1}{4}} \big]:  h = \textstyle{\frac{4t}{T}} \text{ for some } t \in \naturals \big\}. 
\end{align*}
We thus take into account all locations $u$ on an equidistant grid $U_T$ with step length $4/T$ and all scales $h=4/T, 8/T, 12/T,\ldots$ with $\log T /T \le h \le 1/4$. 
The choice of the grid $\grid$ is motivated by the quarterly frequency of the data: each interval $\interval \in \mathcal{G}_T$ spans $8, 16, 24, \ldots$ quarters, i.e., $2, 4, 6, \ldots$ years. The lower bound $\log T / T$ on the scales $h$ in $H_T$ is motivated by Assumption \ref{C-h}, which requires that $\log T /T \ll h_{\min}$ (given that all moments of $\varepsilon_{it}$ exist).


\item To obtain an estimator $\hat{\sigma}_i^2$ of the long-run error variance $\sigma^2_i$ for each $i$, we assume that the error process $\mathcal{E}_i$ follows an AR($p_i$) model and apply the difference-based procedure of \cite{KhismatullinaVogt2020} to the augmented time series $\{\widehat{Y}_{it}: 1\leq t \leq T\}$ with $\widehat{Y}_{it} = Y_{it} - \widehat{\bfbeta}_i^\top \X_{it} - \widehat{\alpha}_{i}$. We set the tuning parameters $q$ and $r$ of the procedure to $20$ and $10$, respectively, and choose the AR order $p_i$ by minimizing the Bayesian Information Criterion (BIC), which yields $p_i = 3$ for Australia, Canada and the UK and $p_i = 1$ for all other countries.\footnote{We also calculated the values of other information criteria such as FPE, AIC and HQ which, in most of the cases, resulted in the same values of $p_i$.} 

\item The critical values $q_{n, T}(\alpha)$ are computed by Monte Carlo methods as described in Section \ref{subsec:test:impl}, where we set $L=5000$.
\end{itemize}
Besides these choices, we construct and implement the multiscale test exactly as described in Section \ref{sec:test}.

The thus implemented multiscale test rejects the global null hypothesis $H_0$ at the usual significance levels $\alpha =0.01,0.05, 0.1$. This result is in line with the findings in \cite{Zhang2012} where the null hypothesis of a common trend is rejected at level $\alpha = 0.1$.
The main advantage of our multiscale test over the method in \cite{Zhang2012} is that it is much more informative. In particular, it provides information about \textit{which} of the $n=11$ countries have different trends and \textit{where} the trends differ. This information is presented graphically in Figures~\ref{fig:Australia:Norway}--\ref{fig:Australia:France}. Each figure corresponds to a specific pair of countries $(i, j)$ and is divided into three panels (a)--(c). 
Panel (a) shows the augmented time series $\{\widehat{Y}_{it}: 1 \le t \le T\}$ and $\{\widehat{Y}_{jt}: 1 \le t \le T\}$ for the two countries $i$ and $j$ that are compared. 
Panel (b) presents smoothed versions of the time series from (a), in particular, it shows local linear kernel estimates of the two trend functions $m_i$ and $m_j$, where the bandwidth is set to $14$ quarters (that is, to $0.1$ in terms of rescaled time) and an Epanechnikov kernel is used. Panel (c) presents the results produced by our test for the significance level $\alpha = 0.05$:
it depicts in grey the set $\mathcal{S}^{[i, j]}(\alpha)$ of all the intervals for which the test rejects the local null $H_0^{[i, j]}(u, h)$. The set of minimal intervals $\mathcal{S}^{[i, j]}_{min}(\alpha) \subseteq \mathcal{S}^{[i, j]}(\alpha)$ is highlighted in black. According to \eqref{eq:CS-v2}, we can make the following simultaneous confidence statement about the intervals plotted in panels (c) of Figures \ref{fig:Australia:Norway}--\ref{fig:Australia:France}: we can claim, with confidence of about $95\%$, that there is a difference between the functions $m_i$ and $m_j$ on each of these intervals.

\begin{sidewaysfigure}[p!]
\begin{minipage}[t]{0.24\textwidth}
\includegraphics[width=\textwidth]{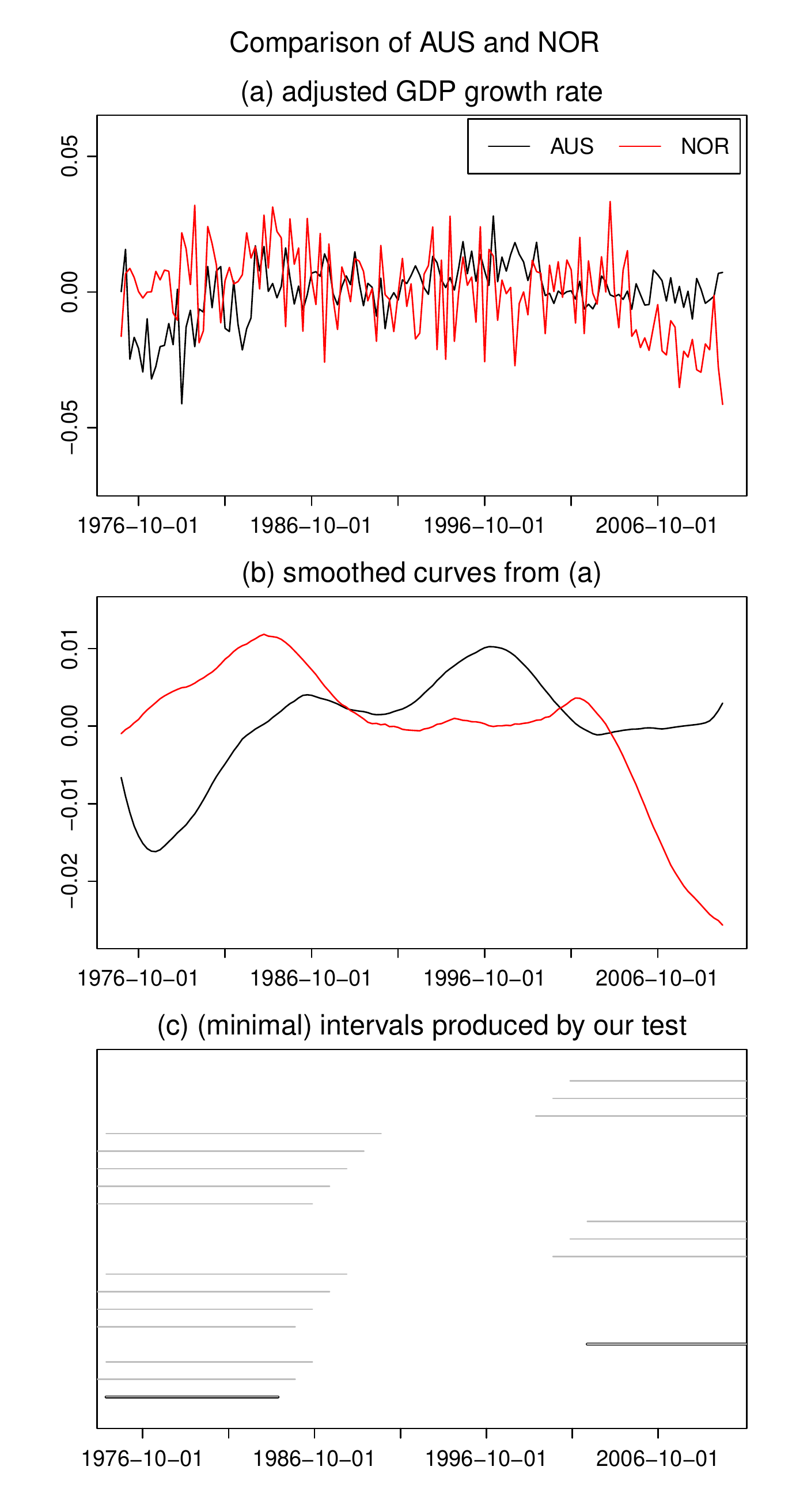}
\caption{Test results for the comparison of Australia and Norway.}\label{fig:Australia:Norway}
\end{minipage}
\hspace{0.1cm}
\begin{minipage}[t]{0.24\textwidth}
\includegraphics[width=\textwidth]{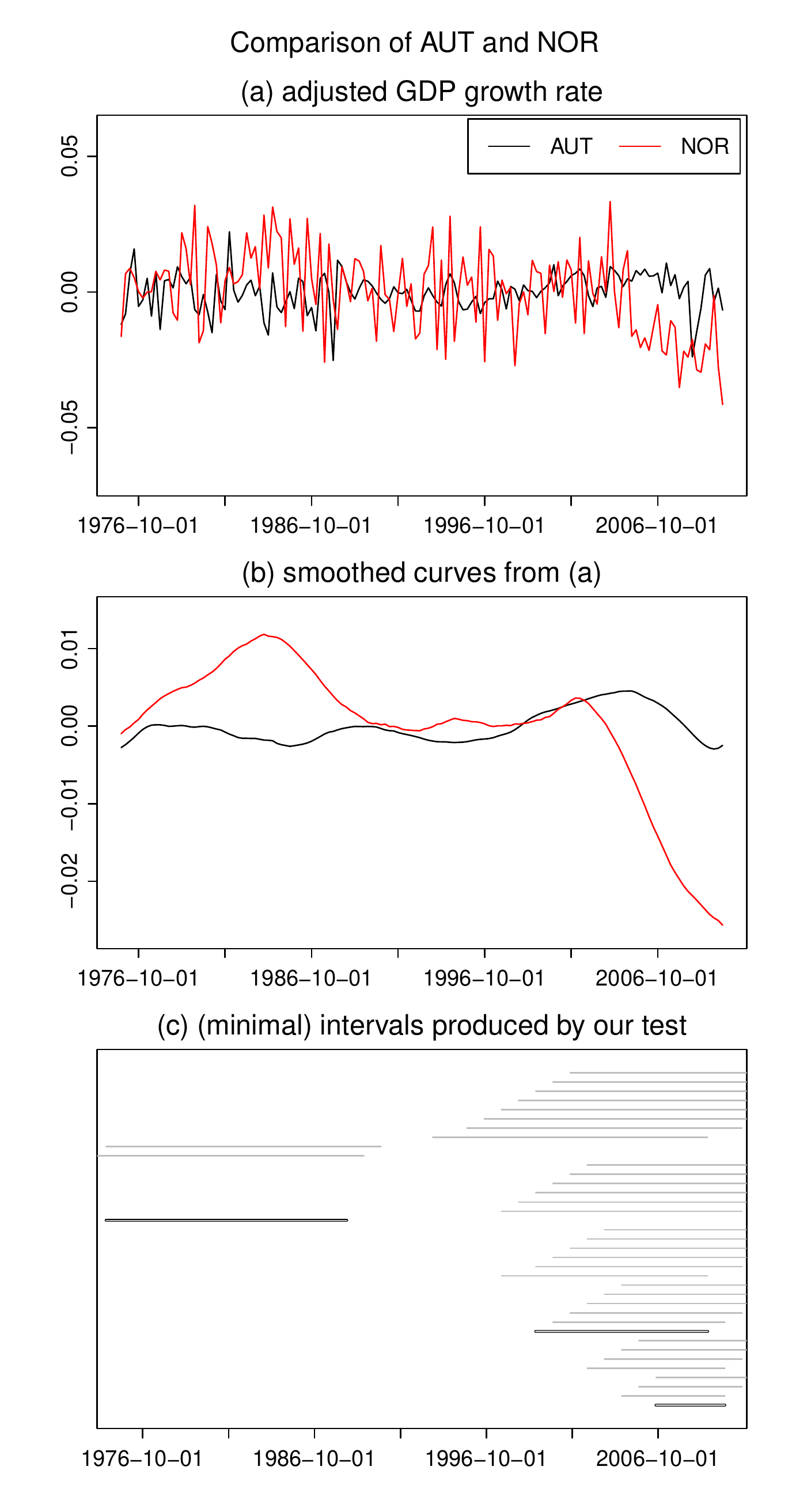}
\caption{Test results for the comparison of Austria and Norway.}\label{fig:Austria:Norway}
\end{minipage}
\hspace{0.1cm}
\begin{minipage}[t]{0.24\textwidth}
\includegraphics[width=\textwidth]{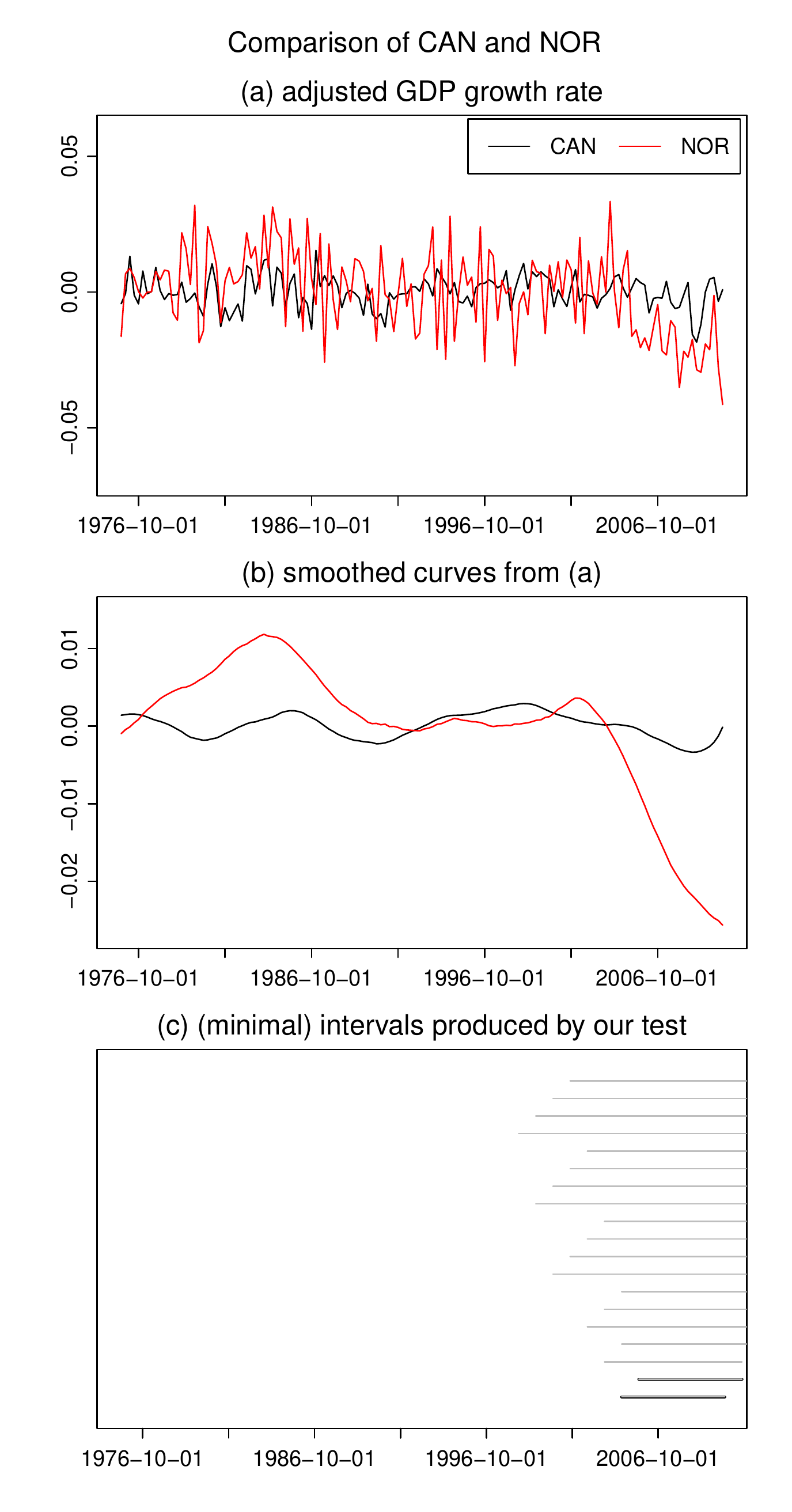}
\caption{Test results for the comparison of Canada and Norway.}\label{fig:Canada:Norway}
\end{minipage}
\hspace{0.1cm}
\begin{minipage}[t]{0.24\textwidth}
\includegraphics[width=\textwidth]{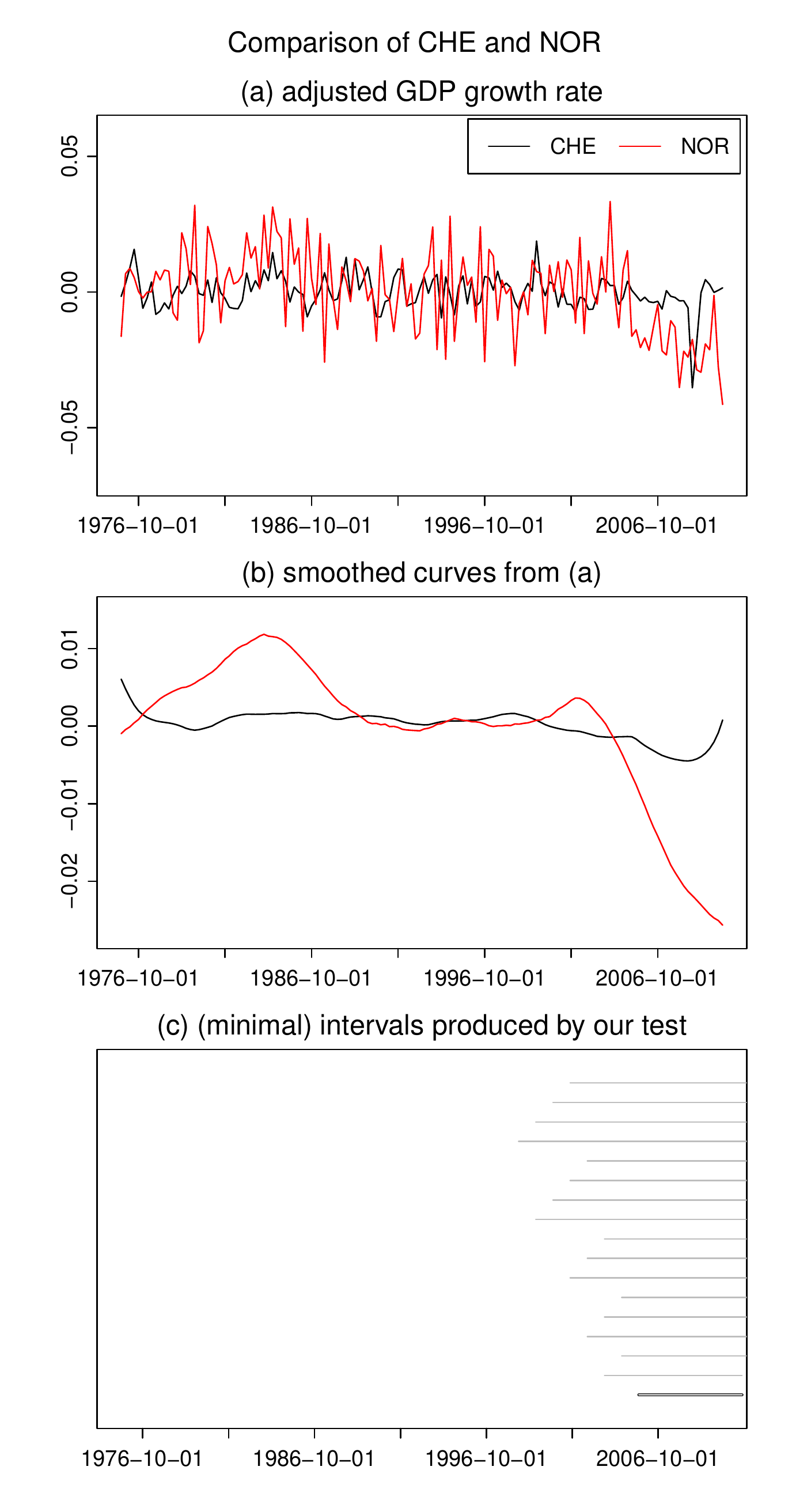}
\caption{Test results for the comparison of Switzerland and Norway.}\label{fig:Switzerland:Norway}
\end{minipage}
\caption*{Note: In each figure, panel (a) shows the two augmented time series, panel (b) presents smoothed versions of the augmented time series, and panel (c) depicts the set of intervals $\mathcal{S}^{[i, j]}(\alpha)$ in grey and the subset of minimal intervals $\mathcal{S}^{[i, j]}_{min}(\alpha)$ in black. }
\end{sidewaysfigure}

\begin{sidewaysfigure}[p!]
\begin{minipage}[t]{0.24\textwidth}
\includegraphics[width=\textwidth]{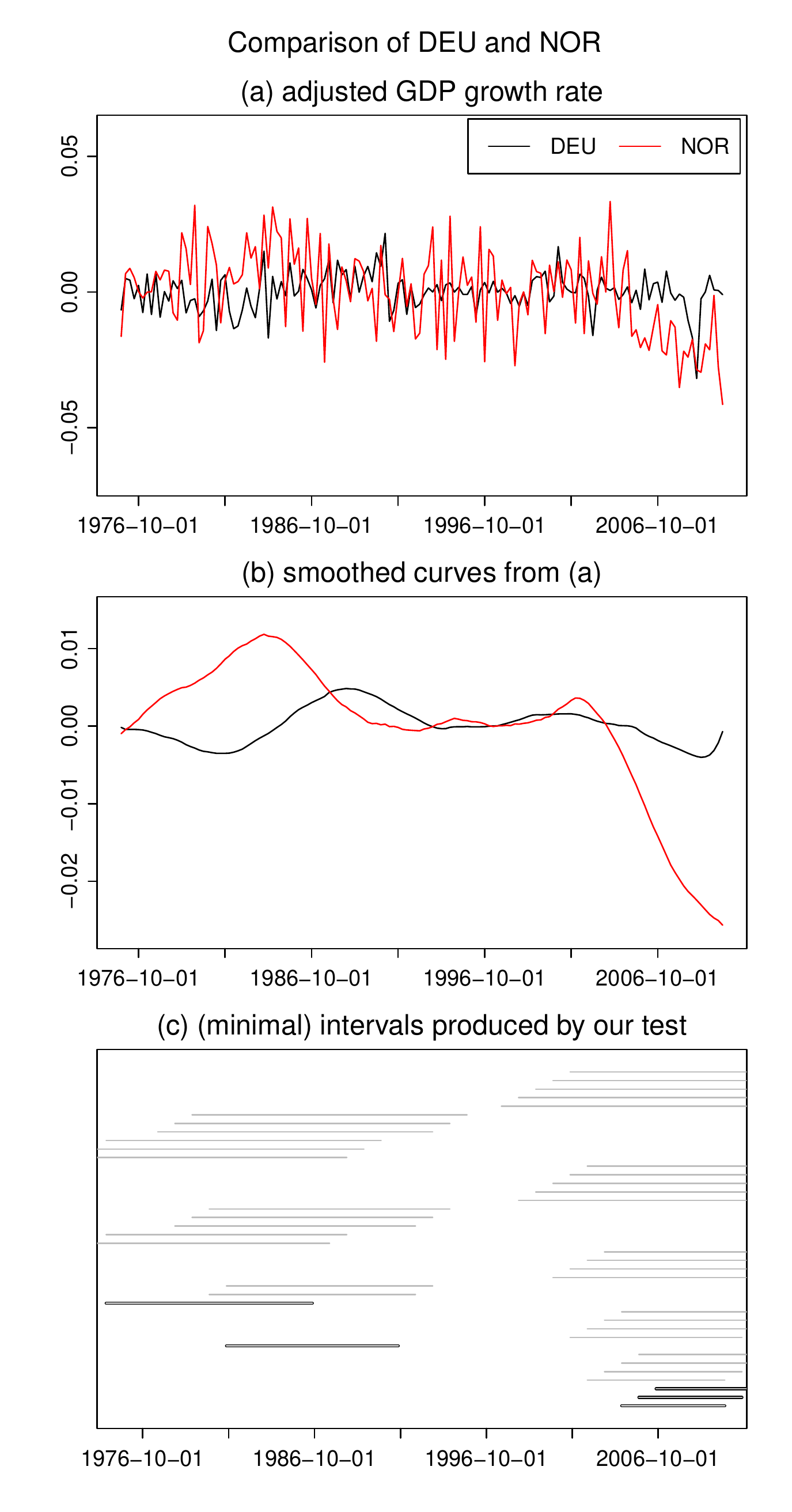}
\caption{Test results for the comparison of Germany and Norway.}\label{fig:Germany:Norway}
\end{minipage}
\hspace{0.1cm}
\begin{minipage}[t]{0.24\textwidth}
\includegraphics[width=\textwidth]{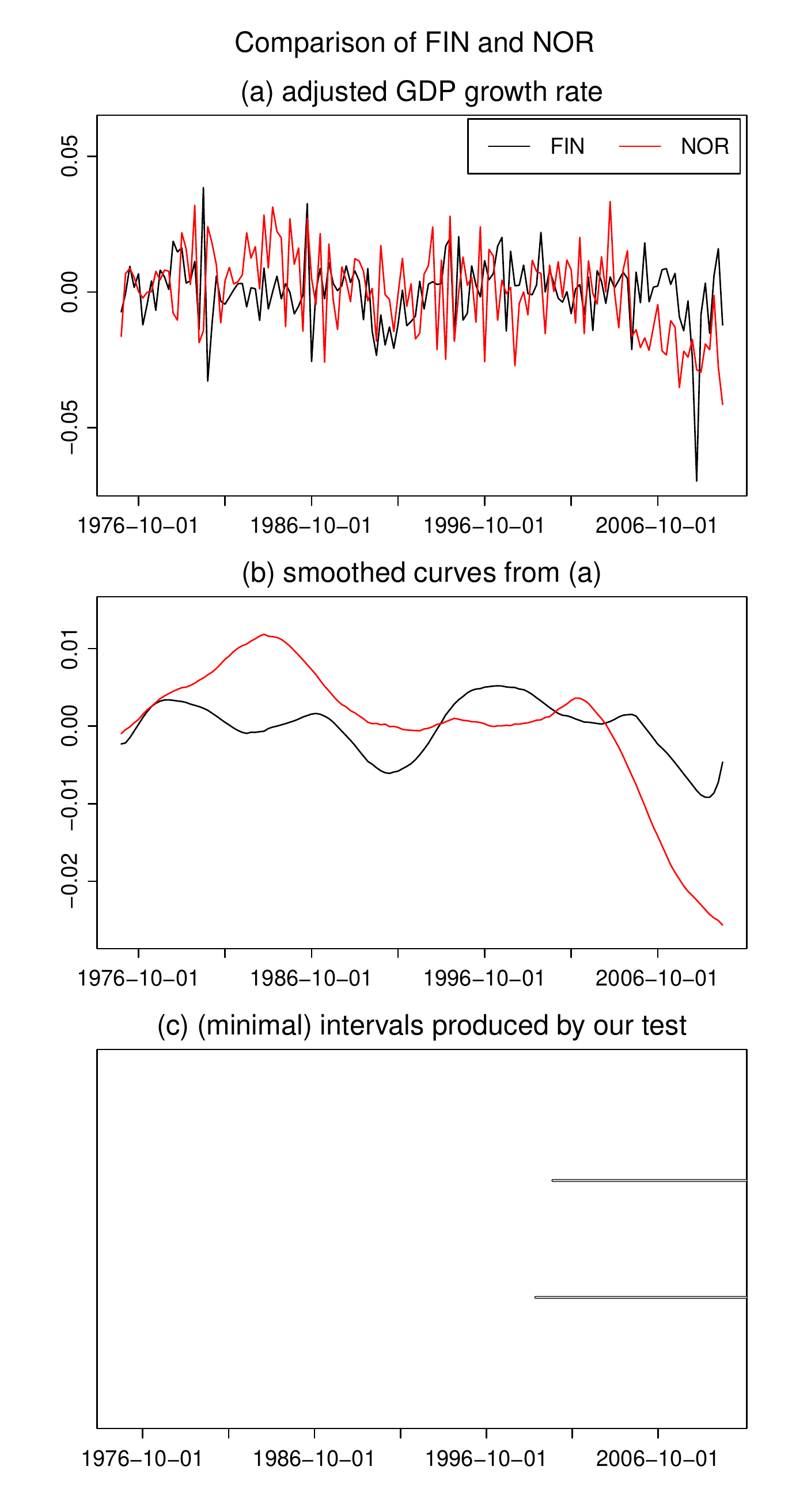}
\caption{Test results for the comparison of Finland and Norway.}\label{fig:Finland:Norway}
\end{minipage}
\hspace{0.1cm}
\begin{minipage}[t]{0.24\textwidth}
\includegraphics[width=\textwidth]{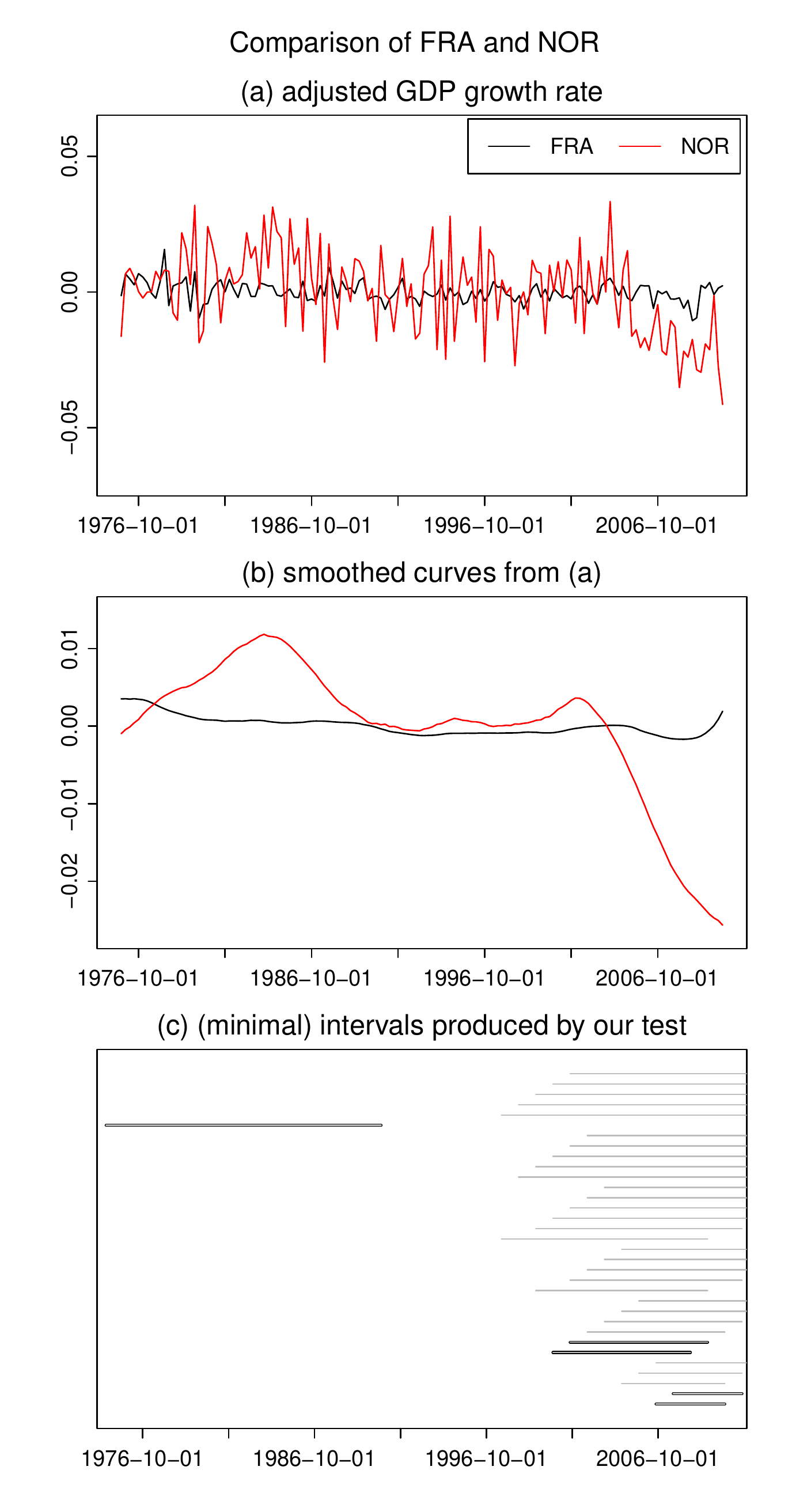}
\caption{Test results for the comparison of France and Norway.}\label{fig:France:Norway}
\end{minipage}
\hspace{0.1cm}
\begin{minipage}[t]{0.24\textwidth}
\includegraphics[width=\textwidth]{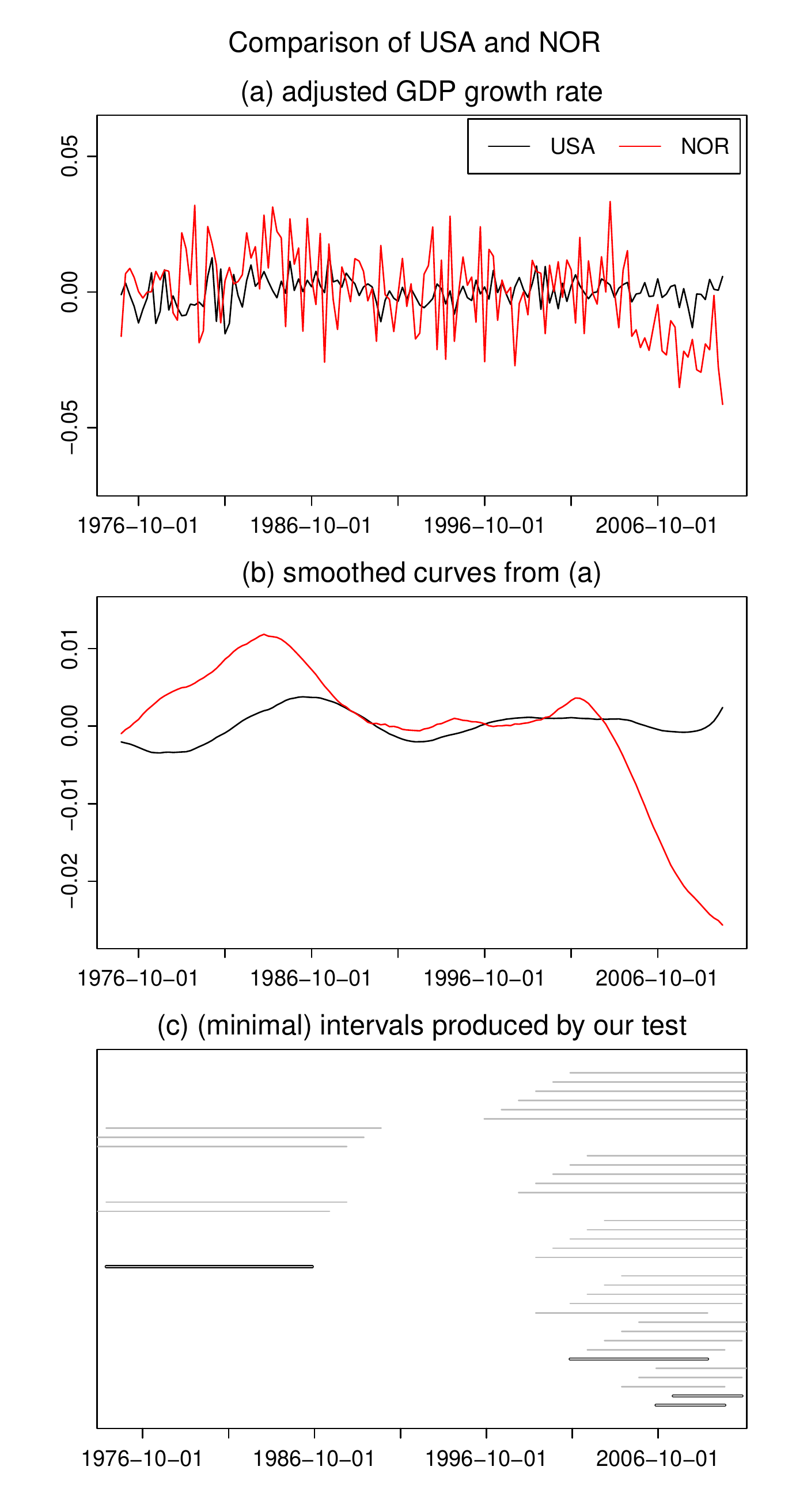}
\caption{Test results for the comparison of the USA and Norway.}\label{fig:USA:Norway}
\end{minipage}
\caption*{Note: In each figure, panel (a) shows the two augmented time series, panel (b) presents smoothed versions of the augmented time series, and panel (c) depicts the set of intervals $\mathcal{S}^{[i, j]}(\alpha)$ in grey and the subset of minimal intervals $\mathcal{S}^{[i, j]}_{min}(\alpha)$ in black.}
\end{sidewaysfigure}

\begin{sidewaysfigure}[p!]
\centering
\begin{minipage}[t]{0.24\textwidth}
\includegraphics[width=\textwidth]{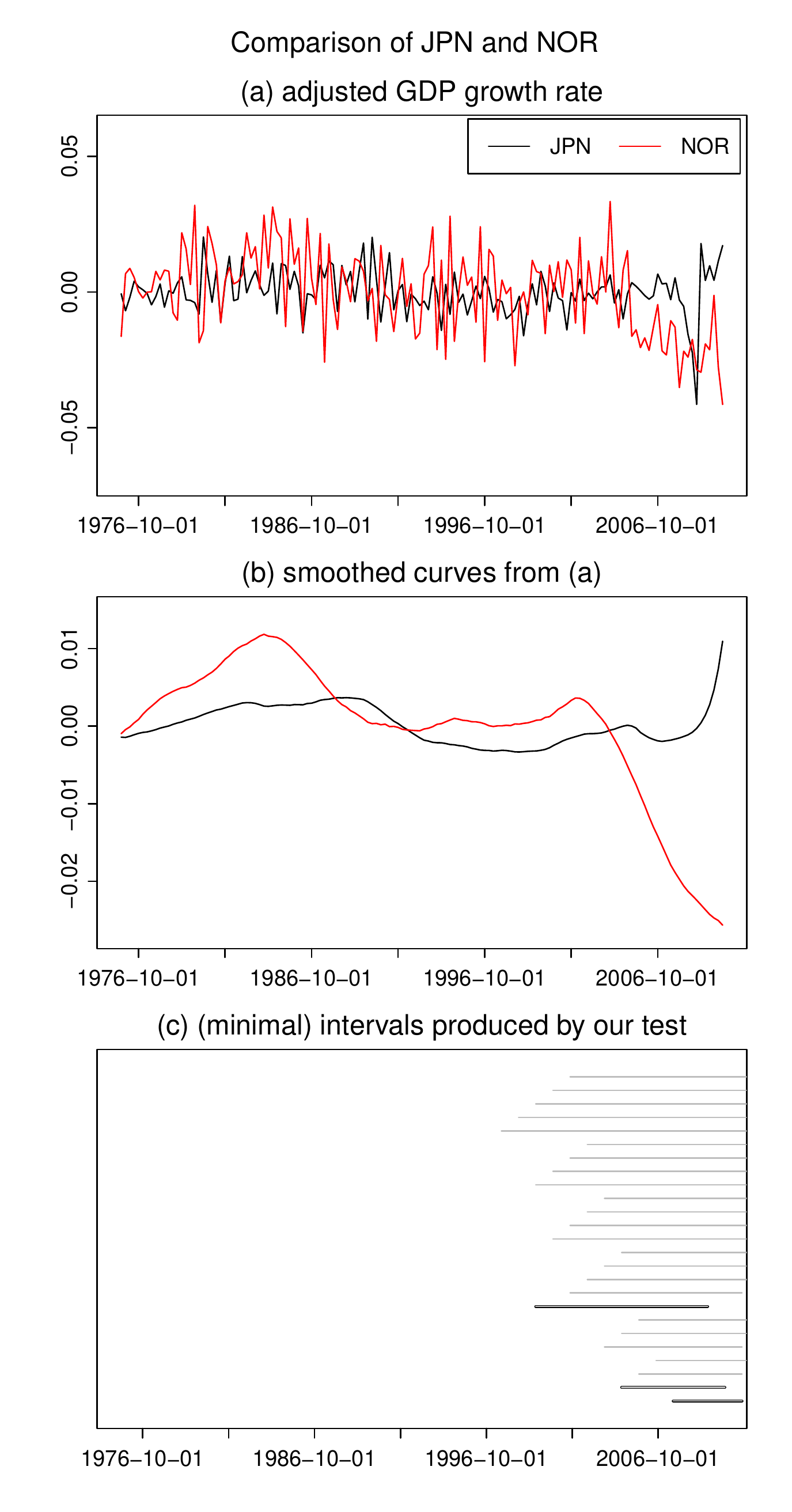}
\caption{Test results for the comparison of Japan and Norway.}\label{fig:Japan:Norway}
\end{minipage}
\hspace{0.1cm}
\begin{minipage}[t]{0.24\textwidth}
\includegraphics[width=\textwidth]{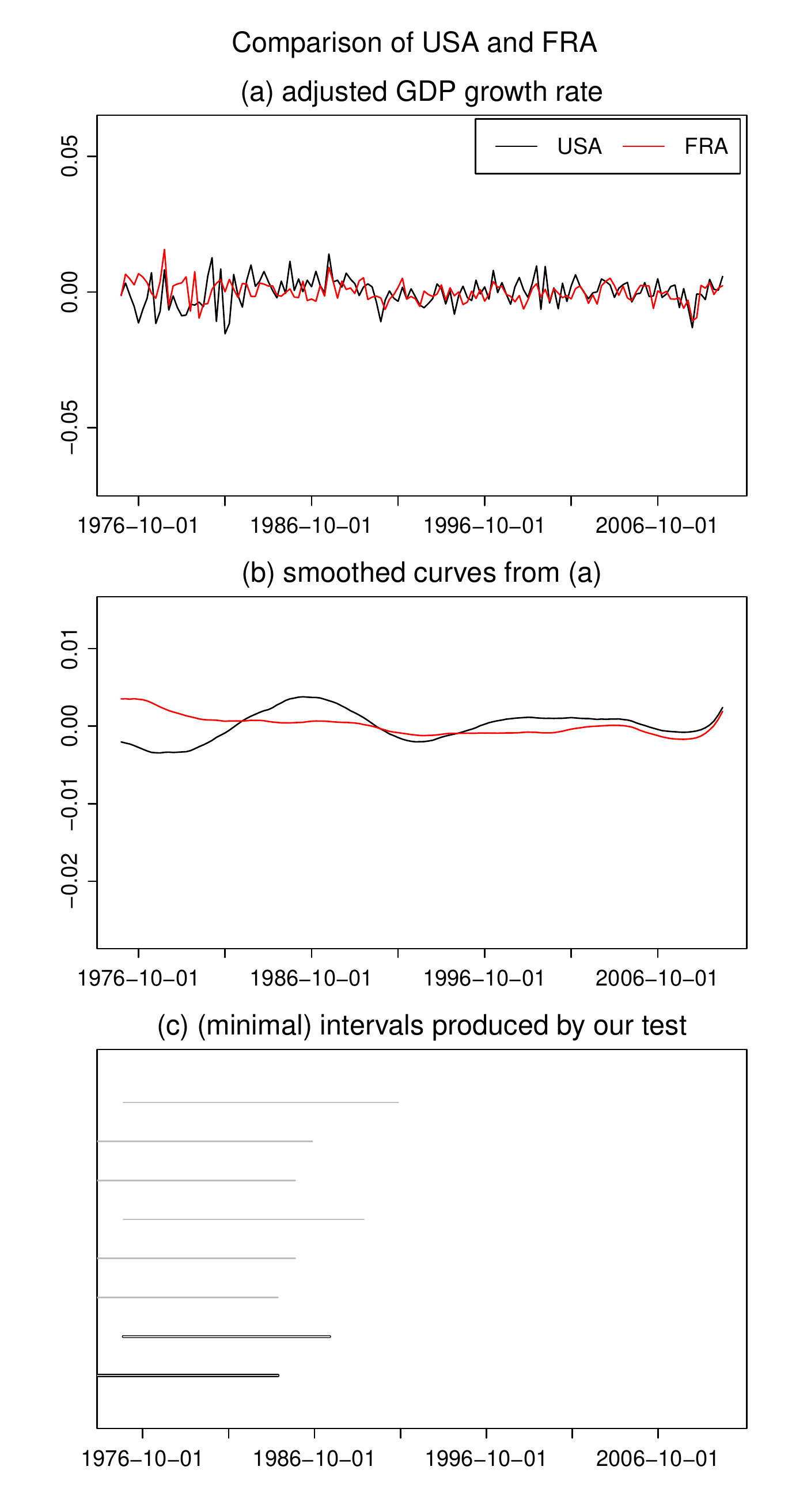}
\caption{Test results for the comparison of the USA and France.}\label{fig:USA:France}
\end{minipage}
\hspace{0.1cm}
\begin{minipage}[t]{0.24\textwidth}
\includegraphics[width=\textwidth]{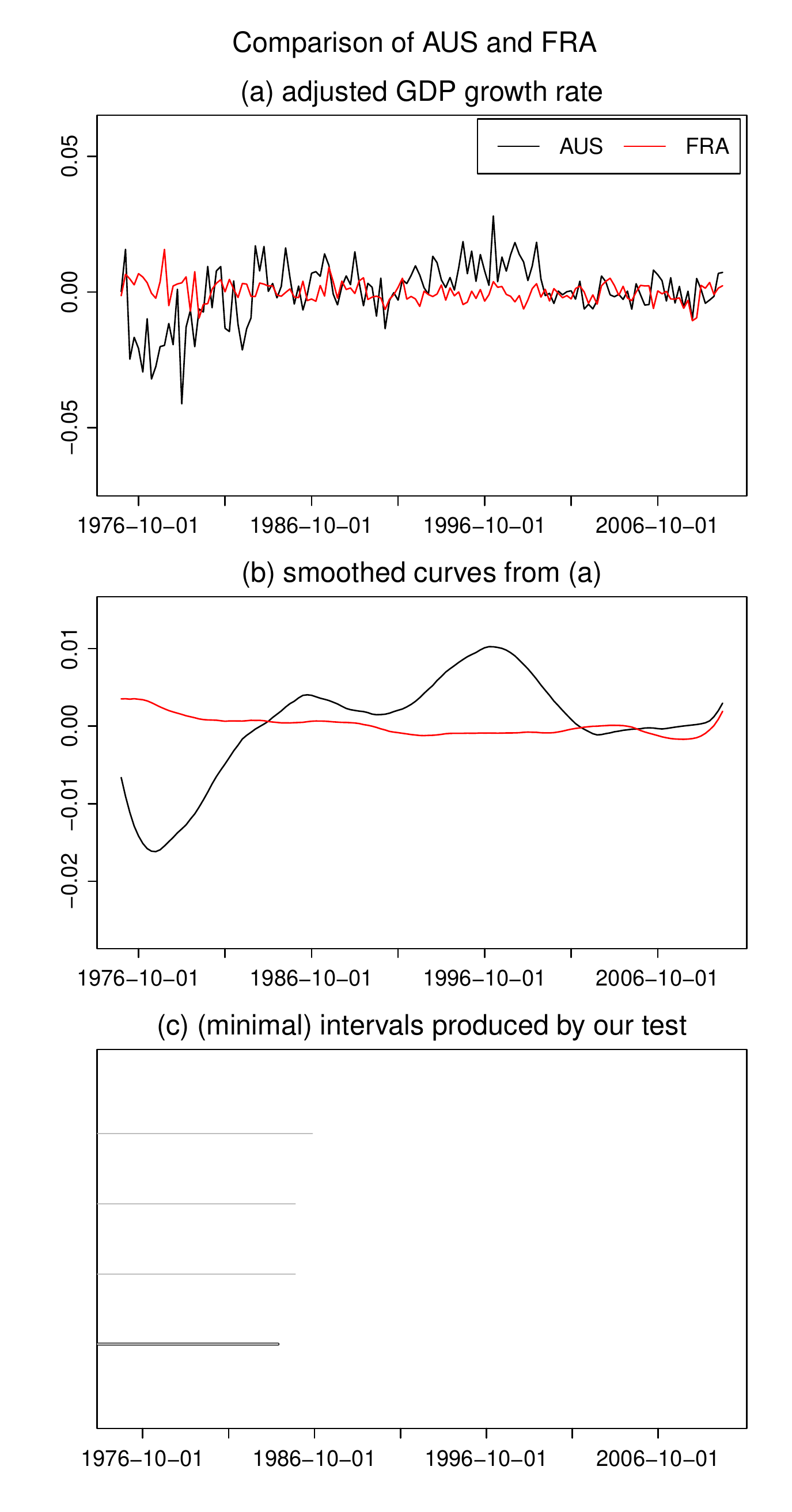}
\caption{Test results for the comparison of Australia and France.}\label{fig:Australia:France}
\end{minipage}
\caption*{Note: In each figure, panel (a) shows the two augmented time series, panel (b) presents smoothed versions of the augmented time series, and panel (c) depicts the set of intervals $\mathcal{S}^{[i, j]}(\alpha)$ in grey and the subset of minimal intervals $\mathcal{S}^{[i, j]}_{min}(\alpha)$ in black.}
\end{sidewaysfigure}

Out of $55$ pairwise comparisons, our test detects differences for $11$ pairs of countries $(i,j)$. These $11$ cases are presented in Figures~\ref{fig:Australia:Norway}--\ref{fig:Australia:France}. In $9$ cases (Figures \ref{fig:Australia:Norway}--\ref{fig:Japan:Norway}), one of the involved countries is Norway. Inspecting the trend estimates in panels (b) of Figures \ref{fig:Australia:Norway}--\ref{fig:Japan:Norway}, the Norwegian trend estimate can be seen to exhibit a strong downward movement at the end of the observation period, whereas the other trend estimates show a much less pronounced downward movement (or even a slight upward movement). According to our test, this is a significant difference between the Norwegian and the other trend functions rather than an artefact of the sampling noise: In all $9$ cases, the test rejects the local null for at least one interval which covers the last $10$ years of the analysed time period (from the first quarter in $2000$ up to the third quarter in $2010$). Apart from these differences at the end of the sampling period, our test also finds differences in the beginning, however, only for part of the pairwise comparisons.

Figures \ref{fig:USA:France} and \ref{fig:Australia:France} present the results of the pairwise comparison between Australia and France and between the USA and France, respectively. In both cases, our test detects differences between the GDP trends only in the beginning of the considered time period. In the case of Australia and France, it is clearly visible in the raw data (panel (a) in Figure~\ref{fig:Australia:France}) that there is a difference between the trends, whereas this is not so obvious in the case of the USA and France. According to our test, there are indeed significant differences in both cases. In particular, we can claim with confidence at least 95\%, that there are differences between the trends of the USA and France (of Australia and France) up to the fourth quarter in 1991 (the fourth quarter in 1986), but there is no evidence of any differences between the trends from 1992 (1987) onwards.

\begin{figure}[t!]
\begin{center}
\includegraphics[width=0.85\textwidth]{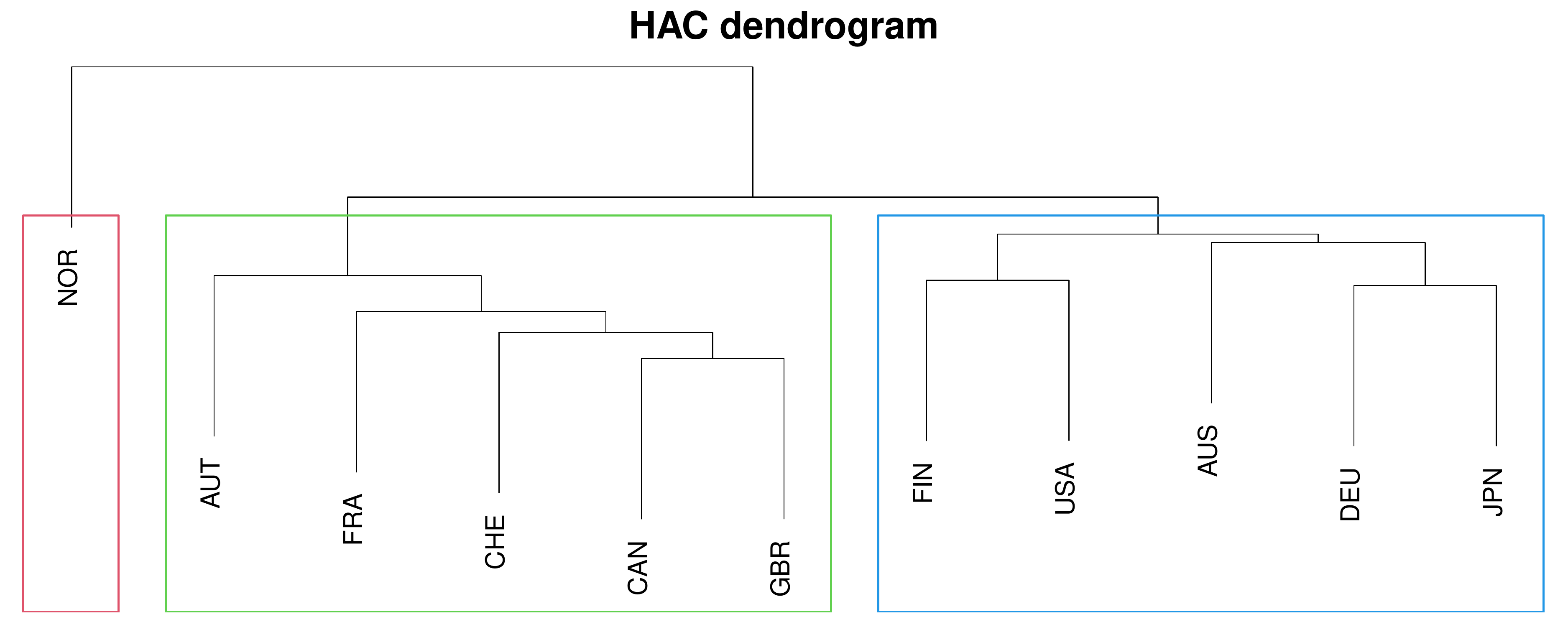}
\caption{Dendrogram of the HAC algorithm. Each coloured rectangle corresponds to one of the clusters.}\label{fig:gdp:dend}

\includegraphics[width=0.85\textwidth]{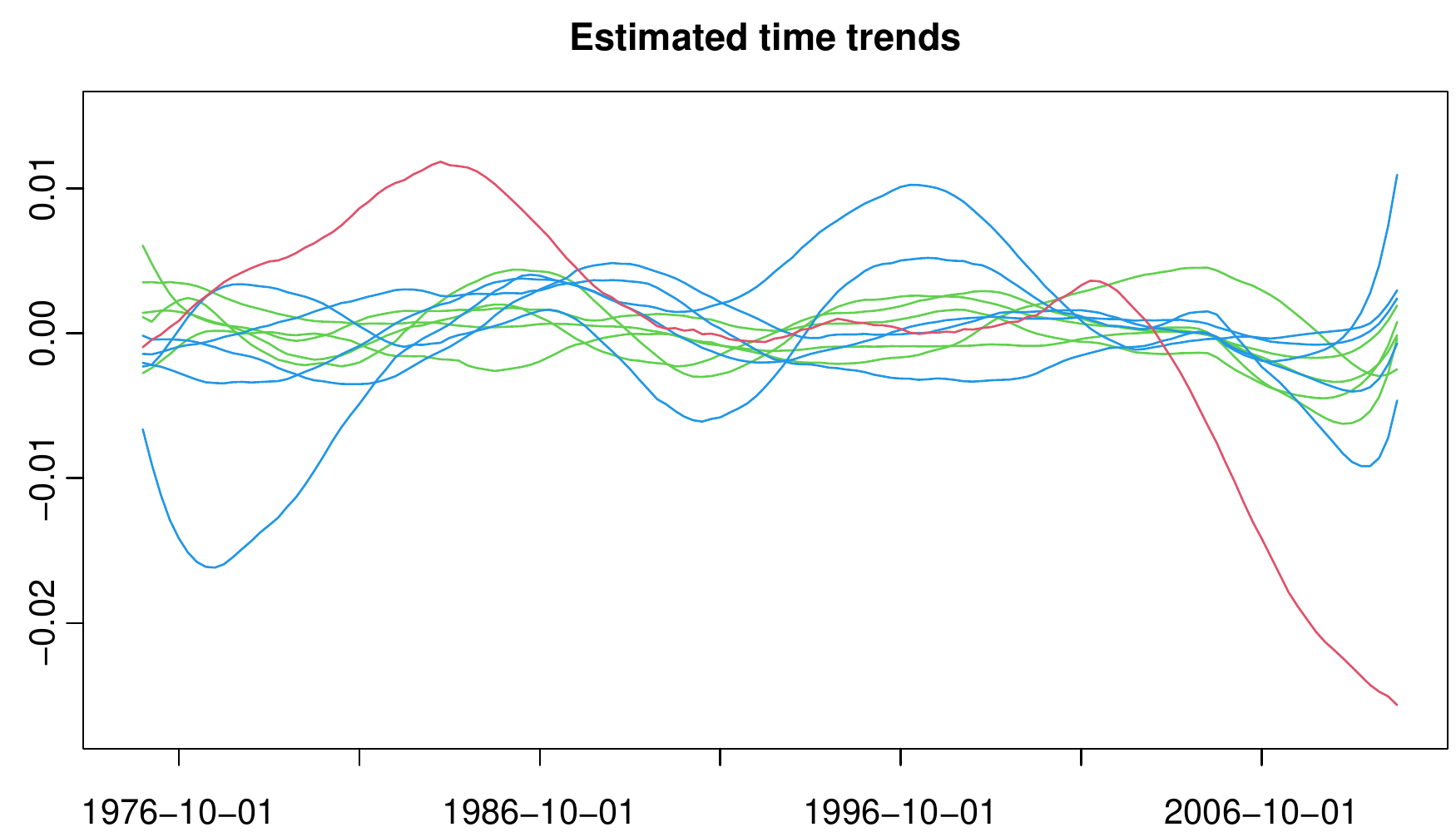}
\caption{Local linear estimates of the $n=11$ time trends (calculated from the augmented time series $\widehat{Y}_{it}$ with bandwidth $h = 0.1$ and Epanechnikov kernel). Each trend estimate is coloured according to the cluster that it is assigned to. }\label{fig:gdp:all_clusters}
\end{center}
\end{figure}

We next apply our clustering techniques to find groups of countries that have the same time trend. We implement our HAC algorithm with $\alpha = 0.05$ and the same choices as detailed above. The dendrogram that depicts the clustering results is plotted in Figure \ref{fig:gdp:dend}. The number of clusters is estimated to be $\widehat{N} = 3$. The rectangles in Figure \ref{fig:gdp:dend} indicate the $\widehat{N} = 3$ clusters. In particular, each rectangle is drawn around the branches of the dendrogram that correspond to one of the three clusters. Figure \ref{fig:gdp:all_clusters} depicts local linear kernel estimates of the $n=11$ GDP time trends (calculated from the augmented time series $\widehat{Y}_{it}$ with bandwidth $0.1$ and Epanechnikov kernel). Their colour indicates which cluster they belong to.

The results in Figures \ref{fig:gdp:dend} and \ref{fig:gdp:all_clusters} show that there is one cluster which consists only of Norway (plotted in red). As we have already discussed above and as becomes apparent from Figure \ref{fig:gdp:all_clusters}, the Norwegian trend exhibits a strong downward movement at the end of the sampling period, whereas the other trends show a much more moderate downward movement (if at all). This is presumably the reason why the clustering procedure puts Norway in a separate cluster. The algorithm further finds two other clusters, one consisting of the 5 countries Australia, Finland, Germany, Japan and the USA (plotted in blue in Figures \ref{fig:gdp:dend} and \ref{fig:gdp:all_clusters}) and the other one consisting of the 5 countries Austria, Canada, France, Switzerland and the UK (plotted in  green in Figures \ref{fig:gdp:dend} and \ref{fig:gdp:all_clusters}). Visual inspection of the trend estimates in Figure \ref{fig:gdp:all_clusters} suggests that the GDP time trends in the blue cluster exhibit more pronounced decreases and increases than the GDP time trends in the green cluster. Hence, overall, the clustering procedure appears to produce a reasonable grouping of the GDP trends.

\subsection{Analysis of house prices}\label{subsec:app:hp}

We next analyse a historical dataset on nominal annual house prices from \cite{Knoll2017} that contains data for $14$ advanced economies covering $143$ years from $1870$ to $2012$. In our analysis, we consider 8 countries (Australia, Belgium, Denmark, France, Netherlands, Norway, Sweden and USA) over the time period 1890--2012. The data for all these countries except one (Belgium) contain no missing values, and for Belgium there are only five missing observations\footnote{The missing values in the Belgium time series span five years during World War I.} which we impute by linear interpolation. The time series of the other $6$ countries contain more than $10$ missing values each, which is why we exclude them from our analysis.

We deflate the nominal house prices with the corresponding consumer price index (CPI) to obtain real house prices ($HP$). Variables that can potentially influence the average house prices are numerous, and there seems to be no general consensus about what the main determinants are. Possible determinants include, but are not limited to, demographic factors such as population growth (\citealt{Holly2010}, \citealt{Wang2014}, \citealt{Churchill2021}); fundamental economic factors such as real GDP (\citealt{Huang2013}, \citealt{Churchill2021}), interest rate and inflation (\citealt{Abelson2005}, \citealt{Otto2007}, \citealt{Huang2013}, \citealt{Jorda2015}); urbanisation (\citealt{Chen2011}, \citealt{Wang2017}); government subsidies and regulations (\citealt{Malpezzi1999}); stock markets (\citealt{Gallin2006}); etc. In our analysis, we focus on the following determinants of the average house prices: real GDP ($GDP$), population size ($POP$), long-term interest rate ($I$) and inflation ($INFL$) which is measured as change in CPI. Most other factors (such as government regulations, construction costs, and real wages) vary rather slowly over time and can be captured by time trend, fixed effects and slope heterogeneity.
Data for CPI, real GDP, population size and long-term interest rate are taken from the Jordà-Schularick-Taylor Macrohistory Database\footnote{See \cite{Jorda2017} for a detailed description of the variable construction.}, which is freely available at \url{http://www.macrohistory.net/data/} (accessed on 13 January 2022).

In summary, we observe a panel of $n = 8$ time series $\mathcal{T}_i = \{(Y_{it}, \X_{it}): 1 \le t \le T \}$ of length $T = 123$ for each country $i \in \{1,\ldots, 8\}$, where $Y_{it} = \ln HP_{it}$ and $\X_{it} = (\ln GDP_{it}, \ln POP_{it}, I_{it}, INFL_{it})^\top$. For each $i$, the time series $\mathcal{T}_i$ is assumed to follow the model $Y_{it} = m_i(t/T) + \bfbeta^\top_i \X_{it} + \alpha_i + \varepsilon_{it}$, or equivalently, 
\begin{equation}\label{eq:model:app4}
\ln HP_{it} =  m_i \Big( \frac{t}{T} \Big) + \beta_{i, 1} \ln GDP_{it} + \beta_{i, 2} \ln POP_{it} + \beta_{i, 3} I_{it} + \beta_{i, 4} INFL_{it} + \alpha_i + \varepsilon_{it}
\end{equation}
for $1 \le t \le T$, where $\bfbeta_i = (\beta_{i, 1}, \beta_{i, 2}, \beta_{i, 3}, \beta_{i, 4})^\top$ is a vector of unknown parameters, $m_i$ is a country-specific unknown nonparametric time trend and $\alpha_i$ is a fixed effect.

The inclusion of a nonparametric trend function $m_i$ in model \eqref{eq:model:app4} is supported by the literature. \cite{Ugarte2009}, for example, model the trend in average Spanish house prices by means of splines. \cite{Winter2022} include a long-term stochastic trend component when describing the dynamic behaviour of real house prices in $8$ advanced economies. Including a nonparametric trend function when modelling the evolution of house prices is also the main conclusion in \cite{Zhang2016}, where it is shown that the time series of logarithmic US house prices is trend-stationary, i.e., can be transformed into a stationary series by subtracting a deterministic trend.

We implement the multiscale test from Section \ref{sec:test} in the same way as in the previous application example with one minor modification: we let $\mathcal{G}_T = U_T \times H_T$ with 
\begin{align*}
U_T & = \big\{ u \in [0,1]: u = \textstyle{\frac{t}{T}} \text{ for some } t \in \naturals \big\} \\
H_T & = \big\{ h \in \big[ \textstyle{\frac{\log T}{T}}, \textstyle{\frac{1}{4}} \big]:  h = \textstyle{\frac{5t - 3}{T}} \text{ for some } t \in \naturals \big\}. 
\end{align*}
We thus take into account all locations $u$ on an equidistant grid $U_T$ with step length $1/T$ and all scales $h=2/T, 7/T, 12/T,\ldots$ with $\log T /T \le h \le 1/4$. This implies that each interval $\interval =[u-h,u+h]$ with $(u,h) \in \grid$ spans $5, 15, 25, \ldots$ years. The lower bound $\log T / T$ is motivated by Assumption \ref{C-h}. As in Section \ref{subsec:app:gdp}, we assume that for each $i$, the error process $\mathcal{E}_i = \{\varepsilon_{it}: 1 \leq t \leq T\}$ follows an AR($p_i$) model and we estimate the long-run variances $\sigma_i^2$ by the difference-based estimator from \cite{KhismatullinaVogt2020} with tuning parameters $q$ and $r$ equal to $15$ and $10$, respectively. 
We choose $p_i$ by minimizing the BIC. For $7$ out of $8$ countries the order $p_i$ determined by BIC\footnote{Applying other information criteria such as FPE, AIC and HQ yields exactly the same results in these cases.} is equal to $1$. For the sake of simplicity, we thus assume that $p_i = 1$ for all $i$.

\begin{sidewaysfigure}[p!]
\begin{minipage}[t]{0.24\textwidth}
\includegraphics[width=\textwidth]{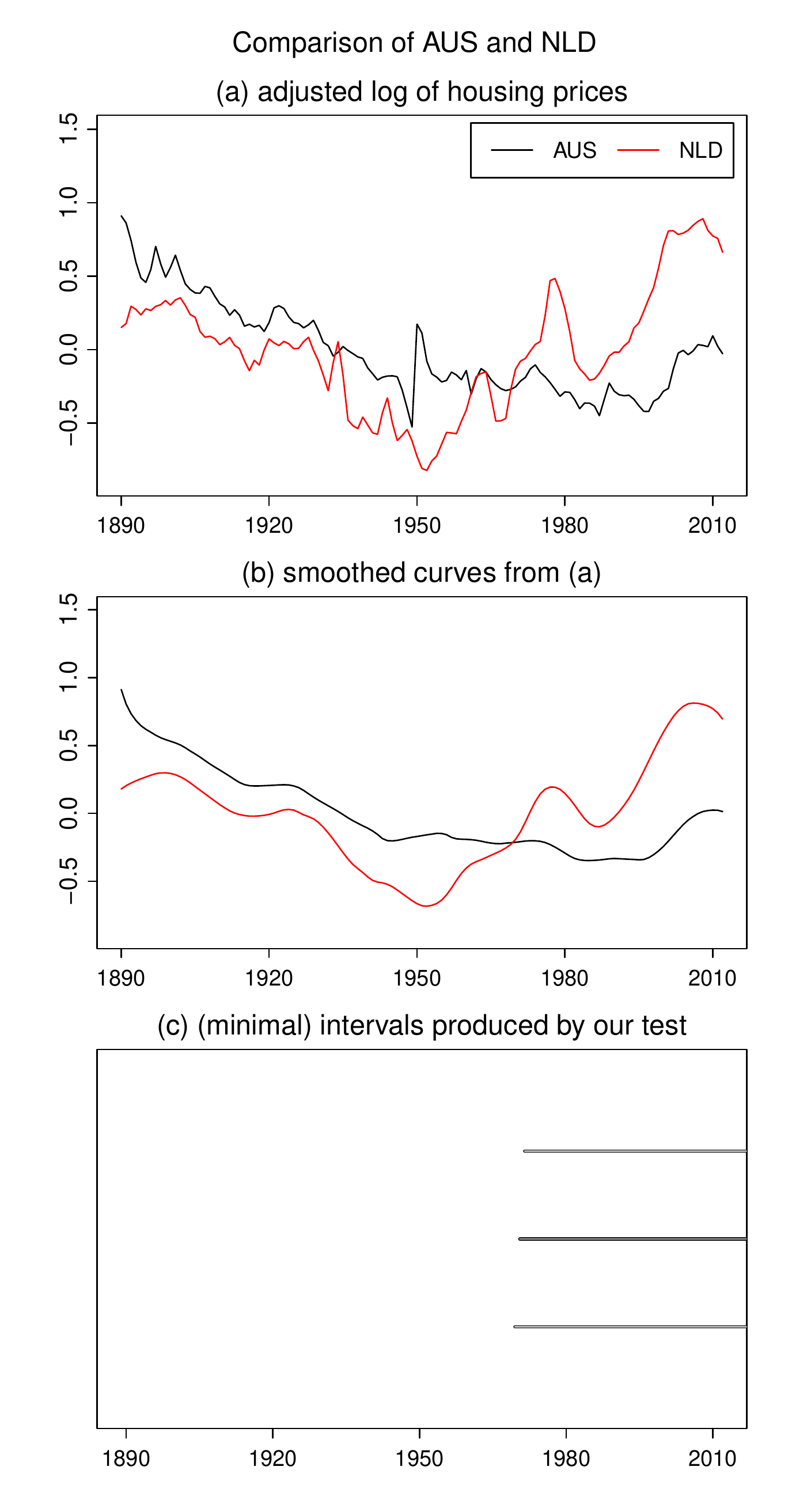}
\caption{Test results for the comparison of the house prices in Australia and the Netherlands.}\label{fig:hp:Australia:Netherlands}
\end{minipage}
\hspace{0.1cm}
\begin{minipage}[t]{0.24\textwidth}
\includegraphics[width=\textwidth]{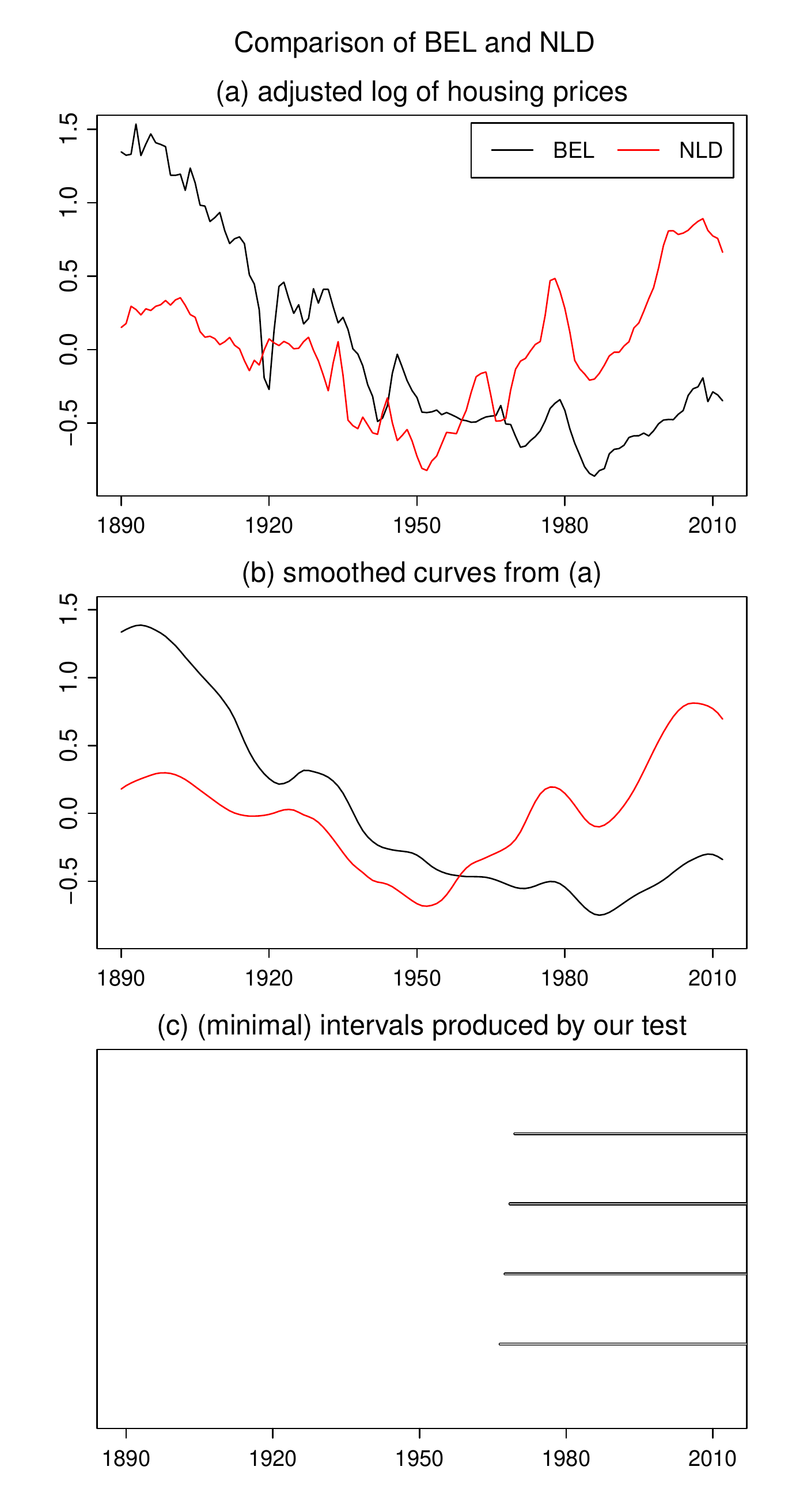}
\caption{Test results for the comparison of the house prices in Belgium and the Netherlands.}\label{fig:hp:Belgium:Netherlands}
\end{minipage}
\hspace{0.1cm}
\begin{minipage}[t]{0.24\textwidth}
\includegraphics[width=\textwidth]{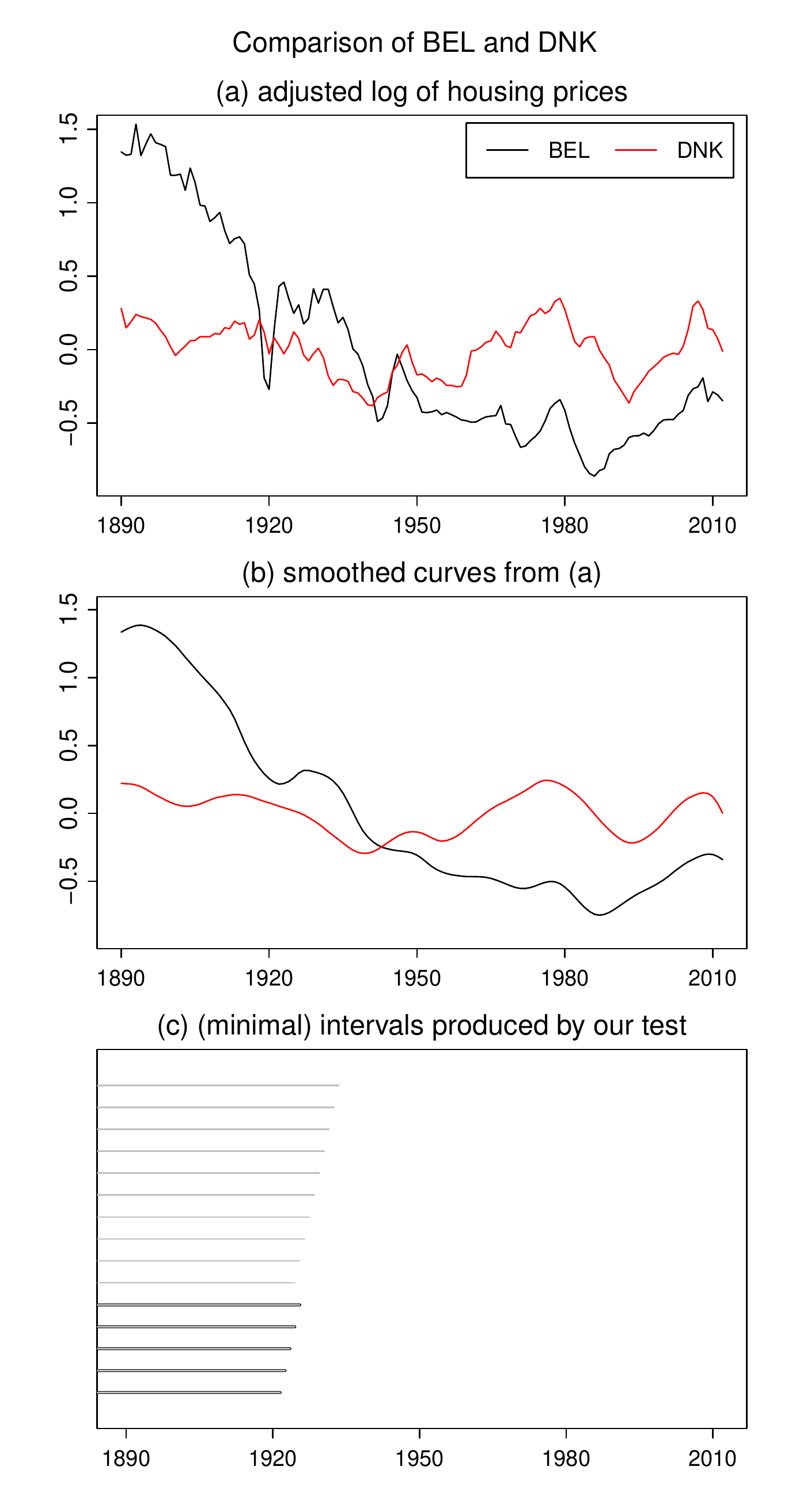}
\caption{Test results for the comparison of the house prices in Belgium and Denmark.}\label{fig:hp:Belgium:Denmark}
\end{minipage}
\hspace{0.1cm}
\begin{minipage}[t]{0.24\textwidth}
\includegraphics[width=\textwidth]{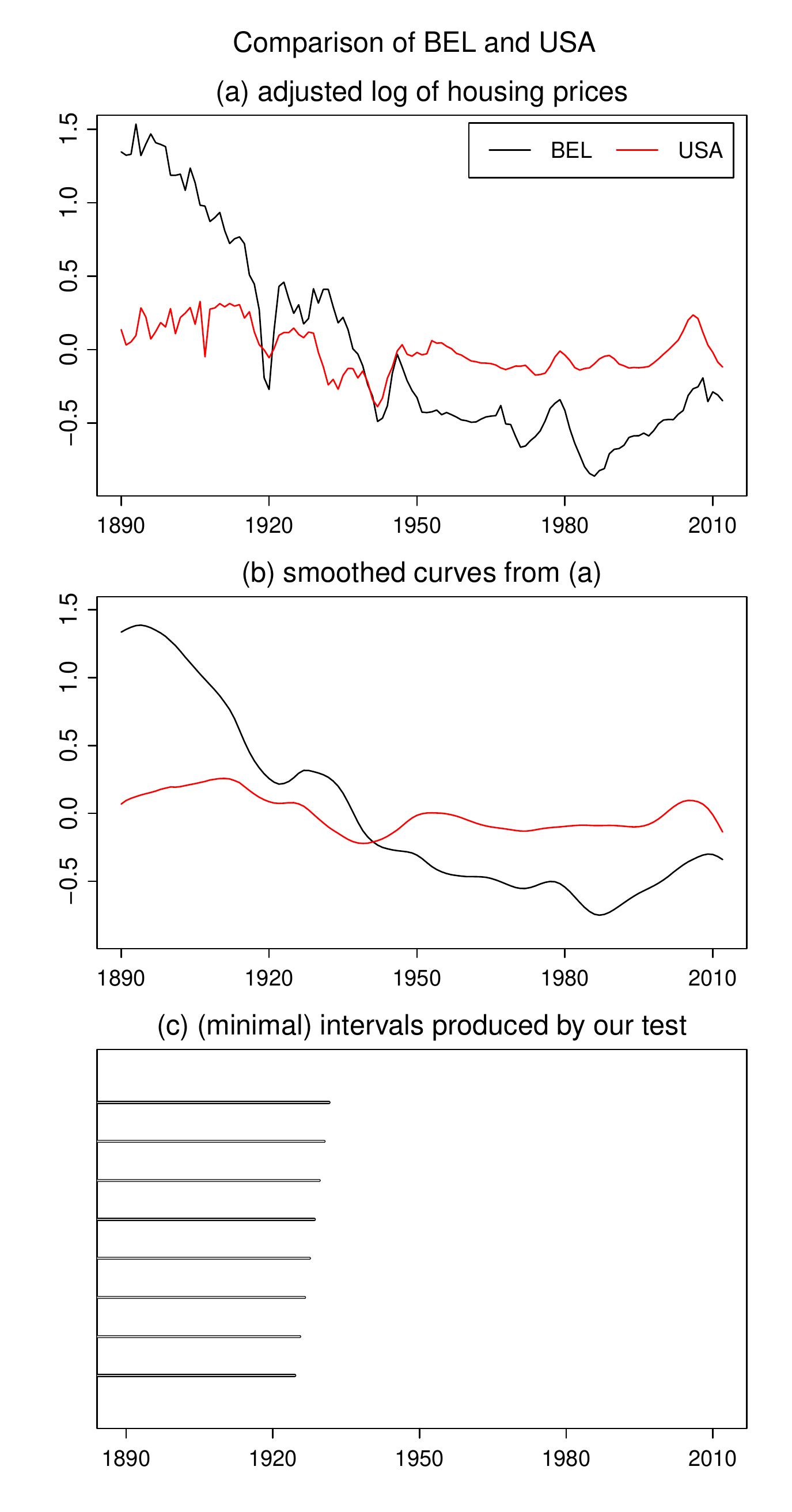}
\caption{Test results for the comparison of the house prices in Belgium and the USA.}\label{fig:hp:Belgium:USA}
\end{minipage}
\caption*{Note: In each figure, panel (a) shows the two augmented time series of the house prices, panel (b) presents smoothed versions of the augmented time series, and panel (c) depicts the set of intervals $\mathcal{S}^{[i, j]}(\alpha)$ in grey and the subset of minimal intervals $\mathcal{S}^{[i, j]}_{min}(\alpha)$ in black.}
\end{sidewaysfigure}

We are now ready to apply our test to the data. The overall null hypothesis $H_0$ is rejected at levels $\alpha = 0.05$ and $\alpha = 0.1$. The detailed test results for the significance level $\alpha =0.05$ are presented in Figures \ref{fig:hp:Australia:Netherlands}--\ref{fig:hp:Belgium:USA}. As in Section \ref{subsec:app:gdp}, each figure corresponds to the comparison of a pair of countries $(i, j)$ for which our test detects differences between the trends. Panel (a) shows the augmented time series $\{\widehat{Y}_{it}: 1 \le t \le T\}$ and $\{\widehat{Y}_{jt}: 1 \le t \le T\}$ for the two countries $i$ and $j$ under consideration. Panel (b) presents smoothed versions of the time series from (a) with the bandwidth window covering $15$ years. Panel (c) presents the test results for the level $\alpha=0.05$. As before, the set of intervals $\mathcal{S}^{[i, j]}(\alpha)$ for which our test rejects and the set of minimal intervals $\mathcal{S}^{[i, j]}_{min}(\alpha) \subseteq \mathcal{S}^{[i, j]}(\alpha)$ are depicted in grey and black, respectively. According to \eqref{eq:CS-v2}, we can make the following simultaneous confidence statement about the intervals plotted in panels (c) of Figures \ref{fig:hp:Australia:Netherlands}--\ref{fig:hp:Belgium:USA}: we can claim, with confidence of about $95\%$, that there is a difference between the trends $m_i$ and $m_j$ on each of these intervals.

Overall, our findings are in line with the observations in \cite{Knoll2017}, where the authors find considerable cross-country heterogeneity in the house price trends. The authors note that before World War II, the countries exhibit similar trends in real house prices, while the trends start to diverge sometime after World War II. This fits with our findings in Figures \ref{fig:hp:Australia:Netherlands} and \ref{fig:hp:Belgium:Netherlands} which show that our test detects differences between the trends of Australia and the Netherlands (of Belgium and the Netherlands) starting only from 1968 (1966) onwards. Contrary to \cite{Knoll2017}, however, our test also finds significant differences in the first half of the observed time period, specifically, between the time trends of Belgium and Denmark and of Belgium and the USA (Figures \ref{fig:hp:Belgium:Denmark} and \ref{fig:hp:Belgium:USA}, respectively). This discrepancy in the results is most certainly due to the fact that the method used in \cite{Knoll2017} does not account for effects of other factors such as GDP or population growth, while our test allows us to include various determinants of the average house prices in model \eqref{eq:model:app4}. 

\begin{figure}[t!]
\begin{center}
\includegraphics[width=0.85\textwidth]{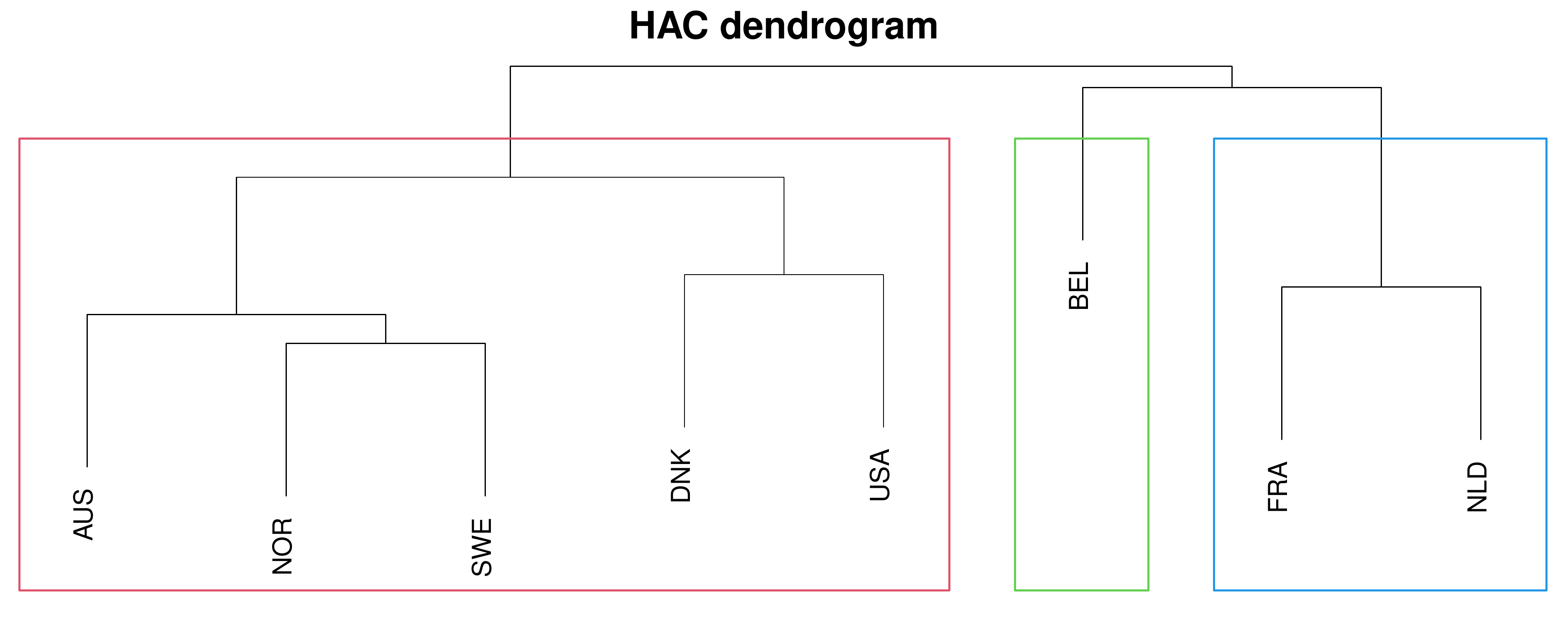}
\caption{Dendrogram of the HAC algorithm. Each coloured rectangle corresponds to one of the clusters.}\label{fig:hp:dend}

\includegraphics[width=0.85\textwidth]{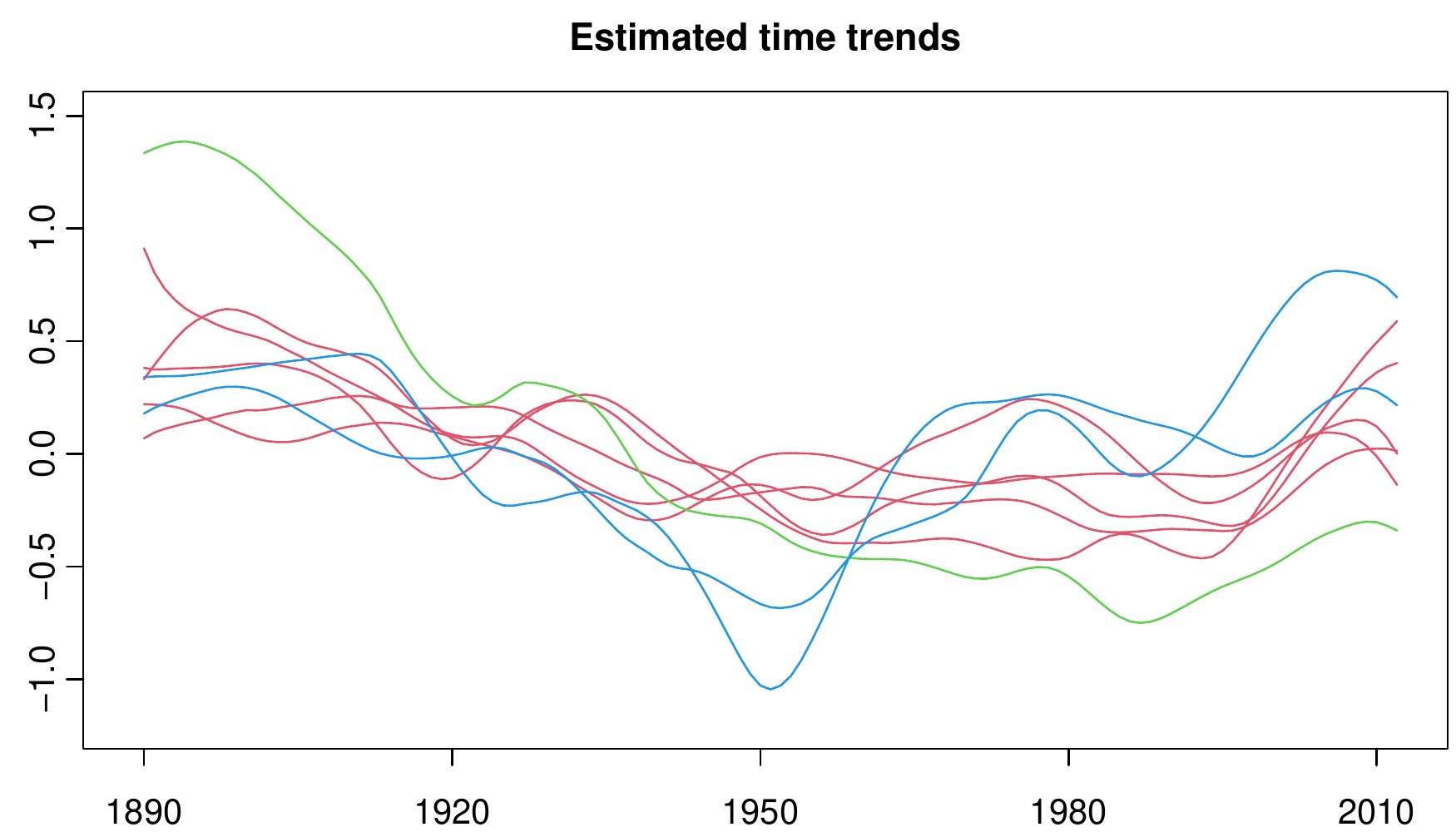}
\caption{Local linear estimates of the $n=8$ time trends (calculated from the augmented time series $\widehat{Y}_{it}$ with a bandwidth window covering $15$ years and an Epanechikov kernel). Each trend estimate is coloured according to the cluster that it is assigned to. }\label{fig:hp:all_clusters}
\end{center}
\end{figure}

We next apply the clustering procedure from Section \ref{sec:clustering} to the data. As in the previous application example, we set $\alpha = 0.05$. The results are displayed in Figures \ref{fig:hp:dend} and \ref{fig:hp:all_clusters}. Specifically, Figure \ref{fig:hp:dend} shows the dendrogram with the results of the HAC algorithm and Figure \ref{fig:hp:all_clusters} presents local linear kernel estimates of the trend curves (calculated with a bandwidth window of $15$ years and an Epanechnikov kernel). The number of clusters is estimated to be $\widehat{N} = 3$. As before, the coloured rectangles in Figure \ref{fig:hp:dend} are drawn around the countries that belong to the same cluster and the same colours are used to display the trend estimates in Figure \ref{fig:hp:all_clusters}.

\pagebreak
Inspecting the results, we can see that there is one cluster consisting only of Belgium (plotted in green). Figure \ref{fig:hp:all_clusters} suggests that the Belgium time trend indeed evolves somewhat differently from the other trends in the first $30$ years of the observed time period. 
The algorithm further detects a cluster that consists of two countries: France and the Netherlands (plotted in blue). The time trends of these two countries display some kind of dip around 1950, which is not present in the time trends of the other countries. Overall, our algorithm thus appears to produce a reasonable clustering of the house price trends.

\section*{Acknowledgements}

Financial support by the Deutsche Forschungsgemeinschaft (DFG, German Research Foundation), Germany – grant VO 2503/1-1, project number 430668955 – is gratefully acknowledged.

\bibliographystyle{ims}
{\small
\setlength{\bibsep}{0.55em}
\bibliography{bibliography}}

\newpage
\allowdisplaybreaks[3]
\appendix

\section{Appendix}\label{appendix}

In what follows, we prove the theoretical results from Sections \ref{sec:theo} and \ref{sec:clustering}. We use the following notation: The symbol $C$ denotes a universal real constant which may take a different value on each occurrence. For $a,b \in \reals$, we write $a \vee b = \max\{a,b\}$. For $x \in \reals_{\geq 0}$, we let $\lfloor x \rfloor$ denote the integer value of $x$ and $\lceil x \rceil$ the smallest integer greater than or equal to $x$. For any set $A$, the symbol $|A|$ denotes the cardinality of $A$. The expression $X \stackrel{\mathcal{D}}{=} Y$ means that the two random variables $X$ and $Y$ have the same distribution. Finally, we sometimes use the notation $a_T \ll b_T$ to express that $a_T = o(b_T)$.

\enlargethispage{0.25cm}
\subsection*{Auxiliary results}\label{subsec:appendix:aux}

Let $\{Z_t\}_{t=-\infty}^\infty$ be a stationary time series process with $Z_t \in \mathcal{L}^q$ for some $q > 2$ and $\ex[Z_t] = 0$. Assume that $Z_t$ can be represented as $Z_t = g(\ldots, \eta_{t-1}, \eta_t)$, where $\eta_t$ are i.i.d.\ variables and $g: \reals^\infty \to \reals$ is a measurable function. We first state a Nagaev-type inequality from \cite{Wu2016}.

\begin{definitionA}\label{defA-DAN} 
Let $q > 0$ and $\alpha > 0$. The dependence adjusted norm of the process $Z. = \{Z_t\}_{t=-\infty}^\infty$ is given by 
$\|Z.\|_{q, \alpha} = \sup_{t\geq 0} (t+1)^{\alpha} \sum_{s=t}^{\infty} \delta_{q}(g,s)$.
\end{definitionA}

\begin{propA}[\cite{Wu2016}, Theorem 2]\label{theo-wu2016}
Assume that $\|Z.\|_{q, \alpha} < \infty$ with $q > 2$ and $\alpha > 1/2 - 1/q$. Let $S_T = a_1 Z_1 + \ldots + a_T Z_T$, where $a_1,\ldots,a_T$ are real numbers with $\sum_{t=1}^T a_t^2 = T$. Then for any $w>0$,
\[ \pr(|S_T| \geq w) \leq C_1 \frac{|a|_{q}^{q}\|Z.\|^{q}_{q, \alpha}}{w^{q}} + C_2 \exp \left( - \frac{C_3 w^2} {T \|Z.\|^2_{2, \alpha}}\right), \]
where $C_1, C_2, C_3$ are constants that only depend on $q$ and $\alpha$.
\end{propA}

The following lemma is a simple consequence of the above inequality.

\begin{lemmaA}\label{lemma-wlln}
Let $\sum_{s=t}^\infty \delta_{q}(g,s) = O(t^{-\alpha})$ for some $q > 2$ and $\alpha > 1/2 - 1/q$. Then 
\[ \frac{1}{\sqrt{T}} \sum_{t=1}^T Z_t = O_p(1). \]
\end{lemmaA}

\begin{proof}[\textnormal{\textbf{Proof of Lemma \ref{lemma-wlln}.}}]
Let $\eta > 0$ be a fixed number. We apply Proposition \ref{theo-wu2016} to the sum $S_T = \sum_{t=1}^T a_t Z_t$ with $a_t = 1$ for all $t$. The assumption $\sum_{s=t}^\infty \delta_{q}(g,s) = O(t^{-\alpha})$ implies that $\|Z.\|_{2, \alpha} \le \|Z.\|_{q, \alpha} \le C_Z < \infty$. Hence, for $w$ chosen sufficiently large, we get 
\begin{align*}
\pr\Big( \Big| \sum_{t=1}^T Z_t \Big| \geq \sqrt{T} w\Big) 
 & \leq C_1 \frac{T C_Z^q}{T^{q/2} w^q} + C_2 \exp \left( - \frac{C_3 T w^2} {T C_Z^2}\right) \\
 & = \frac{\{C_1 C_Z^q\} T^{1-q/2}}{w^q} + C_2 \exp \left( - \frac{C_3 w^2} {C_Z^2}\right) \le \eta
\end{align*}
for all $T$. This means that $\sum_{t=1}^T Z_t / \sqrt{T} = O_p(1)$. 
\end{proof}

Let $\Delta \varepsilon_{it} = \varepsilon_{it} - \varepsilon_{it-1}$ and $\Delta \X_{it} = \X_{it} - \X_{it-1}$. By Assumptions \ref{C-err1} and \ref{C-reg1}, $\Delta \varepsilon_{it} = \Delta g_i(\mathcal{F}_{it})$ and $\Delta \X_{it} = \Delta \boldsymbol{h}_{i}(\mathcal{G}_{it})$. We further define
\begin{align*}
\boldsymbol{a}_{i}(\mathcal{H}_{it}) & := \Delta \boldsymbol{h}_{i}(\mathcal{G}_{it}) \Delta g_i(\mathcal{F}_{it})\phantom{^\top} = \Delta \X_{it} \Delta \varepsilon_{it} \\
\boldsymbol{b}_{i}(\mathcal{G}_{it}) & := \Delta \boldsymbol{h}_{i}(\mathcal{G}_{it}) \Delta \boldsymbol{h}_{i}(\mathcal{G}_{it})^\top = \Delta \X_{it} \Delta \X_{it}^\top,
\end{align*}
where $\boldsymbol{a}_i = (a_{ij})_{j=1}^d$, $\boldsymbol{b}_i = (b_{ikl})_{k,l=1}^d$ and $\mathcal{H}_{it} = (\mathcal{H}_{it,1},\ldots,\mathcal{H}_{it,d})^\top$ with $\mathcal{H}_{it,j} = (\ldots, \nu_{it-1,j},\linebreak \nu_{it,j})$ and $\nu_{it,j} = (\eta_{it},\xi_{it,j})$. The next result gives bounds on the physical dependence measures of the processes $\{ \boldsymbol{a}_{i}(\mathcal{H}_{it}) \}_{t=-\infty}^\infty$ and $\{ \boldsymbol{b}_{i}(\mathcal{G}_{it}) \}_{t=-\infty}^\infty$.

\begin{lemmaA}\label{lemma-bounds-dep-measure}
Let Assumptions \ref{C-err1}, \ref{C-err3}, \ref{C-reg1} and \ref{C-reg3} be satisfied. Then for each $i$, $j$, $k$ and $l$, it holds that
\begin{align*}
 & \sum_{s=t}^\infty \delta_p(a_{ij}, s) = O(t^{-\alpha}) \, \qquad \text{for } p = \min\{q,q^\prime\}/2 \text{ and some } \alpha > 1/2 - 1/p \\
 & \sum_{s=t}^\infty \delta_p(b_{ikl}, s) = O(t^{-\alpha}) \qquad \text{for } p = q^\prime/2 \text{ and some } \alpha > 1/2 - 1/p.
\end{align*}
\end{lemmaA}

\begin{proof}[\textnormal{\textbf{Proof of Lemma \ref{lemma-bounds-dep-measure}.}}]
We only prove the first statement. The second one follows by analogous arguments. By the definition of the physical dependence measure and the Cauchy-Schwarz inequality, we have with $p = \min\{q,q^\prime\}/2$ that
\begin{align*}
\delta_p(a_{ij}, s) 
 & = \| a_{ij}(\mathcal{H}_{it,j}) - a_{ij}(\mathcal{H}_{it,j}^\prime) \|_p \\
 &= \| \Delta h_{ij}(\mathcal{G}_{it}) \Delta g_i(\mathcal{F}_{it}) -  \Delta h_{ij}(\mathcal{G}_{it}^\prime) \Delta g_i(\mathcal{F}_{it}^\prime) \|_p \\
 &\leq \| h_{ij}(\mathcal{G}_{it}) g_i(\mathcal{F}_{it}) - h_{ij}(\mathcal{G}_{it}^\prime) g_i(\mathcal{F}_{it}^\prime) \|_p \\
 &\quad + \| h_{ij}(\mathcal{G}_{it-1}) g_i(\mathcal{F}_{it-1}) - h_{ij}(\mathcal{G}_{it-1}^\prime) g_i(\mathcal{F}_{it-1}^\prime) \|_p \\
 &\quad + \| h_{ij}(\mathcal{G}_{it-1}) g_i(\mathcal{F}_{it}) - h_{ij}(\mathcal{G}_{it-1}^\prime) g_i(\mathcal{F}_{it}^\prime) \|_p \\
 &\quad + \| h_{ij}(\mathcal{G}_{it}) g_i(\mathcal{F}_{it-1}) - h_{ij}(\mathcal{G}_{it}^\prime) g_i(\mathcal{F}_{it-1}^\prime) \|_p \\
 & = \| \{ h_{ij}(\mathcal{G}_{it}) - h_{ij}(\mathcal{G}_{it}^\prime)\} g_i(\mathcal{F}_{it}) + h_{ij}(\mathcal{G}_{it}^\prime) \{g_i(\mathcal{F}_{it}) - g_i(\mathcal{F}_{it}^\prime)\} \|_p \\
 & \quad + \| \{ h_{ij}(\mathcal{G}_{it-1}) - h_{ij}(\mathcal{G}_{it-1}^\prime)\} g_i(\mathcal{F}_{it-1}) + h_{ij}(\mathcal{G}_{it-1}^\prime) \{g_i(\mathcal{F}_{it-1}) - g_i(\mathcal{F}_{it-1}^\prime)\} \|_p \\
 &\quad + \| \{ h_{ij}(\mathcal{G}_{it-1}) - h_{ij}(\mathcal{G}_{it-1}^\prime)\} g_i(\mathcal{F}_{it}) + h_{ij}(\mathcal{G}_{it-1}^\prime) \{g_i(\mathcal{F}_{it}) - g_i(\mathcal{F}_{it}^\prime)\} \|_p \\
 &\quad + \| \{ h_{ij}(\mathcal{G}_{it}) - h_{ij}(\mathcal{G}_{it}^\prime)\} g_i(\mathcal{F}_{it-1}) + h_{ij}(\mathcal{G}_{it}^\prime) \{g_i(\mathcal{F}_{it-1}) - g_i(\mathcal{F}_{it-1}^\prime)\} \|_p \\
 &\leq \delta_{2p}(h_{ij}, t) \| g_i(\mathcal{F}_t) \|_{2p} + \delta_{2p} (g_i, t) \| h_{ij}(\mathcal{G}_{it}^\prime) \|_{2p} \\
&\quad + \delta_{2p}(h_{ij}, t-1) \| g_i(\mathcal{F}_{t-1}) \|_{2p} + \delta_{2p} (g_i, t-1) \| h_{ij}(\mathcal{G}_{it-1}^\prime) \|_{2p} \\
&\quad + \delta_{2p}(h_{ij}, t-1) \| g_i(\mathcal{F}_t) \|_{2p} + \delta_{2p} (g_i, t) \| h_{ij}(\mathcal{G}_{it-1}^\prime) \|_{2p} \\
&\quad + \delta_{2p}(h_{ij}, t) \| g_i(\mathcal{F}_{t-1}) \|_{2p} + \delta_{2p} (g_i, t-1) \| h_{ij}(\mathcal{G}_{it}^\prime) \|_{2p}, 
\end{align*}
where $\mathcal{H}_{it,j}^\prime  = (\ldots, \nu_{i(-1),j}, \nu^\prime_{i0,j}, \nu_{i1,j}, \ldots, \nu_{it-1,j}, \nu_{it,j})$, $\mathcal{G}_{it,j}^\prime  = (\ldots, \xi_{i(-1),j}, \xi^\prime_{i0,j}, \xi_{i1,j}, \ldots, \linebreak \xi_{it-1,j}, \xi_{it,j})$ and $\mathcal{F}_{it}^\prime  = (\ldots, \eta_{i(-1)}, \eta^\prime_{i0}, \eta_{i1}, \ldots, \eta_{it-1}, \eta_{it})$ are coupled processes with $\nu_{i0,j}^\prime$, $\xi_{i0,j}^\prime$ and $\eta_{i0,j}^\prime$ being i.i.d.\ copies of $\nu_{i0,j}$, $\xi_{i0,j}$ and $\eta_{i0}$. From this and Assumptions \ref{C-err1}, \ref{C-err3}, \ref{C-reg1} and \ref{C-reg3}, it immediately follows that $\sum_{s=t}^\infty \delta_p(a_{ij}, s) = O(t^{-\alpha})$.
\end{proof}

We now show that the estimator $\widehat{\bfbeta}_i$ is $\sqrt{T}$-consistent for each $i$ under our conditions.

\begin{lemmaA}\label{lemma-beta-rate}
Let Assumptions \ref{C-err1}, \ref{C-err3} and \ref{C-reg1}--\ref{C-reg-err} be satisfied. Then for each $i$, it holds that
\[ \widehat{\bfbeta}_i - \bfbeta_i = O_p\Big(\frac{1}{\sqrt{T}}\Big). \]
\end{lemmaA}

\begin{proof}[\textnormal{\textbf{Proof of Lemma \ref{lemma-beta-rate}.}}]
The estimator $\widehat{\bfbeta}_i$ can be written as
\begin{align*}
\widehat{\bfbeta}_i &= \Big( \sum_{t=2}^T \Delta \X_{it} \Delta \X_{it}^\top \Big)^{-1} \sum_{t=2}^T \Delta \X_{it} \Delta Y_{it} \\
& =  \Big( \sum_{t=2}^T \Delta \X_{it} \Delta \X_{it}^\top \Big)^{-1} \sum_{t=2}^T \Delta \X_{it} \bigg(\Delta \X_{it}^\top \bfbeta_i +  \Delta m_{it}+ \Delta \varepsilon_{it} \bigg) \\
&= \bfbeta_i + \Big( \sum_{t=2}^T \Delta \X_{it} \Delta \X_{it}^\top \Big)^{-1} \sum_{t=2}^T \Delta \X_{it} \Delta m_{it} +  \Big( \sum_{t=2}^T \Delta \X_{it} \Delta \X_{it}^\top \Big)^{-1} \sum_{t=2}^T \Delta \X_{it} \Delta \varepsilon_{it}, 
\end{align*}
where $\Delta \X_{it} = \X_{it} - \X_{it-1}$, $\Delta \varepsilon_{it} = \varepsilon_{it} - \varepsilon_{it-1}$ and $\Delta m_{it} = m_i (\frac{t}{T}) - m_i(\frac{t-1}{T})$. Hence, 
\begin{align}
 \sqrt{T}( \widehat{\bfbeta}_i - \bfbeta_i) = &\Big( \frac{1}{T}\sum_{t=2}^T \Delta \X_{it} \Delta \X_{it}^\top \Big)^{-1} \frac{1}{\sqrt{T}}\sum_{t=2}^T \Delta \X_{it} \Delta m_{it} \nonumber \\
&\quad+  \Big(\frac{1}{T} \sum_{t=2}^T \Delta \X_{it} \Delta \X_{it}^\top \Big)^{-1}\frac{1}{\sqrt{T}} \sum_{t=2}^T \Delta \X_{it} \Delta \varepsilon_{it}. \label{theo:beta:proof1}
\end{align}
In what follows, we show that 
\begin{align}
\frac{1}{\sqrt{T}} \sum_{t=2}^T \Delta \X_{it} \Delta \varepsilon_{it} & = O_p(1) \label{lemma-beta-rate-claim1} \\
\Big( \frac{1}{T}\sum_{t=2}^T \Delta \X_{it} \Delta \X_{it}^\top \Big)^{-1} & = O_p(1) \label{lemma-beta-rate-claim2} \\
\frac{1}{\sqrt{T}} \sum_{t=2}^T \Delta \X_{it} \Delta m_{it} & = O_p\Big(\frac{1}{\sqrt{T}}\Big). \label{lemma-beta-rate-claim3}
\end{align} 
Lemma \ref{lemma-beta-rate} follows from applying these three statements together with standard arguments to formula \eqref{theo:beta:proof1}.

Since $\ex[\Delta \X_{it} \Delta \varepsilon_{it}] = 0$ by \ref{C-reg-err} and $\sum_{s=t}^\infty \delta_p(a_{ij}, s) = O(t^{-\alpha})$ for some $p > 2$, $\alpha > 1/2 - 1/p$ and all $j$ by Lemma \ref{lemma-bounds-dep-measure}, the claim \eqref{lemma-beta-rate-claim1} follows upon applying Lemma \ref{lemma-wlln}. Another application of Lemma \ref{lemma-wlln} yields that 
\[ \frac{1}{T}\sum_{t=2}^T \Big\{ \Delta \X_{it} \Delta \X_{it}^\top - \ex[\Delta \X_{it} \Delta \X_{it}^\top] \Big\} = O_p\Big(\frac{1}{\sqrt{T}}\Big). \]
As $\ex[\Delta \X_{it} \Delta \X_{it}^\top]$ is invertible, we can invoke Slutsky's lemma to obtain \eqref{lemma-beta-rate-claim2}. By assumption, $m_i$ is Lipschitz continuous, which implies that $|\Delta m_{it}| = |m_i (\frac{t}{T}) - m_i (\frac{t-1}{T}) | \leq C/T$ for all $t \in \{1, \ldots, T\}$ and some constant $C > 0$. Hence, 
\begin{align*}
\Big| \frac{1}{\sqrt{T}}\sum_{t=2}^T \Delta X_{it,j} \Delta m_{it}\Big| &\leq \frac{1}{\sqrt{T}}\sum_{t=2}^T \big|\Delta X_{it,j} \big| \cdot \big| \Delta m_{it} \big| \\
	& \leq \frac{C}{\sqrt{T}} \cdot \frac{1}{T} \sum_{t=2}^T \left|\Delta X_{it,j} \right| = O_p\Big(\frac{1}{\sqrt{T}}\Big),
\end{align*}
where we have used that $T^{-1} \sum_{t=2}^T |\Delta X_{it,j}| = O_p(1)$ by Markov's inequality. This yields \eqref{lemma-beta-rate-claim3}.
\end{proof}

\begin{lemmaA}\label{lemmaA:lrv}
Let $s_T \asymp T^{1/3}$. Under Assumptions \ref{C-err1}--\ref{C-reg-err}, 
$$\widehat{\sigma}_i^2 = \sigma_i^2 + O_p(T^{-1/3})$$
for each $i$, where $\widehat{\sigma}_i^2$ is the subseries variance estimator of $\sigma_i^2$ introduced in \eqref{eq:lrv}.
\end{lemmaA}

\begin{proof}[\textnormal{\textbf{Proof of Lemma \ref{lemmaA:lrv}}}]
Let $Y_{it}^\circ := m_i(t/T) + \varepsilon_{it}$. Using simple arithmetic calculations, we can rewrite $\widehat{\sigma}_i^2$ as $\widehat{\sigma}_i^2 = \widehat{\sigma}_{i,A}^2 + \widehat{\sigma}_{i,B}^2 - \widehat{\sigma}_{i,C}^2$, where
\begin{align}
\widehat{\sigma}_{i,A}^2 &= \frac{1}{2(M-1)s_T}\sum_{m=1}^M \left[\sum_{t = 1}^{s_T} \left(Y_{i(t + ms_T)}^\circ - Y_{i(t + (m-1)s_T)}^\circ\right)\right]^2 \nonumber\\
\widehat{\sigma}_{i,B}^2 &= \frac{1}{2(M-1)s_T}\sum_{m=1}^M \left[\sum_{t = 1}^{s_T} (\widehat{\bfbeta}_i - \bfbeta_i)^\top\left(\X_{i(t+ms_T)} - \X_{i(t+(m-1)s_T)}\right) \right]^2 \nonumber \\
\widehat{\sigma}_{i,C}^2 &= \frac{1}{(M-1)s_T}\sum_{m=1}^M \bigg[\sum_{t = 1}^{s_T}  \left(Y_{i(t + ms_T)}^\circ - Y_{i(t + (m-1)s_T)}^\circ\right) \nonumber \\ & \phantom{\quad - \frac{1}{(M-1)s_T}\sum_{m=1}^M \bigg[} \times \sum_{t = 1}^{s_T} (\widehat{\bfbeta}_i - \bfbeta_i)^\top\left(\X_{i(t+ms_T)} - \X_{i(t+(m-1)s_T)}\right) \bigg]. \nonumber
\end{align}
By \cite{Carlstein1986} and \cite{WuZhao2007}, we have $\widehat{\sigma}_{i,A}^2 = \sigma_i^2 + O_p(T^{-1/3})$. Moreover, under our assumptions, it is straightforward to see that $\widehat{\sigma}_{i,B}^2 = O_p(T^{-1/3})$ and $\widehat{\sigma}_{i,C}^2 = O_p(T^{-1/3})$.
\end{proof}

\subsection*{Proof of Theorem \ref{theo:stat:global}}\label{subsec-appendix-stat-equality}

We first summarize the main proof strategy, which splits up into five steps, and then fill in the details. We in particular defer the proofs of some intermediate results to the end of the section.

\subsubsection*{Step 1}

To start with, we consider a simplified setting where the parameter vectors $\bfbeta_i$ are known. In this case, we can replace the estimators $\widehat{\bfbeta}_i$ in the definition of the statistic $\widehat{\Phi}_{n,T}$ by the true vectors $\bfbeta_i$ themselves. This leads to the simpler statistic 
\[ \doublehattwo{\Phi}_{n,T} = \max_{1 \le i < j \le n} \max_{(u,h) \in \mathcal{G}_T}\Bigg\{ \bigg| \frac{\doublehat{\phi}_{ij,T}(u,h)} {\{ \doublehattwo{\sigma}_i^2 + \doublehattwo{\sigma}_j^2 \}^{1/2}} \bigg| - \lambda(h)\Bigg\}, \]
where
\[ \doublehat{\phi}_{ij,T}(u,h) = \sum_{t=1}^T w_{t,T}(u,h) \big\{ (\varepsilon_{it} - \bar{\varepsilon}_i) - (\varepsilon_{jt} - \bar{\varepsilon}_j)  \big\} \]
and $\doublehattwo{\sigma}_i^2$ is computed in exactly the same way as $\widehat{\sigma}_i^2$ except that all occurrences of $\widehat{\bfbeta}_i$ are replaced by $\bfbeta_i$. By assumption, $\widehat{\sigma}_i^2 = \sigma^2_i + o_p(\rho_T)$ with $\rho_T = o(1/\log T)$. For most estimators of $\sigma_i^2$ including those discussed in Section \ref{sec:theo}, this assumption immediately implies that $\doublehattwo{\sigma}_i^2 = \sigma^2_i + o_p(\rho_T)$ as well. In the sequel, we thus take for granted that the estimator $\doublehattwo{\sigma}_i^2$ has this property.

We now have a closer look at the statistic $\doublehattwo{\Phi}_{n,T}$. We in particular show that there exists an identically distributed version $\widetilde{\Phi}_{n, T}$ of $\doublehattwo{\Phi}_{n, T}$ which is close to the Gaussian statistic $\Phi_{n, T}$ from \eqref{eq:Phi}. More formally, we prove the following result.  
\begin{propA}\label{propA:strong_approx}
There exist statistics $\{ \widetilde{\Phi}_{n,T}: T =1,2,\ldots \}$ with the following two properties: (i) $\widetilde{\Phi}_{n, T}$ has the same distribution as $\doublehattwo{\Phi}_{n, T}$ for any $T$, and (ii)
\begin{equation*}
\big| \widetilde{\Phi}_{n, T} - \Phi_{n,T} \big| = o_p(\delta_T),
\end{equation*}
where $\delta_T = T^{1/q} / \sqrt{T h_{\min}} + \rho_T \sqrt{\log T}$ and $\Phi_{n,T}$ is a Gaussian statistic as defined in \eqref{eq:Phi}. 
\end{propA}
The proof makes heavy use of strong approximation theory for dependent processes. As it is quite technical, it is postponed to the end of this section.

\subsubsection*{Step 2}

In this step, we establish some properties of the Gaussian statistic $\Phi_{n,T}$. Specifically, we prove the following result.  
\begin{propA}\label{propA:anticon}
It holds that 
\begin{equation*}
\sup_{x \in \reals} \pr \big( | \Phi_{n,T} - x | \le \delta_T \big) = o(1),
\end{equation*}
where $\delta_T = T^{1/q} / \sqrt{T h_{\min}} + \rho_T \sqrt{\log T}$.
\end{propA}
Roughly speaking, this proposition says that the random variable $\Phi_{n,T}$ does not concentrate too strongly in small regions of the form $[x-\delta_T,x+\delta_T]$ with $\delta_T$ converging to $0$. The main technical tool for deriving it are anti-concentration bounds for Gaussian random vectors. The details are provided below.

\subsubsection*{Step 3}

We now use Steps 1 and 2 to prove that 
\begin{equation}\label{eq:claim-step3}
\sup_{x \in \reals} \big| \pr(\doublehattwo{\Phi}_{n, T} \le x) - \pr(\Phi_{n,T} \le x) \big| = o(1). 
\end{equation}
\begin{proof}[\textnormal{\textbf{Proof of (\ref{eq:claim-step3}).}}]
It holds that
\begin{align*}
 & \sup_{x \in \reals} \Big| \pr(\doublehattwo{\Phi}_{n, T} \le x) - \pr(\Phi_{n,T} \le x) \Big| \\
 & = \sup_{x \in \reals} \Big| \pr(\widetilde{\Phi}_{n, T} \le x) - \pr(\Phi_{n,T} \le x) \Big| \\
 & = \sup_{x \in \reals} \Big| \ex \Big[ \ind(\widetilde{\Phi}_{n, T} \le x) - \ind (\Phi_{n,T} \le x) \Big] \Big| \\
 & \le \sup_{x \in \reals} \Big| \ex \Big[ \big\{ \ind(\widetilde{\Phi}_{n, T} \le x) - \ind (\Phi_{n,T} \le x) \big\} \ind \big( |\widetilde{\Phi}_{n, T} - \Phi_{n,T}| \le \delta_T \big) \Big] \Big| \\
 & \quad + \ex \Big[ \ind \big( |\widetilde{\Phi}_{n, T} - \Phi_{n,T}| > \delta_T \big) \Big].
\end{align*} 
Moreover, since  
\[ \ex \Big[ \ind \big( |\widetilde{\Phi}_{n, T} - \Phi_{n,T}| > \delta_T \big) \Big] = \pr  \big( |\widetilde{\Phi}_{n, T} - \Phi_{n,T}| > \delta_T \big) = o(1) \]
by Step 1 and 
\begin{align*}
 & \sup_{x \in \reals} \Big| \ex \Big[ \big\{ \ind(\widetilde{\Phi}_{n, T} \le x) - \ind (\Phi_{n,T} \le x) \big\} \ind \big( |\widetilde{\Phi}_{n, T} - \Phi_{n,T}| \le \delta_T \big) \Big] \Big| \\
 & \le \sup_{x \in \reals} \ex \Big[ \ind \big( |\Phi_{n, T} - x| \le \delta_T, |\widetilde{\Phi}_{n, T} - \Phi_{n,T}| \le \delta_T \big) \Big] \\
 & \le \sup_{x \in \reals} \pr \big( |\Phi_{n, T} - x| \le \delta_T \big) = o(1)
\end{align*}
by Step 2, we arrive at \eqref{eq:claim-step3}.
\end{proof}

\subsubsection*{Step 4}

In this step, we show that the auxiliary statistic $\doublehattwo{\Phi}_{n,T}$ is close to $\widehat{\Phi}_{n,T}$ in the following sense.
\begin{propA}\label{propA:step4}
It holds that 
\[ \doublehattwo{\Phi}_{n,T} - \widehat{\Phi}_{n,T} = o_p(\delta_T) \]
with $\delta_T = T^{1/q}/\sqrt{T h_{\min}} + \rho_T \sqrt{\log T}$.
\end{propA}
The proof can be found at the end of this section.

\subsubsection*{Step 5}

We finally show that 
\begin{equation}\label{eq:claim:step5}
\sup_{x \in \reals} \big| \pr(\widehat{\Phi}_{n, T} \le x) - \pr(\Phi_{n,T} \le x) \big| = o(1).
\end{equation}
\begin{proof}[\textnormal{\textbf{Proof of (\ref{eq:claim:step5}).}}]
To start with, we verify that for any $x \in \reals$ and any $\delta > 0$, 
\begin{align}
\pr\Big( \doublehattwo{\Phi}_{n,T} \le x & - \delta\Big) - \pr \Big(\big|\doublehattwo{\Phi}_{n,T} - \widehat{\Phi}_{n,T}\big| > \delta \Big) \nonumber \\
 & \le \pr\big(\widehat{\Phi}_{n, T} \le x\big) \le\pr\Big(\doublehattwo{\Phi}_{n,T} \le x + \delta\Big) + \pr \Big(\big|\doublehattwo{\Phi}_{n,T} - \widehat{\Phi}_{n,T}\big| > \delta \Big). \label{eq:lemmaA-step1}
\end{align}
It holds that
\begin{align*} 
\pr(\widehat{\Phi}_{n, T} \le x) &= \pr \Big(\widehat{\Phi}_{n, T} \le x, \big|\doublehattwo{\Phi}_{n,T} - \widehat{\Phi}_{n,T}\big| \le \delta \Big) + \pr \Big(\widehat{\Phi}_{n, T} \le x, \big|\doublehattwo{\Phi}_{n,T} - \widehat{\Phi}_{n,T}\big| > \delta \Big) \\
& \le  \pr \Big(\widehat{\Phi}_{n, T} \le x, \widehat{\Phi}_{n,T} - \delta \le \doublehattwo{\Phi}_{n,T} \le \widehat{\Phi}_{n,T} + \delta \Big) + \pr \Big(\big|\doublehattwo{\Phi}_{n,T} - \widehat{\Phi}_{n,T}\big| > \delta \Big) \\
& \le  \pr \Big(\doublehattwo{\Phi}_{n,T} \le x + \delta \Big) + \pr \Big(\big|\doublehattwo{\Phi}_{n,T} - \widehat{\Phi}_{n,T}\big| > \delta \Big)
\end{align*}
and analogously 
\begin{align*} \pr(\doublehattwo{\Phi}_{n, T} \le x - \delta)  \le  \pr \Big(\widehat{\Phi}_{n,T} \le x \Big) + \pr \Big(\big|\doublehattwo{\Phi}_{n,T} - \widehat{\Phi}_{n,T}\big| > \delta \Big).
\end{align*}
Combining these two inequalities, we arrive at \eqref{eq:lemmaA-step1}.

Now let $x\in \reals$ be any point such that $\pr(\widehat{\Phi}_{n,T} \le x) \geq \pr(\Phi_{n, T} \le x)$. With the help of \eqref{eq:lemmaA-step1}, we get that
\begin{align*}
\Big| \pr\big(\widehat{\Phi}_{n, T} \le x\big) - \pr\big(\Phi_{n,T} \le x\big) \Big| 
 & = \pr\big(\widehat{\Phi}_{n, T} \le x\big) - \pr\big(\Phi_{n,T} \le x\big) \\ 
 & \le \pr\Big(\doublehattwo{\Phi}_{n,T} \le x + \delta_{T}\Big) + \pr \Big(\big|\doublehattwo{\Phi}_{n,T} - \widehat{\Phi}_{n,T}\big| > \delta_{T} \Big) \\ 
 & \quad - \pr\big(\Phi_{n,T} \le x\big)  \\
 & = \pr\Big(\doublehattwo{\Phi}_{n,T} \le x + \delta_{T}\Big) - \pr\Big(\Phi_{n,T} \le x + \delta_{T}\Big)  \\
 & \quad +  \pr\big(\Phi_{n,T} \le x + \delta_{T}\big)   - \pr\big(\Phi_{n,T} \le x\big) \\ 
 & \quad + \pr \Big(\big|\doublehattwo{\Phi}_{n,T} - \widehat{\Phi}_{n,T}\big| > \delta_{T} \Big). 
\end{align*}
Analogously, for any point $x\in \reals$ with $\pr(\widehat{\Phi}_{n,T}\le x) < \pr(\Phi_{n, T}\le x)$, it holds that 
\begin{align*}
\Big| \pr\big(\widehat{\Phi}_{n, T} \le x\big) - \pr\big(\Phi_{n,T} \le x\big) \Big| 
 & \le \pr\Big( \Phi_{n,T} \le x - \delta_T \Big) - \pr\Big(\doublehattwo{\Phi}_{n,T} \le x - \delta_{T}\Big) \\
 & \quad + \pr\big( \Phi_{n,T} \le x \big) - \pr\big( \Phi_{n,T} \le x - \delta_T \big) \\
 & \quad + \pr \Big(\big|\doublehattwo{\Phi}_{n,T} - \widehat{\Phi}_{n,T}\big| > \delta_{T} \Big).  
\end{align*}
Consequently,  
\begin{align*}
\sup_{x \in \reals} \Big| \pr\big(\widehat{\Phi}_{n, T} \le x\big) - \pr\big(\Phi_{n,T} \le x\big) \Big| 
 & \le \sup_{x \in \reals} \Big| \pr\big(\doublehattwo{\Phi}_{n, T} \le x\big) - \pr\big(\Phi_{n,T} \le x\big) \Big| \\
 & \quad + \sup_{x \in \reals} \pr\big( | \Phi_{n,T} - x | \le \delta_T \big) \\
 & \quad + \pr \Big(\big|\doublehattwo{\Phi}_{n,T} - \widehat{\Phi}_{n,T}\big| > \delta_{T} \Big). 
\end{align*}
Since the terms on the right-hand side are all $o(1)$ by Steps 2--4, we arrive at \eqref{eq:claim:step5}. 
\end{proof}

\subsubsection*{Details on Steps 1--5}

\begin{proof}[\textnormal{\textbf{Proof of Proposition \ref{propA:strong_approx}}}] 
Consider the stationary process $\mathcal{E}_i = \{\varepsilon_{it}: 1 \leq t \leq T\}$ for some fixed $i \in \{1,\ldots,n\}$. By Theorem 2.1 and Corollary 2.1 in \cite{BerkesLiuWu2014}, the following strong approximation result holds true: On a richer probability space, there exist a standard Brownian motion $\mathbb{B}_i$ and a sequence $\{ \widetilde{\varepsilon}_{it}: t \in \naturals \}$ such that $[\widetilde{\varepsilon}_{i1},\ldots,\widetilde{\varepsilon}_{iT}] \stackrel{\mathcal{D}}{=} [\varepsilon_{i1},\ldots,\varepsilon_{iT}]$ for each $T$ and 
\begin{equation}\label{eq-strongapprox-dep}
\max_{1 \le t \le T} \Big| \sum\limits_{s=1}^t \widetilde{\varepsilon}_{is} - \sigma_i \mathbb{B}_i(t) \Big| = o\big( T^{1/q} \big) \quad \text{a.s.},  
\end{equation}
where $\sigma^2_i = \sum_{k \in \integers} \cov(\varepsilon_{i0}, \varepsilon_{ik})$ denotes the long-run error variance. We apply this result separately for each $i \in \{1,\ldots,n\}$. Since the error processes $\mathcal{E}_i = \{\varepsilon_{it}: 1 \leq t \leq T\}$ are independent across $i$, we can construct the processes $\widetilde{\mathcal{E}}_i = \{\widetilde{\varepsilon}_{it}: t\in \naturals\}$ in such a way that they are independent across $i$ as well.

We now define the statistic $\widetilde{\Phi}_{n,T}$ in the same way as $\doublehattwo{\Phi}_{n, T}$ except that the error processes $\mathcal{E}_i$ are replaced by $\widetilde{\mathcal{E}}_i$. Specifically, we set
\[ \widetilde{\Phi}_{n,T} = \max_{1 \le i < j \le n}\max_{(u,h) \in \mathcal{G}_T} \Bigg\{ \bigg|\frac{\widetilde{\phi}_{ij, T}(u,h)}{\big(\widetilde{\sigma}_i^2 + \widetilde{\sigma}_j^2 \big)^{1/2}} \bigg| - \lambda(h)\Bigg\}, \]
where
\[ \widetilde{\phi}_{ij, T}(u,h) = \sum\limits_{t=1}^T w_{t,T}(u,h) \big\{ (\widetilde{\varepsilon}_{it} - \bar{\widetilde{\varepsilon}}_i)  - (\widetilde{\varepsilon}_{jt} - \bar{\widetilde{\varepsilon}}_j)\big\} \]
and the estimator $\widetilde{\sigma}^2_i$ is constructed from the sample $\widetilde{\mathcal{E}}_i$ in the same way as $\doublehattwo{\sigma}^2_i$ is constructed from $\mathcal{E}_i$. Since $[\widetilde{\varepsilon}_{i1},\ldots,\widetilde{\varepsilon}_{iT}] \stackrel{\mathcal{D}}{=} [\varepsilon_{i1},\ldots,\varepsilon_{iT}]$ and $\doublehattwo{\sigma}_i^2 = \sigma_i^2 + o_p(\rho_T)$, we have that $\widetilde{\sigma}_i^2 = \sigma_i^2 + o_p(\rho_T)$ as well. In addition to $\widetilde{\Phi}_{n,T}$, we introduce the Gaussian statistic
\[ \Phi_{n, T} = \max_{1\leq i< j \leq n}\max_{(u,h) \in \mathcal{G}_T} \bigg\{ \bigg|\frac{\phi_{ij, T}(u,h)}{\big(\sigma_i^2 + \sigma_j^2 \big)^{1/2}}\bigg| - \lambda(h) \bigg\} \]
and the auxiliary statistic 
\[ \Phi_{n, T}^{\diamond} = \max_{1\leq i<j \leq n}\max_{(u,h) \in \mathcal{G}_T} \bigg\{ \bigg|\frac{\phi_{ij, T}(u,h)}{\big(\widetilde{\sigma}_i^2 + \widetilde{\sigma}_j^2 \big)^{1/2}}\bigg| - \lambda(h) \bigg\}, \]
where $\phi_{ij,T}(u,h) = \sum\nolimits_{t=1}^T w_{t,T}(u,h) \{ \sigma_i (Z_{it} - \bar{Z}_i) - \sigma_j (Z_{jt} - \bar{Z}_j) \}$ and the Gaussian variables $Z_{it}$ are chosen as $Z_{it} = \mathbb{B}_i(t) - \mathbb{B}_i(t-1)$. With this notation, we obtain the obvious bound 
\begin{equation*}
\big| \widetilde{\Phi}_{n, T} - \Phi_{n, T} \big| \le \big| \widetilde{\Phi}_{n, T} - \Phi_{n, T}^{\diamond} \big| + \big| \Phi_{n, T}^{\diamond} - \Phi_{n, T} \big|. 
\end{equation*}
In what follows, we prove that 
\begin{align}
\big| \widetilde{\Phi}_{n, T} - \Phi_{n, T}^{\diamond} \big| & = o_p\Big( \frac{T^{1/q}}{\sqrt{Th_{\min}}} \Big) \label{eq-strongapprox-bound-A} \\
\big| \Phi_{n, T}^{\diamond} - \Phi_{n, T} \big| & = o_p(\rho_T \sqrt{\log T}), \label{eq-strongapprox-bound-B}
\end{align}
which completes the proof.

First consider $|\widetilde{\Phi}_{n, T} - \Phi_{n, T}^{\diamond}|$. Straightforward calculations yield that 
\begin{align}
\big| \widetilde{\Phi}_{n, T} - \Phi_{n, T}^{\diamond} \big| 
 & \le \max_{1\le i < j \le n} \big(\widetilde{\sigma}_i^2 + \widetilde{\sigma}_j^2 \big)^{-1/2} \max_{1\le i < j \le n} \max_{(u,h) \in \mathcal{G}_T} \big| \widetilde{\phi}_{ij, T}(u,h) - \phi_{ij, T}(u,h) \big|\Big\}  \nonumber \\ 
 & = O_p(1) \cdot \max_{1\le i < j \le n} \max_{(u,h) \in \mathcal{G}_T} \big| \widetilde{\phi}_{ij, T}(u,h) - \phi_{ij, T}(u,h) \big|, \label{eqA:strong_approx:bound2}
\end{align}
where the last line follows from the fact that $\widetilde{\sigma}_i^2 = \sigma_i^2 + o_p(\rho_T)$. Using summation by parts (that is, $\sum_{t=1}^T a_t b_t = \sum_{t=1}^{T-1} A_t (b_t - b_{t+1}) + A_T b_T$ with $A_t = \sum_{s=1}^t a_s$), we further obtain that 
\begin{align*}
\big| & \widetilde{\phi}_{ij, T}(u,h) - \phi_{ij, T}(u,h) \big|  \\
 & =\bigg|\sum_{t=1}^T w_{t,T}(u,h) \big\{ (\widetilde{\varepsilon}_{it} - \bar{\widetilde{\varepsilon}}_i) - (\widetilde{\varepsilon}_{jt} - \bar{\widetilde{\varepsilon}}_j) -{\sigma}_i (Z_{it} - \bar{Z}_i) + {\sigma}_j (Z_{jt} - \bar{Z}_j) \big\}\bigg|  \\
 & =\Big|\sum_{t=1}^{T-1} A_{ij, t} \big(w_{t,T}(u,h) -w_{t+1,T}(u,h)\big) + A_{ij, T} w_{T,T}(u,h)\Big|,
\end{align*}
where 
\begin{align*}
A_{ij, t} = \sum_{s=1}^t \big\{ (\widetilde{\varepsilon}_{is} - \bar{\widetilde{\varepsilon}}_i)  - (\widetilde{\varepsilon}_{js} - \bar{\widetilde{\varepsilon}}_j) -{\sigma}_i (Z_{is} - \bar{Z}_i) + {\sigma}_j (Z_{js} - \bar{Z}_j) \big\}
\end{align*}
and $A_{ij, T} = 0$ for all pairs $(i, j)$ by construction. From this, it follows that 
\begin{equation}\label{eq-strongapprox-bound3}
\big| \widetilde{\phi}_{ij, T}(u,h) - \phi_{ij, T}(u,h) \big| \le W_T(u, h) \max_{1 \le t \le T} |A_{ij, t}|
\end{equation}
with $W_T(u,h) = \sum_{t=1}^{T-1} |w_{t+1,T}(u,h) - w_{t,T}(u,h)|$. Straightforward calculations yield that
\begin{align*}
\max_{1 \le t \le T} |A_{ij, t}| 
 & \le \max_{1 \le t \le T} \Big| \sum\limits_{s=1}^t \widetilde{\varepsilon}_{is} -{\sigma}_i \sum\limits_{s=1}^t Z_{is} \Big| + \max_{1 \le t \le T} \Big| t (\bar{\widetilde{\varepsilon}}_{i} - {\sigma}_i \bar{Z_i}) \Big|\\
 & \quad + \max_{1 \le t \le T} \Big| \sum\limits_{s=1}^t \widetilde{\varepsilon}_{js} - {\sigma}_j \sum\limits_{s=1}^t Z_{js} \Big| + \max_{1 \le t \le T} \Big| t (\bar{\widetilde{\varepsilon}}_{j} -{\sigma}_j \bar{Z_j}) \Big| \\
 & \le 2 \max_{1 \le t \le T} \Big| \sum\limits_{s=1}^t \widetilde{\varepsilon}_{is} -{\sigma}_i \sum\limits_{s=1}^t Z_{is} \Big| + 2 \max_{1 \le t \le T} \Big| \sum\limits_{s=1}^t \widetilde{\varepsilon}_{js} -{\sigma}_j \sum\limits_{s=1}^t Z_{js} \Big| \\
 & = 2 \max_{1 \le t \le T} \Big| \sum\limits_{s=1}^t \widetilde{\varepsilon}_{is} - {\sigma}_i \sum\limits_{s=1}^t \big(\mathbb{B}_{i}(s) - \mathbb{B}_{i}(s-1) \big) \Big| \\
 & \quad +  2 \max_{1 \le t \le T} \Big| \sum\limits_{s=1}^t \widetilde{\varepsilon}_{js} -{\sigma}_j \sum\limits_{s=1}^t \big(\mathbb{B}_{j}(s) - \mathbb{B}_{j}(s-1) \big) \Big|\\
 & = 2 \max_{1 \le t \le T} \Big| \sum\limits_{s=1}^t \widetilde{\varepsilon}_{is} - {\sigma}_i \mathbb{B}_{i}(t) \Big| + 2 \max_{1 \le t \le T} \Big| \sum\limits_{s=1}^t \widetilde{\varepsilon}_{js} - {\sigma}_j \mathbb{B}_{j}(t) \Big|.
\end{align*}
Applying the strong approximation result \eqref{eq-strongapprox-dep}, we can infer that
\[ \max_{1 \le t \le T} |A_{ij, t}| = o_p\big(T^{1/q}\big). \]
Moreover, standard arguments show that $\max_{(u,h) \in \mathcal{G}_T} W_T(u,h) = O( 1/\sqrt{Th_{\min}} )$. Plugging these two results into \eqref{eq-strongapprox-bound3}, we obtain that 
\[ \max_{1\le i < j \le n} \max_{(u,h) \in \mathcal{G}_T} \big| \widetilde{\phi}_{ij, T}(u,h) - \phi_{ij, T}(u,h) \big| = o_p \Big( \frac{T^{1/q}}{\sqrt{Th_{\min}}} \Big), \]
which in view of \eqref{eqA:strong_approx:bound2} yields that $| \widetilde{\Phi}_{n, T} - \Phi_{n, T}^{\diamond} | = o_p( T^{1/q}/\sqrt{Th_{\min}})$. This completes the proof of \eqref{eq-strongapprox-bound-A}.

\pagebreak
Next consider $|\Phi_{n, T}^{\diamond} - \Phi_{n, T}|$. It holds that
\begin{align}
\big| \Phi_{n, T}^{\diamond} - \Phi_{n, T} \big| 
 & \le \max_{1\leq i< j \leq n}\max_{(u,h) \in \mathcal{G}_T} \Big|\frac{\phi_{ij, T}(u,h)}{\{\widetilde{\sigma}_i^2 + \widetilde{\sigma}_j^2 \}^{1/2}} - \frac{\phi_{ij, T}(u,h)}{\{{\sigma}_i^2 + {\sigma}_j^2 \}^{1/2}}\Big| \nonumber \\
 & \le \max_{1 \le i < j \le n} \left\{ \Big|\big(\widetilde{\sigma}_i^2 + \widetilde{\sigma}_j^2 \big)^{-1/2} - \big(\sigma_i^2 + \sigma_j^2 \big)^{-1/2}\Big| \right\} \max_{1 \le i < j \le n} \max_{(u,h) \in \mathcal{G}_T} \left|\phi_{ij,T}(u,h)\right| \nonumber \\
 & = o_p(\rho_T)  \max_{1 \le i < j \le n} \max_{(u,h) \in \mathcal{G}_T} \left|\phi_{ij,T}(u,h)\right|, \label{eqA:strong_approx:bound5}
\end{align}
where the last line is due to the fact that $\widetilde{\sigma}_i^2 = \sigma_i^2 + o_p(\rho_T)$. We can write $\phi_{ij, T}(u,h) = \phi_{ij, T}^{(I)}(u,h) - \phi_{ij, T}^{(II)}(u,h)$, where
\begin{align*} 
\phi_{ij, T}^{(I)}(u,h) & = \sum\limits_{t=1}^T w_{t,T}(u,h) \, (\sigma_i Z_{it} - \sigma_j Z_{jt}) \sim \normal(0,{\sigma}^2_i + {\sigma}^2_j) \\
\phi_{ij, T}^{(II)}(u,h) & = \sum\limits_{t=1}^T w_{t,T}(u,h) \, ( \sigma_i \bar{Z}_i - \sigma_j \bar{Z}_j) \sim \normal\big(0, (\sigma_i^2 + \sigma_j^2) c_T(u,h)\big) 
\end{align*}
with $c_T(u,h) = \{\sum_{t=1}^T w_{t, T}(u, h)\}^2/T \le C < \infty$ for all $(u,h) \in \mathcal{G}_T$ and $1\le i < j \le n$. This shows that $\phi_{ij, T}(u,h)$ are centred Gaussian random variables with bounded variance for all $(u,h) \in \mathcal{G}_T$ and $1\le i < j \le n$. Hence, standard results on the maximum of Gaussian random variables yield that 
\begin{equation}\label{eq:phi-bound-max-Gaussians}
\max_{1\leq i< j \leq n}\max_{(u,h) \in \mathcal{G}_T} \big|\phi_{ij, T}(u,h)\big| = O_p(\sqrt{\log T}),
\end{equation}
where we have used that $n$ is fixed and $|\mathcal{G}_T| = O(T^\theta)$ for some large but fixed constant $\theta$ by Assumption \ref{C-grid}. Plugging this into \eqref{eqA:strong_approx:bound5} yields 
$| \Phi_{n, T}^{\diamond} - \Phi_{n, T} | = o_p(\rho_T \sqrt{\log T})$, which completes the proof of \eqref{eq-strongapprox-bound-B}.
\end{proof}

\begin{proof}[\textnormal{\textbf{Proof of Proposition \ref{propA:anticon}.}}] 
The proof is an application of anti-concentration bounds for Gaussian random vectors. We in particular make use of the following anti-concentra\-tion inequality from \cite{Nazarov2003}, which can also be found as Lemma A.1 in \cite{Chernozhukov2017}. 
\begin{lemmaA}\label{lemma-Nazarov}
Let $\boldsymbol{Z} = (Z_1,\ldots,Z_p)^\top$ be a centred Gaussian random vector in $\reals^p$ such that $\ex[Z_j^2] \ge b$ for all $1 \le j \le p$ and some constant $b > 0$. Then for every $\boldsymbol{z} \in \reals^p$ and $a > 0$,
\[ \pr(\boldsymbol{Z} \le \boldsymbol{z} + a) - \pr(\boldsymbol{Z} \le \boldsymbol{z}) \le C a \sqrt{\log p}, \]  
where the constant $C$ only depends on $b$. 
\end{lemmaA}
To apply this result, we introduce the following notation: We write $x = (u,h)$ and $\mathcal{G}_T = \{x_1,\ldots,x_p\}$, where $p := |\mathcal{G}_T| \le O(T^\theta)$ for some large but fixed $\theta > 0$ by our assumptions. For $k = 1,\ldots,p$ and $1 \le i < j \le n$, we further let 
\[ Z_{ij, 2k-1} = \frac{\phi_{ij, T}(x_{k1},x_{k2})}{\{{\sigma}_i^2 + {\sigma}_j^2\}^{1/2}} \quad \text{and} \quad Z_{ij, 2k} = -\frac{\phi_{ij, T}(x_{k1},x_{k2})}{\{{\sigma}_i^2 + {\sigma}_j^2\}^{1/2}} \]  
along with $\lambda_{ij,2k-1} = \lambda(x_{k2})$ and $\lambda_{ij,2k} = \lambda(x_{k2})$, where $x_k = (x_{k1},x_{k2})$. Under our assumptions, it holds that $\ex[Z_{ij,l}] = 0$ and $\ex[Z_{ij,l}^2] \ge b > 0$ for all $i$, $j$ and $l$. We next construct the random vector $\boldsymbol{Z} = ( Z_{ij,l} : 1 \le i < j \le n, 1 \le l \le 2p)$ by stacking the variables $Z_{ij, l}$ in a certain order (which can be chosen freely) and construct the vector $\boldsymbol{\lambda} = (\lambda_{ij,l}: 1 \le i < j \le n, 1 \le l \le 2p)$ in an analogous way. Since the variables $Z_{ij,l}$ are normally distributed, $\boldsymbol{Z}$ is a Gaussian random vector of length $(n-1)np$.

With this notation at hand, we can express the probability $\pr(\Phi_{n,T} \le q)$ as follows for each $q \in \reals$: 
\begin{align*} 
\pr(\Phi_{n,T} \le q) 
 & = \pr \Big( \max_{1\leq i< j \leq n}\max_{1 \le l \le 2p} \big\{ Z_{ij,l} - \lambda_{ij,l} \big\} \le q \Big) \\
 & = \pr \big( Z_{ij,l} \le \lambda_{ij,l} + q \text{ for all } (i,j,l) \Big) \\
 & = \pr \big( \boldsymbol{Z} \le \boldsymbol{\lambda} + q \big). 
\end{align*}
Consequently,
\begin{align*}
\pr\big( |\Phi_{n,T} - x| \le \delta_T \big) 
 & = \pr \big( x - \delta_T \le \Phi_{n,T} \le x + \delta_T \big) \\
 & = \pr \big( \Phi_{n,T} \le x + \delta_T \big) - \pr \big( \Phi_{n,T} \le x \big) \\
 & \quad + \pr \big( \Phi_{n,T} \le x \big) - \pr \big( \Phi_{n,T} \le x - \delta_T \big) \\
 & = \pr \big( \boldsymbol{Z} \le \boldsymbol{\lambda} + x + \delta_T \big) - \pr \big( \boldsymbol{Z} \le \boldsymbol{\lambda} + x \big) \\
 & \quad + \pr \big( \boldsymbol{Z} \le \boldsymbol{\lambda} + x \big) - \pr \big( \boldsymbol{Z} \le \boldsymbol{\lambda} + x - \delta_T\big) \\
 & \le 2 C \delta_T \sqrt{\log((n-1)np)},
\end{align*} 
where the last line is by Lemma \ref{lemma-Nazarov}. This immediately implies Proposition \ref{propA:anticon}.
\end{proof}

\begin{proof}[\textnormal{\textbf{Proof of Proposition \ref{propA:step4}}}] 
Straightforward calculations yield that
\begin{align*}
\begin{split}
\big| \doublehattwo{\Phi}_{n, T} - \widehat{\Phi}_{n, T} \big| &\le \max_{1 \le i < j \le n} \max_{(u,h) \in \mathcal{G}_T} \left|\frac{\doublehat{\phi}_{ij,T}(u,h)}{\big(\doublehattwo{\sigma}_i^2 + \doublehattwo{\sigma}_j^2\big)^{1/2}} - \frac{\doublehat{\phi}_{ij,T}(u,h)}{\big(\widehat{\sigma}_i^2 + \widehat{\sigma}_j^2\big)^{1/2}}\right| \\
&\quad+\max_{1 \le i < j \le n} \max_{(u,h) \in \mathcal{G}_T} \left|\frac{\doublehat{\phi}_{ij,T}(u,h)}{\big(\widehat{\sigma}_i^2 + \widehat{\sigma}_j^2\big)^{1/2}} - \frac{\widehat{\phi}_{ij,T}(u,h)} {\big( \widehat{\sigma}_i^2 + \widehat{\sigma}_j^2 \big)^{1/2}} \right|.
\end{split}
\end{align*}
Since $\doublehattwo{\sigma}_i^2 = \sigma_i^2 + o_p(\rho_T)$ and $\widehat{\sigma}_i^2 = \sigma_i^2 + o_p(\rho_T)$, we further get that 
\begin{align*}
\max_{1 \le i < j \le n} &\max_{(u,h) \in \mathcal{G}_T} \left|\frac{\doublehat{\phi}_{ij,T}(u,h)}{\big(\doublehattwo{\sigma}_i^2 + \doublehattwo{\sigma}_j^2\big)^{1/2}} - \frac{\doublehat{\phi}_{ij,T}(u,h)}{\big(\widehat{\sigma}_i^2 + \widehat{\sigma}_j^2\big)^{1/2}}\right|  \\
&\le\max_{1 \le i < j \le n} \left\{ \Big|\big(\doublehattwo{\sigma}_i^2 + \doublehattwo{\sigma}_j^2 \big)^{-1/2} - \big(\widehat{\sigma}_i^2 + \widehat{\sigma}_j^2 \big)^{-1/2}\Big| \right\} \max_{1 \le i < j \le n} \max_{(u,h) \in \mathcal{G}_T} \Big|\doublehat{\phi}_{ij,T}(u,h)\Big| \\
& = o_p(\rho_T) \max_{1 \le i < j \le n} \max_{(u,h) \in \mathcal{G}_T} \Big|\doublehat{\phi}_{ij,T}(u,h)\Big|
\end{align*}
and 
\begin{align*}
\max_{1 \le i < j \le n} &\max_{(u,h) \in \mathcal{G}_T} \left|\frac{\doublehat{\phi}_{ij,T}(u,h)}{\big(\widehat{\sigma}_i^2 + \widehat{\sigma}_j^2\big)^{1/2}} - \frac{\widehat{\phi}_{ij,T}(u,h)} {\big( \widehat{\sigma}_i^2 + \widehat{\sigma}_j^2\big)^{1/2}} \right| \\
 & \le \max_{1\le i < j \le n} \left\{ \big(\widehat{\sigma}_i^2 + \widehat{\sigma}_j^2 \big)^{-1/2} \right\} \max_{1\le i < j \le n} \max_{(u,h) \in \mathcal{G}_T} \Big| \doublehat{\phi}_{ij, T}(u,h) - \widehat{\phi}_{ij, T}(u,h) \Big| \\
 & = O_p(1) \max_{1 \le i < j \le n} \max_{(u,h) \in \mathcal{G}_T} \Big| \doublehat{\phi}_{ij, T}(u,h) - \widehat{\phi}_{ij, T}(u,h) \Big|,
\end{align*}
where the difference of the kernel averages $\doublehat{\phi}_{ij, T}(u,h) - \widehat{\phi}_{ij, T}(u,h) $ does not include the error terms (they cancel out) and can be written as
\begin{align*}
\Big| & \doublehat{\phi}_{ij, T}(u,h) - \widehat{\phi}_{ij, T}(u,h) \Big| \\
 & = \bigg| \sum_{t=1}^T w_{t,T}(u,h) \big\{ (\bfbeta_i - \widehat{\bfbeta}_i)^\top (\X_{it} - \bar{\X}_{i}) - (\bfbeta_j - \widehat{\bfbeta}_j)^\top (\X_{jt} - \bar{\X}_{j}) \big\} \bigg| \\
 & \le \Big|(\bfbeta_i - \widehat{\bfbeta}_i)^\top \sum_{t=1}^T w_{t,T}(u,h) \X_{it} \Big| +  \big|(\bfbeta_i - \widehat{\bfbeta}_i)^\top\bar{\X}_{i}\big| \bigg| \sum_{t=1}^T w_{t,T}(u,h)  \bigg| \\
 & \quad +\Big|(\bfbeta_j - \widehat{\bfbeta}_j)^\top \sum_{t=1}^T w_{t,T}(u,h) \X_{jt}  \Big| + \big|(\bfbeta_j - \widehat{\bfbeta}_j)^\top\bar{\X}_{j}\big| \bigg| \sum_{t=1}^T w_{t,T}(u,h)  \bigg|. 
\end{align*}
Hence,
\begin{equation}\label{ineq-diff-1}
\big| \doublehattwo{\Phi}_{n, T} - \widehat{\Phi}_{n, T} \big| \le o_p( \rho_T) A_{n,T} + O_p(1) \big\{ 2B_{n,T} + 2C_{n,T} \}, 
\end{equation}
where 
\begin{align*}
A_{n,T} & = \max_{1 \le i< j \le n} \max_{(u,h) \in \mathcal{G}_T} \Big|\doublehat{\phi}_{ij,T}(u,h)\Big| \\
B_{n,T} & = \max_{1 \le i \le n} \max_{(u,h) \in \mathcal{G}_T} \Big| (\bfbeta_i - \widehat{\bfbeta}_i)^\top\sum_{t=1}^T w_{t,T}(u,h) \X_{it} \Big| \\
C_{n,T} & = \max_{1 \le i \le n}\big|(\bfbeta_i - \widehat{\bfbeta}_i)^\top\bar{\X}_{i}\big| \max_{(u,h) \in \mathcal{G}_T}  \Big| \sum_{t=1}^T w_{t,T}(u,h)  \Big|. 
\end{align*}
We examine these three terms separately.

We first prove that 
\begin{equation}\label{eq:Ant:1} 
A_{n,T}  = \max_{1 \le i< j \le n} \max_{(u,h) \in \mathcal{G}_T} \Big|\doublehat{\phi}_{ij,T}(u,h)\Big| = O_p\big(\sqrt{\log T}\big).
\end{equation}
From the proof of Proposition \ref{propA:strong_approx}, we know that there exist identically distributed versions $\widetilde{\phi}_{ij, T}(u, h)$ of the statistics $\doublehat{\phi}_{ij,T}(u,h)$ with the property that 
\begin{equation}\label{eq:result-from-propA-strong_approx}
\max_{1\le i < j \le n} \max_{(u,h) \in \mathcal{G}_T} \big| \widetilde{\phi}_{ij, T}(u,h) - \phi_{ij, T}(u,h) \big| = o_p \Big( \frac{T^{1/q}}{\sqrt{Th_{\min}}} \Big). 
\end{equation}
Instead of \eqref{eq:Ant:1}, it thus suffices to show that 
\begin{equation}\label{eq:Ant:2} 
\max_{1 \le i< j \le n}\max_{(u,h) \in \mathcal{G}_T}\Big|\widetilde{\phi}_{ij, T}(u, h)\Big| = O_p\big(\sqrt{\log T}\big).
\end{equation}
Since for any constant $c > 0$, 
\begin{align*}
&\pr\left( \max_{i,j,(u,h)}\left|\phi_{ij, T}(u, h)\right| \leq \frac{c \sqrt{\log T}}{2} \right) \\
&\leq \pr\left( \max_{i,j,(u,h)}\left|\widetilde{\phi}_{ij, T}(u, h)\right| \leq c \sqrt{\log T}\right) \\
& \quad + \pr\left(\left|\max_{i,j,(u,h)}\left|\widetilde{\phi}_{ij, T}(u, h)\right| - \max_{i,j,(u,h)}\left|\phi_{ij, T}(u, h)\right| \right| > \frac{c\sqrt{\log T}}{2}\right) \\
&\leq \pr\left( \max_{i,j,(u,h)}\left|\widetilde{\phi}_{ij, T}(u, h)\right| \leq c \sqrt{\log T}\right) \\
&\quad + \pr\left(\max_{i,j,(u,h)}\left|\widetilde{\phi}_{ij, T}(u, h)- \phi_{ij, T}(u, h)\right| > \frac{c\sqrt{\log T}}{2}\right)
\end{align*}
and $\pr(\max_{i,j,(u,h)}|\widetilde{\phi}_{ij, T}(u, h)- \phi_{ij, T}(u, h)| > c \sqrt{\log T}/2) = o(1)$ by \eqref{eq:result-from-propA-strong_approx}, we get that
\begin{align} 
\pr\bigg(\max_{i,j,(u,h)} & \left|\widetilde{\phi}_{ij, T}(u, h)\right| \leq c\sqrt{\log T} \bigg) \nonumber \\
 & \geq \pr\left( \max_{i,j,(u,h)}\left|\phi_{ij, T}(u, h)\right| \leq \frac{c\sqrt{\log T}}{2}\right) - o(1). \label{eq:Ant:intermediate}
\end{align}
Moreover, since $\max_{i,j,(u,h)}\left|\phi_{ij, T}(u, h)\right| = O_p(\sqrt{\log{T}})$ as already proven in \eqref{eq:phi-bound-max-Gaussians}, we can make the probability $\pr( \max_{i,j,(u,h)}|\phi_{ij, T}(u, h)| \leq c\sqrt{\log T}/2)$ on the right-hand side of \eqref{eq:Ant:intermediate} arbitrarily close to $1$ by choosing the constant $c$ sufficiently large. Hence, for any $\delta > 0$, we can find a constant $c > 0$ such that $\pr(\max_{i,j,(u,h)}|\widetilde{\phi}_{ij, T}(u, h)| \leq c\sqrt{\log T} ) \ge 1 - \delta$ for sufficiently large $T$. This proves \eqref{eq:Ant:2}, which in turn yields \eqref{eq:Ant:1}.

\pagebreak
We next turn to $B_{n,T}$. Without loss of generality, we assume that $\X_{it}$ is real-valued. The vector-valued case can be handled analogously. To start with, we have a closer look at the term $\sum_{t=1}^T w_{t,T}(u,h) \X_{it}$. By construction, the kernel weights $w_{t, T}(u, h)$ are unequal to $0$ if and only if $T(u-h) \le t \le T(u+h)$. We can use this fact to write
\begin{align*}
\Big| \sum_{t=1}^T w_{t,T}(u,h) \X_{it} \Big|  = \bigg| \sum_{t=\lfloor T(u-h) \rfloor}^{\lceil T(u+h) \rceil} w_{t,T}(u,h) \X_{it}   \bigg|.
\end{align*}
Note that
\begin{align}\label{eq:sum_weights}
\begin{split}
\sum_{t=\lfloor T(u-h) \rfloor}^{\lceil T(u+h) \rceil} w^2_{t,T}(u,h) &= \sum_{t=1}^T w^2_{t,T}(u,h) = \sum_{t=1}^T\frac{\Lambda^2_{t,T}(u,h)}{\sum\nolimits_{s=1}^T\Lambda^2_{s,T}(u,h) } = 1.
\end{split}
\end{align}
Denoting by $D_{T, u, h}$ the number of integers between $\lfloor T(u-h) \rfloor$ and $\lceil T(u+h) \rceil$ (with the obvious bounds $2Th \leq D_{T, u, h} \leq 2Th + 2$) and using \eqref{eq:sum_weights}, we can normalize the kernel weights as follows:
\begin{align*}
\sum_{t=\lfloor T(u-h) \rfloor}^{\lceil T(u+h) \rceil} \big(\sqrt{D_{T, u, h}}\cdot w_{t,T}(u,h)\big)^2 = D_{T, u, h}.
\end{align*}
Next, we apply Proposition \ref{theo-wu2016} with the weights $a_t = \sqrt{D_{T, u, h}}\cdot w_{t,T}(u,h)$ to obtain that
\begin{align}
\pr \bigg(\bigg| \sum_{t=\lfloor T(u-h) \rfloor}^{\lceil T(u+h) \rceil} & \sqrt{D_{T, u, h}}\cdot w_{t,T}(u,h) \X_{it}  \bigg| \geq x\bigg) \nonumber \\
 & \leq C_1 \frac{\big( \sum_{t=\lfloor T(u-h) \rfloor}^{\lceil T(u+h) \rceil} |\sqrt{D_{T, u, h}}\cdot w_{t,T}(u,h)|^{q^\prime}\big) \| \X_{i \cdot}\|^{q^\prime}_{q^\prime, \alpha}}{ x^{q^\prime}} \nonumber \\
 & \qquad \qquad \qquad \qquad \qquad \qquad + C_2 \exp \left(-\frac{C_3  x^2}{D_{T, u, h}\| \X_{i\cdot}\|^{2}_{2, \alpha}}\right) \label{ineq-diff-8}
\end{align}
for any $x > 0$, where $\| \X_{i\cdot}\|^{q^\prime}_{q^\prime, \alpha} = \sup_{t\geq 0} (t+1)^{\alpha} \sum_{s=t}^{\infty}\delta_{q^\prime}(\boldsymbol{h}_i, s)$ is the dependence adjusted norm introduced in Definition \ref{defA-DAN} and $\| \X_{i\cdot}\|^{q^\prime}_{q^\prime, \alpha} < \infty$ by Assumption \ref{C-reg3}. From \eqref{ineq-diff-8}, it follows that for any $\delta > 0$, 
\begin{align}
&\pr\left(\max_{(u, h) \in \mathcal{G}_T} \Big| \sum_{t=\lfloor T(u-h) \rfloor}^{\lceil T(u+h) \rceil} w_{t,T}(u,h) \X_{it}  \Big| \geq \delta T^{1/q} \right) \nonumber \\
&\leq \sum_{(u, h) \in \mathcal{G}_T} \pr \left( \Big| \sum_{t=\lfloor T(u-h) \rfloor}^{\lceil T(u+h) \rceil} w_{t,T}(u,h) \X_{it}  \Big| \geq \delta T^{1/q} \right) \nonumber \\
&= \sum_{(u, h) \in \mathcal{G}_T} \pr \left( \Big| \sum_{t=\lfloor T(u-h) \rfloor}^{\lceil T(u+h) \rceil} \sqrt{D_{T, u, h}}\cdot w_{t,T}(u,h) \X_{it}  \Big| \geq \delta\sqrt{D_{T, u, h}}T^{1/q}  \right) \nonumber \\
&\leq \sum_{(u, h) \in \mathcal{G}_T} \left[C_1 \frac{(\sqrt{D_{T, u, h}})^{q^\prime}\big( \sum |w_{t,T}(u,h)|^{q^\prime}\big) \| \X_{i\cdot}\|^{q^\prime}_{q^\prime, \alpha}}{ \big(\delta\sqrt{D_{T, u, h}}T^{1/q}\big)^{q^\prime}} + C_2 \exp \left(-\frac{C_3 \big(\delta\sqrt{D_{T, u, h}}T^{1/q} \big)^2}{D_{T, u, h}\| \X_{i\cdot}\|^{2}_{2, \alpha}}\right) \right] \nonumber \\
&= \sum_{(u, h) \in \mathcal{G}_T} \left[C_1 \frac{\big( \sum |w_{t,T}(u,h)|^{q^\prime}\big) \| \X_{i\cdot}\|^{q^\prime}_{q^\prime, \alpha}}{\delta^{q^\prime}T^{q^\prime/q} } + C_2 \exp \left(-\frac{C_3 \delta^2 T^{2/q} }{\|\X_{i\cdot}\|^{2}_{2, \alpha}}\right) \right] \nonumber \\
&\leq C_1 \frac{ T^\theta \|\X_{i\cdot}\|^{q^\prime}_{q^\prime, \alpha}}{\delta^{q^\prime} T^{q^\prime/q}} \max_{(u, h) \in \mathcal{G}_T} \left( \sum\nolimits_{t=\lfloor T(u-h) \rfloor}^{\lceil T(u+h) \rceil} |w_{t,T}(u,h)|^{q^\prime}\right)+ C_2 T^\theta \exp \left(-\frac{C_3 \delta^2 T^{2/q}}{\|\X_{i\cdot}\|^{2}_{2, \alpha}}\right) \nonumber \\
&= C \frac{ T^{\theta - q^\prime/q}}{\delta^{q^\prime}} + C T^\theta \exp \left(-C T^{2/q} \delta^2\right), \label{eq:exp-bound-BnT}
\end{align}
where the constant $C$ depends neither on $T$ nor on $\delta$. In the last equality of the above display, we have used the following facts:
\begin{enumerate}[label=(\roman*),leftmargin=0.85cm]
\item $\|\X_{i\cdot}\|^{q^\prime}_{q^\prime, \alpha}  < \infty$ by Assumption \ref{C-reg3}.
\item $\|\X_{i\cdot}\|^{2}_{2, \alpha} < \infty$ (which follows from (i)).
\item $\max_{(u, h) \in \mathcal{G}_T} ( \sum_{t=\lfloor T(u-h) \rfloor}^{\lceil T(u+h) \rceil} |w_{t,T}(u,h)|^{q^\prime} ) \leq 1$ for the following reason: By \eqref{eq:sum_weights}, it holds that $\sum_{t=1}^{T} w^2_{t,T}(u,h) = 1$ and thus $0 \leq w^2_{t,T}(u,h) \leq 1$ for all $t$, $T$ and $(u, h)$. This implies that $0 \leq |w_{t,T}(u,h)|^{q^\prime} =  (w^2_{t,T}(u,h))^{q^\prime/2} \leq w^2_{t,T}(u,h) \leq 1$ for all $t$, $T$ and $(u, h)$. As a result, 
\begin{align*}
\max_{(u, h) \in \mathcal{G}_T} \left( \sum_{t=\lfloor T(u-h) \rfloor}^{\lceil T(u+h) \rceil} |w_{t,T}(u,h)|^{q^\prime}\right) \leq
\max_{(u, h) \in \mathcal{G}_T} \left( \sum_{t=\lfloor T(u-h) \rfloor}^{\lceil T(u+h) \rceil} w_{t,T}^2(u,h)\right) =1.
\end{align*}
\end{enumerate}
Since $\theta - q^\prime/q <0$ by Assumption \ref{C-reg1}, the bound in \eqref{eq:exp-bound-BnT} converges to $0$ as $T \to \infty$ for any fixed $\delta >0$. Consequently, we obtain that
\begin{align}\label{ineq-diff-9}
\max_{(u, h) \in \mathcal{G}_T} \bigg| \sum_{t=\lfloor T(u-h) \rfloor}^{\lceil T(u+h) \rceil} w_{t,T}(u,h)\X_{it}  \bigg| = o_p(T^{1/q}).
\end{align}
Using this together with the fact that $\bfbeta_i - \widehat{\bfbeta}_i = O_p(1/\sqrt{T})$ (which is the statement of Lemma \ref{lemma-beta-rate}), we arrive at the bound 
\[ B_{n,T} = \max_{1 \le i \le n} \max_{(u,h) \in \mathcal{G}_T} \Big| (\bfbeta_i - \widehat{\bfbeta}_i)^\top\sum_{t=1}^T w_{t,T}(u,h) \X_{it} \Big| = o_p\Big(\frac{T^{1/q}}{\sqrt{T}}\Big). \]

We finally turn to $C_{n,T}$. Straightforward calculations yield that $| \sum_{t=1}^T w_{t,T}(u,h) | \le C \sqrt{T h_{\max}} = o(\sqrt{T})$. 
Moreover, $\bar{\X}_i = O_p(1/\sqrt{T})$ by Lemma \ref{lemma-wlln} and $\bfbeta_i - \widehat{\bfbeta}_i = O_p(1/\sqrt{T})$ by Lemma \ref{lemma-beta-rate}. This immediately yields that
\[ C_{n,T} = \max_{1\le i  \le n}\big|(\bfbeta_i - \widehat{\bfbeta}_i)^\top\bar{\X}_{i}\big| \max_{(u,h) \in \mathcal{G}_T}  \Big| \sum_{t=1}^T w_{t,T}(u,h)  \Big| = o_p\Big(\frac{1}{\sqrt{T}}\Big). \]

To summarize, we have shown that 
\begin{align*}
\big| \doublehattwo{\Phi}_{n, T} - \widehat{\Phi}_{n, T} \big| 
 & \le o_p( \rho_T) A_{n,T} + O_p(1) \big\{ 2B_{n,T} + 2C_{n,T} \} \\
 & = o_p( \rho_T) O_p(\sqrt{\log T}) + o_p \Big( \frac{T^{1/q}}{\sqrt{T}} \Big) + o_p \Big( \frac{1}{\sqrt{T}} \Big).
\end{align*}
This immediately implies the desired result. 
\end{proof}

\subsection*{Proof of Proposition \ref{prop:test}}

We first show that 
\begin{equation}\label{eq:quant-exact}
\pr(\Phi_{n,T} \le q_{n,T}(\alpha)) = 1 - \alpha. 
\end{equation}
We proceed by contradiction. Suppose that \eqref{eq:quant-exact} does not hold true. Since $\pr(\Phi_{n,T} \le q_{n,T}(\alpha)) \ge 1 - \alpha$ by definition of the quantile $q_{n,T}(\alpha)$, there exists $\xi > 0$ such that $\pr(\Phi_{n,T} \le q_{n,T}(\alpha)) = 1-\alpha + \xi$. From the proof of Proposition \ref{propA:anticon}, we know that for any $\delta > 0$, 
\begin{align*}
\pr \big(\Phi_{n,T} & \le q_{n,T}(\alpha)\big) - \pr\big(\Phi_{n,T} \le q_{n,T}(\alpha) - \delta\big) \\
 & \quad \le \sup_{x \in \reals} \pr \big(|\Phi_{n,T} - x| \le \delta \big)  \le 2 C \delta \sqrt{\log((n-1)np)}. 
\end{align*}
Hence, 
\begin{align*}
\pr \big(\Phi_{n,T} \le q_{n,T}(\alpha) - \delta \big) 
 & \ge \pr\big(\Phi_{n,T} \le q_{n,T}(\alpha) \big) - 2 C \delta \sqrt{\log((n-1)np)} \\
 & = 1-\alpha + \xi - 2 C \delta \sqrt{\log((n-1)np)} > 1-\alpha
\end{align*}
for $\delta > 0$ small enough. This contradicts the definition of the quantile $q_{n,T}(\alpha)$ according to which $q_{n,T}(\alpha) = \inf_{q \in \reals} \{ \pr(\Phi_{n,T} \le q) \ge 1-\alpha \}$. We thus arrive at \eqref{eq:quant-exact}.

Proposition \ref{prop:test} is a simple consequence of Theorem \ref{theo:stat:global} and equation \eqref{eq:quant-exact}. Specifically, we obtain that under $H_0$, 
\begin{align*}
\big| \pr(\widehat{\Psi}_{n, T} \le q_{n,T}(\alpha)) - (1-\alpha) \big| 
 & = \big| \pr(\widehat{\Phi}_{n, T} \le q_{n,T}(\alpha)) - (1-\alpha) \big| \\
 & = \big| \pr(\widehat{\Phi}_{n, T} \le q_{n,T}(\alpha)) - \pr(\Phi_{n,T} \le q_{n,T}(\alpha)) \big| \\
 & \le \sup_{x \in \reals} \big| \pr(\widehat{\Phi}_{n, T} \le x) - \pr(\Phi_{n,T} \le x) \big| = o(1). 
\end{align*}

\subsection*{Proof of Proposition \ref{prop:test:power}}

To start with, note that for some sufficiently large constant $C$ we have
\begin{equation}\label{eqA:power:lambda}
\lambda(h) = \sqrt{2\log\{1/(2h)\}} \le \sqrt{2\log\{1/(2h_{\min})\}} \le C \sqrt{\log T}.
\end{equation}
Write $\widehat{\psi}_{ij, T}(u,h) = \widehat{\psi}_{ij, T}^A(u,h) + \widehat{\psi}_{ij, T}^B(u,h)$ with 
\begin{align*}
\widehat{\psi}^A_{ij,T}(u,h) &= \sum_{t=1}^T w_{t,T}(u,h) \big\{ (\varepsilon_{it} - \bar{\varepsilon}_i) + (\bfbeta_i - \widehat{\bfbeta}_i)^\top (\X_{it} - \bar{\X}_{i}) - \bar{m}_{i, T} \\*[-0.2cm]
& \qquad \qquad \qquad \quad \ - (\varepsilon_{jt} - \bar{\varepsilon}_j) -  (\bfbeta_j - \widehat{\bfbeta}_j)^\top (\X_{jt} - \bar{\X}_{j}) + \bar{m}_{j, T} \big\} \\
\widehat{\psi}_{ij, T}^B(u,h) &= \sum\nolimits_{t=1}^T w_{t,T}(u,h) \bigg(m_{i, T}\Big(\frac{t}{T}\Big) - m_{j, T}\Big(\frac{t}{T}\Big) \bigg),
\end{align*}
where $\bar{m}_{i, T} = T^{-1} \sum_{t=1}^T m_{i, T} (t/T)$. Without loss of generality, consider the following scenario: there exists $(u_0,h_0) \in \mathcal{G}_T$ with $[u_0-h_0,u_0+h_0] \subseteq [0,1]$ such that \begin{align}\label{eqA:power2}
m_{i,T}(w) - m_{j,T}(w) \ge c_T \sqrt{\log T/(Th_0)}
\end{align}
for all $w \in [u_0-h_0,u_0+h_0]$.

We first derive a lower bound on the term $\widehat{\psi}_{ij, T}^B(u_0,h_0)$.  
Since the kernel $K$ is symmetric and $u_0 = t/T$ for some $t$, it holds that $S_{T,1}(u_0,h_0) = 0$ and thus,
\begin{align*} 
w_{t,T}(u_0,h_0) 
&= \frac{K\Big(\frac{\frac{t}{T}-u_0}{h_0}\Big) S_{T, 2}(u_0, h_0)}{\Big\{ \sum_{t=1}^T K^2\Big(\frac{\frac{t}{T}-u_0}{h_0}\Big)S^2_{T, 2}(u_0, h_0) \Big\}^{1/2}} \\
&=\frac{K\Big(\frac{\frac{t}{T}-u_0}{h_0}\Big)}{\Big\{ \sum_{t=1}^T K^2\Big(\frac{\frac{t}{T}-u_0}{h_0}\Big)\Big\}^{1/2}} \ge 0.
\end{align*}
Together with \eqref{eqA:power2}, this implies that 
\begin{equation}\label{eq1-proof-prop-test-power}
\widehat{\psi}_{ij, T}^B(u_0,h_0) \ge c_T \sqrt{\frac{\log T}{Th_0}} \sum\limits_{t=1}^T w_{t,T}(u_0,h_0).
\end{equation}
Using the Lipschitz continuity of the kernel $K$, we can show by straightforward calculations that for any $(u,h) \in \mathcal{G}_T$ and any natural number $\ell$, 
\begin{equation}\label{eq-riemann-sum}
\Big| \frac{1}{Th} \sum\limits_{t=1}^T K\Big(\frac{\frac{t}{T}-u}{h}\Big) \Big(\frac{\frac{t}{T}-u}{h}\Big)^\ell - \int_0^1 \frac{1}{h} K\Big(\frac{w-u}{h}\Big) \Big(\frac{w-u}{h}\Big)^\ell dw \Big| \le \frac{C}{Th}, 
\end{equation}
where the constant $C$ does not depend on $u$, $h$ and $T$. With the help of \eqref{eq-riemann-sum}, we obtain that for any $(u,h) \in \mathcal{G}_T$ with $[u-h,u+h] \subseteq [0,1]$, 
\begin{equation}\label{eq2-proof-prop-test-power}
\Big| \sum\limits_{t=1}^T w_{t,T}(u,h) - \frac{\sqrt{Th}}{\kappa} \Big| \le \frac{C}{\sqrt{Th}}, 
\end{equation}
where $\kappa = (\int K^2(\varphi)d\varphi)^{1/2}$ and the constant $C$ does once again not depend on $u$, $h$ and $T$. From \eqref{eq2-proof-prop-test-power}, it follows that $\sum\nolimits_{t=1}^T w_{t,T}(u,h) \ge \sqrt{Th} / (2\kappa)$ for sufficiently large $T$ and any $(u,h) \in \mathcal{G}_T$ with $[u-h,u+h] \subseteq [0,1]$. This together with \eqref{eq1-proof-prop-test-power} allows us to infer that 
\begin{equation}\label{eqA:power:psiB}
\widehat{\psi}_{ij, T}^B(u_0,h_0) \ge \frac{c_T \sqrt{\log T}}{2 \kappa} 
\end{equation}
for sufficiently large $T$.

We next analyze $\widehat{\psi}^A_{ij,T}(u_0,h_0)$, which can be expressed as $\widehat{\psi}^A_{ij,T}(u_0,h_0) = \widehat\phi_{ij,T}(u, h) + (\bar{m}_{j, T} - \bar{m}_{i, T}) \sum_{t=1}^T w_{t, T}(u, h)$. The proof of Proposition \ref{propA:step4} shows that 
\begin{equation*}
\max_{1 \le i < j \le n} \max_{(u,h) \in \mathcal{G}_T} \Big| \widehat{\phi}_{ij, T}(u,h) \Big| = O_p(\sqrt{\log T}). 
\end{equation*}
Using this together with the bounds $\bar{m}_{i, T} \le C/T$ and $\sum_{t=1}^T w_{t, T}(u, h) \le C \sqrt{T}$, we can infer that 
\begin{align} 
 & \max_{1 \le i < j \le n} \max_{(u,h) \in \mathcal{G}_T} \Big| \widehat\psi_{ij,T}^A(u, h) \Big| \nonumber \\[-0.2cm]
 & = \max_{1 \le i < j \le n} \max_{(u,h) \in \mathcal{G}_T} \Big| \widehat\phi_{ij,T}(u, h) + (\bar{m}_{j, T} - \bar{m}_{i, T}) \sum_{t=1}^T w_{t, T}(u, h) \Big| = O_p(\sqrt{\log T}). \label{eqA:power:psiA}
\end{align}
With the help of \eqref{eqA:power:psiB}, \eqref{eqA:power:psiA}, \eqref{eqA:power:lambda} and the assumption that $\widehat{\sigma}^2_i = \sigma^2_i + o_p(\rho_T)$, we finally arrive at 
\begin{align}
\widehat{\Psi}_{n, T}  
 & \ge \max_{1 \le i < j \le n} \max_{(u,h) \in \mathcal{G}_T} \frac{|\widehat{\psi}_{ij, T}^B(u,h)|}{\{\widehat{\sigma}_i^2 + \widehat{\sigma}_j^2\}^{1/2}} - \max_{1 \le i < j \le n} \max_{(u,h) \in \mathcal{G}_T} \bigg\{ \frac{|\widehat{\psi}_{ij, T}^A(u,h)|}{\{\widehat{\sigma}^2_i + \widehat{\sigma}_j^2\}^{1/2}} + \lambda(h) \bigg\} \nonumber \\
 & = \max_{1 \le i < j \le n} \max_{(u,h) \in \mathcal{G}_T} \frac{|\widehat{\psi}_{ij, T}^B(u,h)|}{\{\widehat{\sigma}_i^2 + \widehat{\sigma}_j^2\}^{1/2}} + O_p(\sqrt{\log T}) \nonumber \\
 & \ge \frac{c_T \sqrt{\log T}}{2 \kappa} \min_{1 \le i < j \le n}\{\widehat{\sigma}_i^2 + \widehat{\sigma}_j^2\}^{-1/2} + O_p(\sqrt{\log T}) \nonumber \\
 & = \frac{c_T \sqrt{\log T}}{2 \kappa} \, O_p(1) + O_p(\sqrt{\log T}). \label{eq5-proof-prop-test-power}
\end{align}
Since $q_{n, T}(\alpha) = O(\sqrt{\log T})$ for any fixed $\alpha \in (0,1)$ and $c_T \to \infty$, \eqref{eq5-proof-prop-test-power} immediately implies that $\pr(\widehat{\Psi}_{n, T} \le q_{n, T}(\alpha)) = o(1)$.

\subsection*{Proof of Proposition \ref{prop:test:fwer}}\label{subsec:app:fwer}

Denote by $\mathcal{M}_0$ the set of quadruples $(i, j, u, h) \in \{1\ldots, n\}^2 \times \grid$ for which $H_0^{[i, j]}(u, h)$ is true. Then we can write the $\text{FWER}$ as
\begin{align*}
\text{FWER}(\alpha)
 & = \pr \Big( \exists (i,j,u, h) \in \mathcal{M}_0: \widehat{\psi}^0_{ij,T}(u, h) > q_{n, T}(\alpha) \Big) \\
 & = \pr \Big( \max_{(i, j, u, h) \in \mathcal{M}_0} \widehat{\psi}^0_{ij,T}(u, h) > q_{n, T}(\alpha) \Big) \\
 & = \pr \Big( \max_{(i,j,u, h) \in \mathcal{M}_0} \widehat{\phi}^0_{ij,T}(u, h) > q_{n, T}(\alpha) \Big) \\
 & \le \pr \Big( \max_{1 \le i < j \le n} \max_{(u, h) \in \grid} \widehat{\phi}_{ij,T}^0(u, h) > q_{n, T}(\alpha) \Big) \\
 & = \pr \big( \widehat{\Phi}_{n, T} > q_{n, T}(\alpha) \big) = \alpha + o(1),
\end{align*}
where the third equality uses that $\hat{\psi}^0_{ijk,T} = \hat{\phi}^0_{ijk,T}$ under $H_0^{[i, j]}(u, h)$.

\subsection*{Proof of Proposition \ref{prop:clustering:1}}

For the sake of brevity, we introduce the following notation. For each $i$ and $j$, we define the statistic $\widehat{\Psi}_{ij,T} : = \max_{(u, h) \in \mathcal{G}_T}\hat{\psi}^0_{ij, T}(u, h)$ which can be interpreted as a distance measure between the two curves $m_i$ and $m_j$ on the whole interval $[0, 1]$. Using this notation, we can rewrite the dissimilarity measure defined in \eqref{dissimilarity} as 
\begin{equation*}
\widehat{\Delta}(S,S^\prime) = \max_{\substack{i \in S, \\ j \in S^\prime}} \widehat{\Psi}_{ij,T}. 
\end{equation*}
Now consider the event  
\[ B_{n,T} = \Big\{ \max_{1 \le \ell \le N} \max_{i,j \in G_\ell} \widehat{\Psi}_{ij,T} \le q_{n,T}(\alpha) \ \text{ and } \ \min_{1 \le \ell < \ell^\prime \le N} \min_{\substack{i \in G_\ell, \\ j \in G_{\ell^\prime}}} \widehat{\Psi}_{ij,T} > q_{n,T}(\alpha) \Big\}. \]
The term $\max_{1 \le \ell \le N} \max_{i,j \in G_\ell} \widehat{\Psi}_{ij,T}$ is the largest multiscale distance between two time series $i$ and $j$ from the same group, whereas $\min_{1 \le \ell < \ell^\prime \le N} \min_{i \in G_\ell, \, j \in G_{\ell^\prime}} \widehat{\Psi}_{ij,T}$ is the smallest multiscale distance between two time series from two different groups. On the event $B_{n,T}$, it obviously holds that 
\begin{equation}\label{eq1-prop-clustering-1}
\max_{1 \le \ell \le N} \max_{i,j \in G_\ell} \widehat{\Psi}_{ij,T} < \min_{1 \le \ell < \ell^\prime \le N} \min_{\substack{i \in G_\ell, \\ j \in G_{\ell^\prime}}} \widehat{\Psi}_{ij,T}. 
\end{equation}
Hence, any two time series from the same class have a smaller distance than any two time series from two different classes. With the help of Proposition \ref{prop:test}, it is easy to see that
\[  \pr \Big( \max_{1 \le \ell \le N} \max_{i,j \in G_\ell} \widehat{\Psi}_{ij,T} \le q_{n,T}(\alpha) \Big) \ge (1 - \alpha) + o(1). \]
Moreover, the same arguments as for Proposition \ref{prop:test:power} show that 
\[  \pr \Big( \min_{1 \le \ell < \ell^\prime \le N} \min_{\substack{i \in G_\ell, \\ j \in G_{\ell^\prime}}} \widehat{\Psi}_{ij,T} \le q_{n,T}(\alpha) \Big) = o(1). \]
Taken together, these two statements imply that 
\begin{equation}\label{eq2-prop-clustering-1}
\pr \big( B_{n,T} \big) \ge (1-\alpha) + o(1). 
\end{equation}
In what follows, we show that on the event $B_{n,T}$, (i) $\{ \widehat{G}_1^{[n-N]},\ldots,\widehat{G}_N^{[n-N]} \big\} = \big\{ G_1,\ldots$ $\ldots,G_N \}$ and (ii) $\widehat{N} = N$. From (i), (ii) and \eqref{eq2-prop-clustering-1}, the statements of Proposition \ref{prop:clustering:1} follow immediately.

\begin{proof}[\textnormal{\textbf{Proof of (i).}}]
Suppose we are on the event $B_{n,T}$. The proof proceeds by induction on the iteration steps $r$ of the HAC algorithm. 
\vspace{7pt}

\textit{Base case} ($r=0$): In the first iteration step, the HAC algorithm merges two singleton clusters $\widehat{G}_i^{[0]} = \{ i \}$ and $\widehat{G}_j^{[0]} = \{ j \}$ with $i$ and $j$ belonging to the same group $G_k$. This is a direct consequence of \eqref{eq1-prop-clustering-1}. The algorithm thus produces a partition $\{ \widehat{G}_1^{[1]},\ldots,\widehat{G}_{n-1}^{[1]} \}$ whose elements $\widehat{G}_\ell^{[1]}$ all have the following property: $\widehat{G}_\ell^{[1]} \subseteq G_k$ for some $k$, that is, each cluster $\widehat{G}_\ell^{[1]}$ contains elements from only one group. 
\vspace{7pt}

\textit{Induction step} ($r \curvearrowright r+1$): Now suppose we are in the $r$-th iteration step for some $r < n-N$. Assume that the partition $\{\widehat{G}_1^{[r]},\ldots,\widehat{G}_{n-r}^{[r]}\}$ is such that for any $\ell$, $\widehat{G}_\ell^{[r]} \subseteq G_k$ for some $k$. Because of \eqref{eq1-prop-clustering-1}, the dissimilarity $\widehat{\Delta}(\widehat{G}_\ell^{[r]},\widehat{G}_{\ell^\prime}^{[r]})$ gets minimal for two clusters $\widehat{G}_\ell^{[r]}$ and $\widehat{G}_{\ell^\prime}^{[r]}$ with the property that $\widehat{G}_\ell^{[r]} \cup \widehat{G}_{\ell^\prime}^{[r]} \subseteq G_k$ for some $k$. Hence, the HAC algorithm produces a partition $\{ \widehat{G}_1^{[r+1]},\ldots,\widehat{G}_{n-(r+1)}^{[r+1]} \}$ whose elements $\widehat{G}_\ell^{[r+1]}$ are all such that $\widehat{G}_\ell^{[r+1]} \subseteq G_k$ for some $k$. 
\vspace{7pt}

The above induction argument shows the following: For any $r \le n - N$, the partition $\{ \widehat{G}_1^{[r]},\ldots,\widehat{G}_{n-r}^{[r]} \}$ consists of clusters $\widehat{G}_\ell^{[r]}$ which all have the property that $\widehat{G}_\ell^{[r]} \subseteq G_k$ for some $k$. This in particular holds for the partition $\{ \widehat{G}_1^{[n-N]},\ldots,\widehat{G}_N^{[n-N]} \}$, which implies that $\{ \widehat{G}_1^{[n-N]},\ldots,\widehat{G}_N^{[n-N]} \} =\{ G_1,\ldots,G_N \}$.  
\end{proof}

\begin{proof}[\textnormal{\textbf{Proof of (ii).}}]
To start with, consider any partition $\{ \widehat{G}_1^{[n-r]},\ldots,\widehat{G}_r^{[n-r]} \}$ with $r < N$ elements. Such a partition must contain at least one element $\widehat{G}_\ell^{[n-r]}$ with the following property: $\widehat{G}_\ell^{[n-r]} \cap G_k \ne \emptyset$ and $\widehat{G}_\ell^{[n-r]} \cap G_{k^\prime} \ne \emptyset$ for some $k \ne k^\prime$. On the event $B_{n,T}$, it obviously holds that $\widehat{\Delta}(S) > q_{n,T}(\alpha)$ for any $S$ with the property that $S \cap G_k \ne \emptyset$ and $S \cap G_{k^\prime} \ne \emptyset$ for some $k \ne k^\prime$. Hence, we can infer that on the event $B_{n,T}$, $\max_{1 \le \ell \le r} \widehat{\Delta} ( \widehat{G}_\ell^{[n-r]} ) > q_{n,T}(\alpha)$ for any $r < N$. 

Next consider the partition $\{ \widehat{G}_1^{[n-r]},\ldots,\widehat{G}_r^{[n-r]} \}$ with $r = N$ and suppose we are on the event $B_{n,T}$. From (i), we already know that $\{ \widehat{G}_1^{[n-N]},\ldots,\widehat{G}_N^{[n-N]} \} =\{ G_1,\ldots,G_N \}$. Moreover, $\widehat{\Delta}(G_\ell) \le q_{n,T}(\alpha)$ for any $\ell$. Hence, we obtain that $\max_{1 \le \ell \le N} \widehat{\Delta} ( \widehat{G}_\ell^{[n-N]} ) = \max_{1 \le \ell \le N} \widehat{\Delta} (G_\ell) \le q_{n,T}(\alpha)$.

Putting everything together, we can conclude that on the event $B_{n,T}$, 
\[ \min \Big\{ r = 1,2,\ldots \Big| \max_{1 \le \ell \le r} \widehat{\Delta} \big( \widehat{G}_\ell^{[n-r]} \big) \le q_{n,T}(\alpha) \Big\} = N, \]
that is, $\widehat{N} = N$. 
\end{proof}

\subsection*{Proof of Proposition \ref{prop:clustering:2}}

We consider the event
\[ D_{n,T} = \Big\{ \widehat{\Phi}_{n,T} \le q_{n,T}(\alpha) \, \text{ and } \,  \min_{1 \le \ell < \ell^\prime \le N} \min_{\substack{i \in G_\ell, \\ j \in G_{\ell^\prime}}} \widehat{\Psi}_{ij,T} > q_{n,T}(\alpha) \Big\}, \]
where we write the statistic $\widehat{\Phi}_{n,T}$ as
\[ \widehat{\Phi}_{n,T} = \max_{1 \le i < j \le n} \max_{(u,h) \in \mathcal{G}_T} \Big\{ \Big| \frac{\widehat{\psi}_{ij,T}(u,h)- \widehat{\psi}_{ij,T}^{\text{trend}}(u,h)} {(\widehat{\sigma}_i^2 + \widehat{\sigma}_j^2)^{1/2}} \Big| - \lambda(h) \Big \} \]
with $\widehat{\psi}_{ij,T}^{\text{trend}}(u,h) = \sum_{t=1}^T w_{t,T}(u,h) \{ (m_{i,T}(t/T) - \bar{m}_{i,T}) - (m_{j,T}(t/T) - \bar{m}_{j,T}) \}$ and $\bar{m}_{i,T} = T^{-1} \sum_{t=1}^T m_{i,T}(t/T)$. The event $D_{n,T}$ can be analysed by the same arguments as those applied to the event $B_{n,T}$ in the proof of Proposition \ref{prop:clustering:1}. In particular, analogous to \eqref{eq2-prop-clustering-1} and statements (i) and (ii) in this proof, we can show that
\begin{equation}\label{eq1-prop-clustering-2}
\pr \big( D_{n,T} \big) \ge (1-\alpha) + o(1)
\end{equation}
and 
\begin{equation}\label{eq2-prop-clustering-2}
D_{n,T} \subseteq \big\{ \widehat{N} = N \text{ and } \widehat{G}_\ell = G_\ell \text{ for all } \ell \big\}.
\end{equation}
Moreover, we have that
\begin{equation}\label{eq3-prop-clustering-2}
D_{n,T} \subseteq \bigcap_{1 \le \ell < \ell^\prime \le \widehat{N}} E_{n,T}^{[\ell,\ell^\prime]}(\alpha),
\end{equation}
which is a consequence of the following observation: For all $i$, $j$ and $(u,h) \in \mathcal{G}_T$ with 
\[ \Big|\frac{\widehat{\psi}_{ij,T}(u,h) - \widehat{\psi}_{ij,T}^{\text{trend}}(u,h)}{(\widehat{\sigma}_i^2 + \widehat{\sigma}_j^2)^{1/2}}\Big| - \lambda(h) \le q_{n,T}(\alpha) \quad \text{and} \quad \Big|\frac{\widehat{\psi}_{ij,T}(u,h)}{(\widehat{\sigma}_i^2 + \widehat{\sigma}_j^2)^{1/2}}\Big| - \lambda(h) > q_{n,T}(\alpha), \]
it holds that $\widehat{\psi}_{ij,T}^{\text{trend}}(u,h) \ne 0$, which in turn implies that $m_i(v) - m_j(v) \ne 0$ for some $v \in I_{u,h}$. From \eqref{eq2-prop-clustering-2} and \eqref{eq3-prop-clustering-2}, we obtain that 
\[ D_{n,T} \subseteq \Big\{ \bigcap_{1 \le \ell < \ell^\prime \le \widehat{N}} E_{n,T}^{[\ell,\ell^\prime]}(\alpha) \Big\} \cap \big\{ \widehat{N} = N \text{ and } \widehat{G}_\ell = G_\ell \text{ for all } \ell \big\} = E_{n,T}(\alpha). \] 
This together with \eqref{eq1-prop-clustering-2} implies that $\pr(E_{n,T}(\alpha)) \ge (1-\alpha) + o(1)$.

\end{document}